\documentclass[11pt]{article}
\usepackage{fullpage}
\usepackage{url}
\usepackage[colorlinks=true,linkcolor=Emerald,citecolor=RoyalBlue]{hyperref}
\usepackage{amsmath,nccmath, amsfonts,amsthm,amssymb,multirow}

\usepackage{times}
\usepackage{paralist}
\usepackage{graphicx}
\usepackage{floatpag}
\usepackage{algorithm}
\usepackage[noend]{algpseudocode}
\usepackage{bbold}
\usepackage{enumitem}
\usepackage{subcaption} 
\usepackage{soul}
\usepackage{comment}
\usepackage{calc}
\usepackage{upgreek}
\usepackage{mathtools}
\usepackage{color}
\usepackage[dvipsnames]{xcolor}
\usepackage{xspace}
\allowdisplaybreaks

\DeclareMathOperator*{\E}{\mathbb{E}}
\DeclareMathOperator*{\pE}{\widetilde{\mathbb{E}}}
\DeclareMathOperator*{\pPr}{\widetilde{\Pr}}
\DeclareMathOperator*{\Ind}{\mathbb{I}}

\DeclareMathOperator*{\Cov}{\mathbb{C}ov}

\DeclarePairedDelimiter\ip{\langle}{\rangle}
\newcommand\skipi{{\vskip 10pt}}

\def \Q {{\mathbb Q}}
\def \R {\mathbb{R}}
\def \Z {\mathbb{Z}}
\def \N {\mathbb{N}}
\def \C {(\{0,1\}^k)^\ell}
\def \T {\widetilde{\E}}

\def \A {\mathcal{A}}
\def \cA {\mathcal{A}}
\def \cC {\mathcal{C}}
\def \cD {\mathcal{D}}
\def \cE {\mathcal{E}}
\def \cF {\mathcal{F}}
\def \cL {\mathcal{L}}

\def \tO {\widetilde{O}}

\def \nn {\mathbb{N}}

\def \sos {SoS\xspace}

\def \sym {\text{sym}}

\def \poly {\mathrm{poly}}

\def \eps {{\varepsilon}}
\def \one {\mathbb{1}}
\def \val {\mathrm{val}}
\def \viol {\mathrm{viol}}

\def \Chi {\raisebox{2pt}{$\chi$}}

\def \pCov {\widetilde{\Cov}}

\renewcommand{\hat}[1]{\widehat{#1}}

\newcommand{\dor}[1]{}
\newcommand{\mitali}[1]{}

\newtheorem{theorem}{Theorem}[section]

\newtheorem{fact}[theorem]{Fact}

\newtheorem{proposition}[theorem]{Proposition}
\newtheorem{corollary}[theorem]{Corollary}
\newtheorem{problem}{Problem}

\newtheorem{lemma}[theorem]{Lemma}
\newtheorem*{lemma*}{Lemma}
\newtheorem*{theorem*}{Theorem}
\newtheorem{claim}[theorem]{Claim}
\newtheorem*{claim*}{Claim}
\newtheorem{conjecture}[theorem]{Conjecture}

\newtheorem{remark}[theorem]{Remark}
\newtheorem{definition}[theorem]{Definition}
\theoremstyle{definition}
\newtheorem{algorithm-thm}[theorem]{Algorithm}

\makeatletter
\let\c@fconjecture\c@conjecture
\makeatother

\makeatletter
\let\c@fconj\c@conj
\makeatother

\title{Solving Unique Games over Globally Hypercontractive Graphs}
\author{	Mitali Bafna \thanks{Carnegie Mellon University. Supported in part by the Computer Science Department, CMU and a gift from CYLAB, CMU.}
\and
	Dor Minzer\thanks{Department of Mathematics, Massachusetts Institute of Technology. Supported by a Sloan Research Fellowship, NSF CCF award 2227876 and NSF CAREER award 2239160.}}

\date{\vspace{-5ex}}

\begin{document}

\maketitle

\begin{abstract}
We study the complexity of affine Unique-Games (UG) over 
\emph{globally hypercontractive graphs}, which are graphs that are not small set expanders but admit a useful and succinct characterization of all small sets that violate the small-set expansion property. This class of graphs includes the Johnson and Grassmann graphs, which have played a pivotal role in recent PCP constructions for UG, and their generalizations via high-dimensional expanders.

Our algorithm shows how to round ``low-entropy'' solutions to sum-of-squares (SoS) semidefinite programs, broadly extending the algorithmic framework of~\cite{BBKSS}. 
We give a new rounding scheme for SoS, which eliminates global correlations in a 
given pseudodistribution so that it retains 
various good properties even after conditioning. 
Getting structural control
over a pseudodistribution after conditioning is a 
fundamental challenge in many SoS based 
algorithms. Due to these challenges,~\cite{BBKSS} 
were not able to establish 
strong algorithms for globally 
hypercontractive graphs, and could only
do so for certifiable small-set expanders. 
Our results improve upon the results 
of~\cite{BBKSS} in various aspects: 
we are able to deal with instances with arbitrarily small (but constant) completeness, and most importantly, 
their algorithm gets a soundness 
guarantee that degrades with other 
 parameters of the graph (which in all PCP constructions grow with the alphabet size), whereas our doesn't.

Our result suggests that UG is easy 
on globally hypercontractive graphs, and therefore highlights the importance of graphs that lack such a characterization in the context of PCP reductions for UG.
\end{abstract}

\section{Introduction}
\subsection{Unique-Games}\label{sec:ug}
The Unique Games Conjecture (UGC in short) 
is a central open problems in Complexity Theory~\cite{Khot02}.
The primary reason for the interest in UGC is that, if true, it
implies a large number of hardness of approximation results 
that are often times tight~\cite{KhotR08,KhotKMO07,Austrin,KNSgroth,Raghavendra08}
(see~\cite{KhotICM,trevisan2012khot}). For example, 
one of the most striking
consequences of UGC is that it 
implies that a class of semi-definite programs
achieves the best possible approximation ratio (among all 
efficient algorithms) for constraint satisfaction problems~\cite{Raghavendra08}.

Research towards a proof of UGC, on the other hand, 
has stalled for the most part up until recently. Unlike problems such 
as $3$-SAT, a randomly sampled instance of Unique-Games (UG) is easy to solve, raising the question of how do hard instances of UG
even look like. To discuss this, we first give 
a formal definition of the Unique-Games problem and the statement of the Unique-Games Conjecture.
\begin{definition}
A instance of Unique-Games $\Psi$ consists of a graph $G = (V,E)$, 
a finite alphabet $\Sigma$ and a collection of constraints, 
$\Phi = \{\Phi_e\}_{e\in E}$, one for each edge in $G$. 
For all $e\in E$, the constraint $\Phi_e$ takes the form 
$\Phi_e = \{(\sigma, \phi_e(\sigma))~|~\sigma\in \Sigma\}$, 
where $\phi_e\colon \Sigma\to\Sigma$ is a $1$-to-$1$ map.

The goal in the Unique-Games problem is to find an assignment 
$A\colon V\to\Sigma$ that satisfies the maximum number of constraints possible, that is, satisfies that $(A(u),A(v))\in \Phi_{e}$ 
for the largest number of edges $e=(u,v)\in E$ as possible. We define
 the value of the instance $\Psi$ by:
\[
{\sf val}(\Psi) = \max_{A\colon V\to\Sigma}\frac{\#\{e~|~\text{A satisfies $e$}\}}{|E|}.
\]
\end{definition}
With this in mind, the Unique-Games Conjecture is the following statement:
\begin{conjecture}
For all $\eps,\delta>0$ there is $k\in\mathbb{N}$ such that 
given a Unique-Games instance $\Psi$ with alphabet size at most $k$,
it is NP-hard to distinguish between:

\textbf{YES case:} ${\sf val}(\Psi)\geq 1-\eps$.

\textbf{NO case:} ${\sf val}(\Psi)\leq \delta$.
\end{conjecture}

It turns out that the topology of the underlying graph $G$ plays a crucial role in the complexity of the UG instance defined over it~\cite{AroraKKSTV08,MMlocalexpander,DBLP:journals/eccc/AroraIMS10}. In particular, it 
turns out that UG over expander graphs is easy:


\begin{definition}
Given a regular graph $G=(V,E)$ and a set 
of vertices $S\subseteq V$, the edge expansion of $S$ is defined by:
\[
\Phi(S) = \Pr_{u\in S, v\in \Gamma(u)}\left[v\not\in S\right].
\]
\end{definition}
A graph $G$ is a called a $(\gamma,\xi)$-small set-expander if for every $S\subseteq V$ of size at most $\xi |V|$  it holds that $\Phi(S)\geq \gamma$. Informally, we say that $G$ is a small set expander if it is a $(\gamma,\xi)$-small set 
expander for $\xi$ that is a small constant, and we say $G$ is an expander 
if $\xi = 1/2$. The results of~\cite{AroraKKSTV08,MMlocalexpander,DBLP:journals/eccc/AroraIMS10,BBKSS}
 assert that UG instances with completeness close to $1$ over small-set expanders are easy. 
Thus, to have any chance of proving UGC, one must use graphs 
which are outside the scope of expanders and small set expanders.

Indeed, recent progress towards UGC~\cite{KhotMS17,DBLP:conf/stoc/DinurKKMS18a,DinurKKMS18,KhotMS18}
has utilized graphs which are not small-set expanders. In these 
works it is proved that $2$-to-$1$-Games are NP-hard (which is a 
very similar problem to UG, except that 
each one of the maps $\phi_e$ defining the constraints 
is a $2$-to-$1$ map). Among others, this implies that for all $\eps > 0$, given
a UG instance $\Psi$ over sufficiently large alphabet, 
it is NP-hard to distinguish between the case that 
${\sf val}(\Psi)\geq 1/2$ and the case that 
${\sf val}(\Psi)\leq \eps$. To prove these results, these works
use graphs that are not small set expanders in two different ways:
\begin{enumerate}
    \item Smooth Parallel Repetition: A key step in the reduction     of~\cite{KhotMS17,DBLP:conf/stoc/DinurKKMS18a,DinurKKMS18,KhotMS18} is an application of the Parallel Repetition Theorem~\cite{Raz} to get a hardness result 
    for a sufficiently 
    smooth outer PCP construction. Roughly speaking, this 
    step in the process may be associated with the Johnson graph
    with a large intersection parameter. That is, with the graph 
    $J(n,\ell,t)$ in which the vertices are $\binom{[n]}{\ell}$, and 
    two vertices $A$ and $B$ are adjacent if $|A\cap B|=\ell-t$, 
    and we think of $t$ as much smaller than $\ell$ (say, $t=\sqrt{\ell}$).
    \item Composition with the Grassmann encoding: The Grassmann
    encoding is an encoding of linear functions based on the Grassmann graph 
    ${\sf Grass}(n,\ell)$ over $\mathbb{F}_2$. The Grassmann graph 
    over $\mathbb{F}_2$ is the graph whose vertices are all 
    $\ell$-dimensional subspaces of $\mathbb{F}_2^n$, denoted by 
    ${n\brack \ell}$, and two vertices $L$ and $L'$ are adjacent if 
    ${\sf dim}(L\cap L') = \ell-1$. 
\end{enumerate}
Both of the graphs above, namely the Johnson graph with large intersection sizes, as well as the Grassmann graph, are not 
small set expanders. However very importantly, the class
of small sets in the Grassmann graph with bad expansion has a succinct and intuitive characterization, 
and the proof of the $2$-to-$1$ Games Theorem heavily relies on this characterization. 


Though the term is not formally defined, we refer to graphs such 
as the Grassmann graph above as globally hypercontractive graphs. 
By that, we mean that there is a collection of 
``obviously-non-expanding local sets'', such that any small set that 
doesn't expand well must have a large intersection with one of the sets from 
the collection (see Section~\ref{sec:abstraction} for a semi-formal definition). Aside from the Grassmann graph, this 
class of graphs includes Johnson graphs with small intersection sizes~\cite{KMMS},
certain Cayley graphs over 
the symmetric group~\cite{filmus2020hypercontractivity},
$p$-biased cubes for $p=o(1)$, other product domains~\cite{keevash2021global} as well as high dimensional expanders~\cite{gur2022hypercontractivity,bafna2022hypercontractivity}.

A natural question to consider then refers to the 
complexity of UG over graphs that are
globally hypercontractive, which is the main question 
of study in this paper. The motivation for this 
is two-fold:
\begin{enumerate}
    \item First, in light of 
the reduction of~\cite{KhotMS17,DBLP:conf/stoc/DinurKKMS18a,DinurKKMS18,KhotMS18} and the above discussion, it is an interesting 
question whether there may be a reduction showing the NP-hardness of UG that only uses globally hypercontractive graphs. Gaining a better understanding for these questions will shed further
light on the source of hardness in this reduction, as well as explain the distinct roles of the ``smooth parallel repetition'' and ``composition with the Grassmann encoding'' steps above.

\item Second, the UGC is related to another well-known computational problem known as the Small-set Expansion (SSE) problem~\cite{RaghavendraS10,RaghavendraST12}. In the SSE problem one is given a graph $G = (V,E)$ such that either (a) for all $S\subseteq V$ of size at most $\gamma|V|$ it holds that $\Phi(S)\geq 1-\eps$, or (b) there is $S\subseteq V$ of size at most $\gamma|V|$ such that $\Phi(S)\leq \delta$. The
Small-set Expansion hypothesis (SSEH) asserts that for all $\eps,\delta>0$
this problem is NP-hard for sufficiently small $\gamma$, and the 
papers~\cite{RaghavendraS10,RaghavendraST12} show that this 
hypothesis implies the UGC.\footnote{To be more precise, it implies
a stronger version of the UGC, which states that UG 
remains hard on graphs that have an expansion profile
as the noisy hypercube.} 

Despite recent progress towards UGC, no new results for the SSEH were proved. In some sense, it seems that to prove hardness results for SSE, one would need 
to only use gadgets that are small-set expanders or perhaps globally hypercontractive. Otherwise, in the final graph one may create small non-expanding sets that are unintended (namely, that don't correspond to any solution to the 
initial problem we reduced from). This motivates the study of the role of globally hypercontractive graphs in reductions related to Unique-Games.
\end{enumerate}

The question has been addressed recently in~\cite{BBKSS} (with 
follow-up by~\cite{DBLP:conf/soda/BafnaHKL22}) for the class of \emph{Affine 
Unique-Games}, which we define and discuss in the next section.

\subsection{Unique-Games over Globally Hypercontractive Graphs}
\begin{definition}
An instance of $\text{Affine-UG}$ is an instance of Unique-Games 
in which the alphabet is the ring of integers modulo $q$, $\mathbb{Z}_q$, and all of the 
constraint maps $\phi_e$ are affine shifts, that is, $\phi_e$ 
of the form $\phi_e(\sigma) = \sigma+b_e$ for some $b_e\in\mathbb{Z}_q$.
\end{definition}
An equivalent but slightly different way to view 
the Affine-UG problem is as a system
of linear equations $(X,E)$ over $\mathbb{Z}_q$.  
Each equation in $E$ is of the form $x_i - x_j = b$ where $x_i,x_j\in X$ are variables 
and $b\in\mathbb{Z}_q$ is some constant. Despite looking very restrictive, 
it is known~\cite{KhotKMO07} that the UGC is true if and 
only if it holds for the class of Affine UG and furthermore this class captures many interesting optimization problems such as Max-Cut and graph coloring, thus we shall focus our attention on Affine UG henceforth. 
\footnote{We remark that the 
reduction of~\cite{KhotKMO07} does not preserve the topology of the graph. 
We are therefore not able to translate our results directly to the 
class of general UG, and believe this is an interesting direction 
for further study.}

In~\cite{BBKSS}, the authors investigate the complexity of Affine UG 
on the Johnson graph, albeit with small intersection sizes, the regime in which a characterization theorem for non-expanding sets holds. 
Namely, they consider
the ``$\alpha$-noisy'' Johnson graph $J(n,\ell,t)$ in the case that $t = \alpha \ell$, for $\alpha\in (0,1)$ bounded away from $0$ and thought of as a fixed constant independent of $\ell$.
The result of~\cite{BBKSS} 
asserts that for small enough $\eps>0$, there is a polynomial time algorithm
that given an instance of UG over $J(n,\ell,t)$ with value at least 
$1-\eps$, finds an assignment that satisfies at least $C(\ell,\alpha,\eps)>0$
of the constraints. The most important feature of this result is the lack
of dependency on the size of the graph $n$ and the alphabet size, however 
it has two significant downsides:
\begin{enumerate}
    \item The completeness guarantee needed for the algorithm to work is close to $1$ ($\eps$ close to $0$). Thus, within the realm of their result, it is still possible that UG over the noisy-Johnson and Grassmann graphs is hard for completeness smaller than $1/2$ (wherein 
    general UG is now known to be hard) but suddenly becomes easy when the completeness exceeds $1/2$. In fact, UG with completeness $> 1/2$ is easy for the Grassmann graph via the UG algorithm for expanders itself, thus their result does not imply anything non-trivial about the complexity of UG on the Grassmann graph, despite its being globally hypercontractive. Therefore, it is not
    clear what their algorithm implies about the hard instances of~\cite{DinurKKMS18} which have completeness less than $1/2$: (1) are they hard due to the smooth parallel repetition step, or the composition with the Grassmann graph? (2) Is the smoothness of the outer PCP necessary -- instead of the Johnson graph with large intersection sizes could one have instead used the noisy-Johnson graph?    
    \item Their soundness guarantee depends on $\ell$, the uniformity of the sets in the Johnson graph, which in most PCP constructions grows with the alphabet size of the instance. Indeed, such graphs are often used to encode a global function $f\colon [n]\to\Gamma$ (often from an error correcting code) where $|\Gamma| = O(1)$. 
    In that case, the Johnson graph corresponds to the restrictions encoding
    $F[A] = f|_{A}$, in which case each vertex $A$ of the Johnson graph is 
    assigned a label from $\Gamma^{\ell}$ corresponding to the restriction of 
    $f$ to $A$. Thus, in the hard instances of UG obtained via PCP constructions (in particular the one above) the parameter $\ell$ 
    grows with the alphabet size of the UG instance, hence we would like the performance of our algorithm to not depend on it. 
    
    We remark though, that it is okay for the runtime of our algorithm to depend on it (just like an algorithm 
    running in time $n^{k}$ that solves UG 
    would refute the conjecture, but the trivial random algorithm that satisfies at least a $1/k$ fraction of the constraints does not).
\end{enumerate}

\subsection{Our Results}
Our main result asserts that there is a polynomial time algorithm for solving 
Affine UG over globally hypercontractive graphs that bypasses these 
two barriers thus addressing the motivations presented in Section~\ref{sec:ug}. As the term globally hypercontractive graph is not formally defined, 
below are some concrete instances of graphs on which this applies. In Section~\ref{sec:abstraction} we give a semi-formal definition of globally hypercontractive graphs and also show how our algorithm and analysis can be abstracted to solve UG on such graphs, as long as one is provided with an SoS certificate of global hypercontractivity.

We first consider the noisy-Johnson graph, 
for which we have the following two results.
The first result addresses the case that 
the completeness of the instance is close to 
$1$, in which case our algorithm matches the 
guarantee of the algorithm of~\cite{BBKSS} 
for certifiable small-set expanders, and in particular the $\alpha$-noisy hypercube graph:
\begin{theorem}\label{thm:main_johnson_close_to_1}
There is $\eps_0>0$
such that for all $\alpha\in (0,1)$ 
the following holds for all $0<\eps\leq \eps_0$.
There exists an algorithm whose running time is $n^{\poly(\ell,|\Sigma|,1/\eps)}$ which, on input
$\Psi$ which is an affine UG instance over $J(n,\ell,\alpha\ell)$ promised to be at least $(1-\eps)$-satisfiable, finds an assignment that satisfies 
at least $2^{-O\left(\frac{\sqrt{\eps}}{\alpha}\right)}$-fraction of the constraints in $\Psi$.
\end{theorem}

The second result addresses the case of UG instances with arbitrarily small (but bounded away from $0$) completeness, in which case our algorithm satisfies a constant fraction of the constraints:
\begin{theorem}\label{thm:main_johnson}
For all $\alpha\in (0,1)$ and $c>0$, there is $\delta>0$ such that the following holds. 
There exists an algorithm whose running time is $n^{D}$ with $D = \ell^{\poly(|\Sigma|\ell^{1/c})}$\footnote{We note that we have not optimized for $D$ and the $\exp(\ell)$-dependence arises due to the degree of the SoS proofs in Section~\ref{sec:johnson_2}. We used a blackbox statement to convert some of the proofs therein to SoS proofs, and we conjecture that one can in fact improve the SoS degree to $O(\ell)$ when done carefully.} which on input
$\Psi$, an affine UG instance over $J(n,\ell,\alpha \ell)$ promised to be at least $c$-satisfiable, finds an assignment that satisfies 
at least $\delta$-fraction of the constraints in $\Psi$.
\end{theorem}

We note that the soundness guarantee in the theorems above does not depend on $\ell$ when $\alpha = \Omega(1)$, but degrades as $\alpha$ approaches $0$. In Section~\ref{sec:interpretation} below we discuss why one cannot get a polytime algorithm that obtains an $\Omega(1)$-valued solution when $\alpha = o(1)$ (e.g. $1/\sqrt{\ell}$) therefore showing that this dependence on $\alpha$ is necessary.

We can get similar results given any of the globally hypercontractive graphs mentioned earlier. Below
we give a corollary for the Grassmann graph.
We show that there is a polynomial time algorithm solving affine UG over the Grassmann graph, even on 
instances with small completeness:
\begin{theorem}\label{thm:main_grass}
For all $c>0$ there exists $\delta>0$ such that the following holds. There exists an algorithm whose running time is $n^{D}$ with $D = \ell^{\poly(|\Sigma|\ell^{1/c})}$ which
on input $\Psi$, an affine UG instance over 
${\sf Grass}(n,\ell)$ promised to be at least $c$-satisfiable, finds an assignment that satisfies 
at least $\delta$-fraction of the constraints in $\Psi$.
\end{theorem}

As discussed before, since the spectral gap of the Grassmann graph is $1/2$, UG algorithms over expanders already imply Theorem~\ref{thm:main_grass} for $c \gg 1/2$. Thus, the main contribution of Theorem~\ref{thm:main_grass} is the algorithm on Grassmann graphs that works for arbitrarily small completeness.

Below we state our result for random walks on high dimensional expanders (HDX), a large class of graphs that generalize the Johnson graphs but do not necessarily possess its strong symmetries. These  include graphs stemming from cut-offs of \cite{lubotzky2005explicit}'s construction of Ramanujan complexes, or  \cite{kaufman2018construction}'s construction of coset complex expanders. These graphs exhibit the nice high-dimensional expansion properties (e.g. global hypercontractivity) of the Johnson graphs yet are substantially different in other aspects, such as being of bounded degree.

\begin{theorem}\label{thm:main_hdx}
For all $\alpha \in (0,1)$ and $c>0$, there exists $\delta > 0$ such that the following holds. Let $X$ be any $d$-dimensional two-sided $\gamma$-local-spectral expander with $\gamma \ll o_\ell(1)$ and $d>\ell$. There exists an algorithm whose running time is $n^{D}$ with $D = \ell^{\poly(|\Sigma|\ell^{1/c})}$ which on input $\Psi$, an affine UG instance over the canonical walk $M$ on $X(\ell)$ of depth $\alpha$, promised to be at least $c$-satisfiable, finds an assignment that satisfies 
at least $\delta$-fraction of the constraints in $\Psi$.
\end{theorem}

Since we have not defined any of the HDX terminology, let us note that this is indeed a generalization of Theorem~\ref{thm:main_johnson}. The Johnson graph corresponds to the complete complex $X$ (which is the simplest instantiation of a two-sided local spectral expander), and the $\alpha$-noisy Johnson graph $J(n,\ell,\alpha\ell)$ corresponds to a ``canonical'' random-walk on $X(\ell)$ that goes down $\alpha\ell$-levels and comes back up randomly to $X(\ell)$ while ensuring that it changes exactly $\alpha\ell$ elements in a vertex. In fact, in the above theorem we can allow $M$ to be any complete random walk on $X(\ell)$ and our soundness guarantee will only depend on $c$ and certain parameters of $M$ that are inherently independent of $\ell$ \footnote{Concretely it depends on the stripped threshold rank of $M$ above a certain threshold as defined in~\cite{DBLP:conf/soda/BafnaHKL22}. For example, when $M$ is the canonical random walk with depth $\alpha$ on $X(\ell)$, and the completeness is $c=1-\eps$, this quantity is $r(M) = O(\sqrt{\eps}/\alpha)$ and our soundness guarantee is $\exp(-r)$, matching that of Theorem~\ref{thm:main_johnson_close_to_1}.}.

For the sake 
of simplicity of presentation, in the first part of the paper we focus
on the Johnson graph. In Section~\ref{sec:other_graphs} we provide an abstraction of our techniques and discuss the algorithm for general globally hypercontractive graphs, in particular for the Grassmann graph and random-walks over HDXs.


\subsubsection*{An Interpretation of Our Results}\label{sec:interpretation}
Our results suggest that the hardness in 
the instances of UG obtained via the reduction  of~\cite{KhotMS17,DBLP:conf/stoc/DinurKKMS18a,DinurKKMS18,KhotMS18} does not come from
the Grassmann graph (which is globally 
hypercontractive), but rather from the smooth parallel repetition step. Recall that this step uses a Johnson graph with a large intersection
parameter ($J(n,\ell,\alpha\ell)$ with $\alpha \approx 0$), that is not globally-hypercontractive. Therefore combining the knowledge from the reduction and our algorithm we get that the $\alpha$-noisy-Johnson graphs are hard for UG when $\alpha = o(1)$ and become easy when $\alpha$ is bounded away from $0$, thus also explaining why our soundness guarantee must necessarily decay with $\alpha$ (under $P \neq NP$). Indeed, we would be able to make such an assertion provided that our
results held for general UG (as opposed to only affine UG) or if the reduction above produced  instances of Affine UG. Though
we believe an algorithm for general UG should exist along the lines of our algorithm, we 
do not know how to prove so and leave this is 
an interesting direction to investigate. 

Albeit, ignoring the subtlety between general and affine UG, this means that any future progress on UGC will have to use graphs that are \emph{not} 
globally hypercontractive, possibly again 
via the technique of \emph{smooth} parallel repetition. 

\subsection{Techniques: New rounding scheme for higher degree SoS}
Our algorithms are obtained via a novel rounding scheme and analysis for the standard higher degree Sum-of-Squares SDP relaxation for Unique Games. Raghavendra's~\cite{Raghavendra08} groundbreaking result showing the optimality of the basic SDP for all CSPs under the UGC, led to efforts to refute the UGC using higher degree SoS relaxations~\cite{Lasserre00,Parrilo00}. The study of SoS algorithms has since produced numerous algorithmic advances across many fronts: high-dimensional robust statistics~\cite{BarakKS14,ma2016polynomial, BakshiDHKKK20, BakshiDJKKV22}, quantum computation~\cite{BKS17} and algorithms for semi-random models~\cite{BuhaiKS22}, to name a few. Most of these works use the sum-of-squares method for \emph{average-case} problems though and unfortunately there remains a dearth of techniques for analysing higher degree SoS relaxations for \emph{worst-case} optimization problems. The handful of techniques known for worst-case rounding are the \emph{global correlation rounding} technique from~\cite{BarakRS11,RT12} and its generalization via reweightings in~\cite{BKS17}.

In this context, very recently~\cite{BBKSS} proposed a new technique for rounding relaxations of UG that have ``low-entropy'' measured via a function called the shift-partition size. Given two fixed assignments for the instance, their shift-partition size is roughly defined as the fraction of variables on which these assignments agree (upto symmetry).
Taking the equivalent view of the SDP solution as a distribution $\cD$ over non-integral solutions, called a pseudodistribution, the expected shift-partition size of two random assignments drawn from $\cD$ is then roughly equal to an average of local collision probabilities under $\cD$ and thus a proxy for the \emph{entropy} of $\cD$. Their analysis proceeds by showing: (1) when the expected shift-partition size (equivalently collision probability) is large, one can round to a high-valued solution, and moreover (2) when the graph is a certifiable small-set expander, the pseudodistribution always has large shift-partition size! They were not able to extend this idea to get high-valued solutions for the broader class of globally hypercontractive graphs though, since in this case the pseudodistribution might be supported over multiple assignments and therefore does not have high collision probability. It turns out though that even in this harder case, the pseudodistribution $\cD \times \cD$ has large expected shift-partition size after \emph{conditioning} on an event $E$. But they could not exploit this property since after conditioning the shift-partition could be large for trivial reasons\footnote{In the worst case, the event $E$ could collapse the product distribution over two random assignments to set the second random assignment to be always equal to the first one. In this case a pair of assignments drawn from $\cD \times \cD ~|~ E$ being equal does not say anything about the collision probability of $\cD$.} and therefore is no longer a good proxy for the collision probability/entropy of the distribution $\cD$.

Our main technical contribution is to strengthen and broadly extend this framework of rounding low-entropy pseudodistributions. We show that after a suitable preprocessing step on the pseudodistribution, one can in fact condition on any event $E$ (with not too small probability) while preserving most of the desired local independence properties of the distribution. Thus, even after conditioning on $E$, the expected shift-partition size of $\cD\times\cD ~|~ E$ being large signifies that the pseudodistribution $\cD$ has high collision probability. One can then use a simple rounding procedure to obtain a high-valued UG solution. Conditioning pseudodistributions is one of the few ways we know of harnessing the power of higher-degree pseudodistributions, hence we believe that the idea of gaining structural control over the distribution after conditioning may be applicable in the analysis of other SoS algorithms too.

At a high level, the algorithm of~\cite{BBKSS} can be viewed as showing one how to round a pseudodistribution when it is essentially supported over \emph{one} assignment (upto symmetry of the solution space). In general though, our problem might not be as structured and the resulting pseudodistribution for it might not be low-entropy in this restricted sense\footnote{For example, it could be supported over $O(1)$ solutions, in which case the collision probability would be high only after conditioning on an appropriate event.}. We expand this notion of ``low-entropy'' thus allowing a broader set of pseudodistributions, intuitively those with a ``few good solutions'', to fall into this class and therefore become amenable to rounding. We expect that with this strengthening, this framework  should be broadly applicable for algorithm design for other optimization problems. Below is a detailed overview of our techniques, starting out with the framework of~\cite{BBKSS}.

\subsubsection{The Approach of~\cite{BBKSS}: Rounding analysis via the Shift Partition}
Fix an Affine Unique-Games instance $\Psi = (G=(V,E), \mathbb{F}_q, \Phi)$. 
In the SoS relaxation of the Unique-Games problem we have a collection 
of variables $X_{v,\sigma}$, one for pair of vertex $v\in V$ and label 
to it $\sigma\in \Sigma$. The output of the program is a pseudoexpectation 
operator $\widetilde{\E}$, which assigns to each monomial involving 
at most $d$ of the variables a real-number, under which:
\begin{enumerate}
    \item The value is high:
    \[
    \widetilde{\E}\left[\sum\limits_{(u,v)\in E}\sum\limits_{\sigma\in \Sigma} X_{v,\sigma} X_{u,\phi_{u,v}(\sigma)}\right]\geq c\cdot|E|.
    \]
    \item $\widetilde{\E}$ is a linear, positive semi-definite operator (when viewed as a matrix over $\mathbb{R}^{M\times M}$ where $M$ is the set of monomials of degree at most $d/2$) satisfying various 
    Booleanity constraints on $X_{u,\sigma}$.
    \item Scaling: $\widetilde{\E}[1] = 1$.
\end{enumerate}
Morally, the pseudoexpectation 
$\widetilde{\E}$ should be thought of 
in the following way: 
there is an unknown distribution $\mathcal{D}$ 
over assignments $A_1,\ldots,A_m$ that each have value 
at least $c$. For the 
assignment $A_i$ we think of Boolean valued assignment to 
the variables $X_{u,\sigma}$ that assigns to a variable $1$
if and only if $A_i(u) = \sigma$, and associate with it 
the expectation operator $\E_i$ which maps monomials to 
Boolean values in the natural way according to $A_i$. The operator 
$\widetilde{\E}$ then is the average of the operators $\E_i$ 
according to $i\sim \mathcal{D}$. 
\footnote{Formally speaking,
when given $\widetilde{\E}$ we are not guaranteed that there 
exists an actual distribution $\mathcal{D}$ over good assignments 
as above, however this intuition will be good enough 
for the sake of this informal presentation.}

\paragraph{Shift-partition:}
Given $\widetilde{\E}$, one can construct a different pseudoexpectation operator that allows access 
to moments of two assignments 
$X = A_i, X' = A_j$ where $i,j\sim \mathcal{D}$ are 
chosen independently. 
In expectation, we get that at least $c^2$ fraction of the edges
get satisfies by both $X$ and $X'$; the algorithm attempts 
to satisfy these edges. Towards this end, given two fixed assignments $X$ and $X'$ we define the shift-partition of the vertices of $V$: $V = \cup_{s\in\mathbb{F}_q} F_s$ where for each $s\in\mathbb{F}_q$ we define
\[
F_s(X,X') = \left\{v\in V~|~X(v) - X'(v) = s\right\}.
\]
The shift-partition size is then defined as: 
\[\pE_{X,X' \sim \cD}\left[\sum_{s \in \Sigma} \left(\frac{|F_s(X,X')|}{|V(G)|}\right)^2\right].\]  
After rearranging, we get that when $X$ and $X'$ are independent, this expression is an average of some local collision probabilities (precisely $\E_{u,v}[CP(X_u - X_v)]$), and hence the shift-partition size being large in expectation turns out to be useful for rounding.

On the other hand, observe that if an edge $(u,v)\in E$ is satisfied by both $X$ and 
$X'$, then $X(u) - X(v) = X'(u) - X'(v)$ and rearranging we 
conclude that $u$ and $v$ are in the same part $F_s$ 
of the shift partition. We therefore conclude that in expectation over $X,X' \sim \cD$ at least $c^2$ fraction of the edges of $G$ stay
inside the same part of the shift partition, implying that the expansion of the shift-partition is small. 

\paragraph{Small-set expanders.} If the graph 
$G$ is a small-set expander, then the above implies that at
least one of the sets $F_s$/the shift-partition size is large and the following 
rounding procedure works in such cases:
\begin{enumerate}
    \item Sample a vertex $v\in V$ and choose $A(v) = \sigma$ according to the distribution $p(\sigma) = \widetilde{\E}[X_{v,\sigma}]$.
    \item For any $u\in V$, sample $A(u)$ according to the distribution $p(a)=\frac{\widetilde{\E}[X_{u,a}X_{v,\sigma}]}
    {\widetilde{\E}[X_{v,\sigma}]}$.
\end{enumerate}
To get an understanding to why this rounding scheme works, think of $X$ as fixed and $X'$ as random. Thus, the fact that part $s$ of
the shift partition is large implies 
that $X' = X+s$ 
on a constant fraction of the vertices.
Therefore, once we sampled the assignment to $v$ in the first part 
of the algorithm, the value of $s$ is determined. In the second step we are sampling the assignment to other 
nodes conditioned on the value of $v$. However, there is one
value for $u$ which is much more likely than others -- namely 
$X(u)+s$, and so we can expect that $X'(u) = X(u) + s$ for a constant fraction of the vertices $u$. In particular, for any edge $(u,w)$ inside $F_s$ that is satisfied by $X$, we will have that the assignments sampled for $u$ and $w$ are $X(u)+s$ and $X(w) + s$ respectively with constant probability, in which case we manage to satisfy $(u,w)$. To analyse this rounding strategy formally,~\cite{BBKSS} crucially use the independence of $X$ and $X'$. 

In essence, the above asserts that the shift-partition being large implies that the solution space of $X$ must have high collision probability, which can then be used for rounding. By that, we mean that our distribution essentially consists of only one assignment (upto shift-symmetry) and its perturbations.

\paragraph{Non small-set expanders.} Consider a graph which is 
not a small set expander, say that $G$ is the Johnson graph $J(n,\ell,t=\ell/2)$. In that case the above reasoning no longer 
works as $F_s$ may indeed be all small sets. However, as explained 
earlier, using global hypercontractivity we can infer that one of the sets $F_s$ must posses a 
certain structure -- it must have large density inside one of the canonical non-expanding sets.
In the case of the Johnson graph specifically, these canonical sets take the following form:
\[
H_R = \left\{A\in\binom{n}{\ell}~\Big|~A\supseteq R\right\},
\]
for $R \subset [n]$ and $|R| = r = O(1)$. In fact, global hypercontractivity gives the following stronger structural property: the set $H = \bigcup_{R\in\mathcal{R}} H_R$ where $\mathcal{R}$ 
consists of all $R$'s inside which some part $F_s$ is dense, 
has a constant measure. Doing simple accounting, it follows that $|\mathcal{R}|\geq \Omega(n^r/\ell^r)$
and as there are at most $\binom{n}{r}$ different canonical sets it follows that $|\mathcal{R}|$ contains an
$\Omega(1/\ell^r)$ fraction of these sets.

For each choice of $X$ and $X'$ though 
we may have a different collection of dense subcubes $\mathcal{R}$. But since $\mathcal{R}$ contains an $\Omega(1/\ell^r)$ fraction of all the subcubes, we get that there must be at least one subcube $H_R$ that is dense with probability $\Omega(1/\ell^r)$ over $X,X' \sim \cD$. Let $H_R$ be such a subcube and $E_R(X,X')$ be the event that $H_R$ is dense. Ideally, at this point one would like to condition on $E_R$ so that one of the parts inside the shift partition $F_s$ becomes large inside $H_R$, and then hope that as was the case for small-set expanders, we can satisfy many of the edges inside $F_s\cap H_R$. 


Unfortunately, this hope does not materialize -- after conditioning on $E_R$ even though the shift-partition is large, the rounding strategy above may break. Indeed, 
for the rounding procedure we wanted the values of $X(u)$ and $X'(u)$ for $X,X'\sim\mathcal{D}$ to be independent 
for every vertex $u$. However, after conditioning the joint distribution $\cD \times \cD ~|~E_R$ over $(X,X')$ might have correlations between $X$ and $X'$. In particular this distribution could even be supported on pairs $(X,X')$ that are always equal to each other, in which case the shift-partition is large because of trivial reasons and therefore its large size doesn't imply anything about the collision probability/entropy of $\cD$. 

Hence in~\cite{BBKSS} the authors don't manage to do this conditioning, and instead settle 
for satisfying an $\Omega\left(\frac{1}{\ell^{2r}}\right)$-fraction of the constraints on $H_R$. After that they iterate this algorithm many times to satisfy an $\Omega\left(\frac{1}{\ell^{2r}}\right)$-fraction of the constraints of the whole graph. 

\subsubsection{Our Approach: Conditioning on the Event $E$ via
(Eliminating) Global Correlations}
Our main contribution to the 
above framework is to show 
that by adding an additional preprocessing step, we can ensure that even after conditioning on the event $E_R$ 
above, the assignments $X$ and 
$X'$ will remain highly 
independent. In particular, 
the fact that some part in
the shift partition becomes
large must happen -- just like
in the case of small-set expanders -- due to the fact
that our distribution has high collision probability.

As the event $E = E_R(X,X')$ has 
probability at least $\Omega\left(\frac{1}{\ell^r}\right)$, 
 if we are sufficiently high up in the SoS hierarchy 
($\Theta(\ell^r)$ levels will do, for an overall running time of 
$n^{\Theta(\ell^r)}$), we do have access to the conditional pseudoexpectation
\[
\pE[Y~|~E] = \frac{\pE[Y 1_E]}{\pE[1_E]}.
\]
This means that we can sample labels of vertices 
conditioned on the event $E$. To make this useful though, we must change the 
rounding procedure. 
To get some intuition consider 
the extreme case in which after conditioning on $E(X,X')$ 
there are huge correlations between $X$
and $X'$ that remain in our distribution. 

Namely, suppose that after conditioning on $E$ it holds that $X(u) = X'(u)$ for almost all vertices $u$. 
In that case, if we sampled 
$X,X'$ from $\cD \times \cD$ (not conditioned on $E$), we would get that with probability 
at least $\Pr[E]\geq \Omega\left(\frac{1}{\ell^r}\right)$
the event $E$ holds, in which case $X$ and $X'$ agree on almost all vertices. This means that if $\cD$ was an actual distribution the assignments have a large global correlation: fix $X' = X_0$ for $X_0$ that satisfies $\Pr_\cD[E(X,X_0) = 1] \geq \Pr_{\cD \times \cD}[E]$. Once $E$ holds, 
we have that $X(u) - X(v) = X_0(u) - X_0(v)$ for almost all pairs of vertices, hence the values of the assignment $X$ to the vertices $u$ and $v$ is correlated across $\cD$.
Therefore, a natural idea is to avoid this issue by transforming $\cD$ to another distribution lacking global correlations, in the sense that the assignments to a typical pair of vertices $u$ and $v$ are almost independent.

For this purpose we use an idea from~\cite{RT12}, which adapted to our setting says 
that for any $\tau>0$ there is $d = d(\tau,|\Sigma|)$ such that conditioning $\pE$ on the values of $d$ randomly chosen vertices ensures that the global correlation is 
at most $\tau$. That is, the values of $X(u)$ and $X(v)$
for two typical vertices $u$ and $v$ are at most $\tau$-correlated, and the same holds for $X'$. 
In Lemma~\ref{lem:correlation} we then show that if we start with such a pseudodistribution that lacks global correlations, then one can condition on the event $E$ and retain near independence between the assignments $X$ and $X'$, at least on most vertices. To be more precise, we show that for $Y_{u,v} = (X(u), X(v))$ and $Y_{u,v}' = (X'(u), X'(v))$, the statistical distance between $Y_{u,v}, Y'_{u,v}~|~E$ 
and $Y_{u,v}, Y'_{u,v}$ is small for almost all pairs of vertices $u,v$.\footnote{To make our rounding succeed we need to use a more complicated version of $Y_{u,v}$ (see Definition~\ref{defn:local-dist}). Proving Lemma~\ref{lem:correlation} then turns out to be  technically challenging since these $Y_{u,v}$'s are not indicator variables, only approximately so and additionally we need to make sure that our proof works for pseudodistributions instead of just for actual distributions.}

Using this idea we are able to get an $\Omega(1)$-valued solution on some basic set $H_R$. To summarize, we first preprocess the pseudodistribution to eliminate global correlations. We can then find an event $E(X,X')$, corresponding to the 
fact that some part $F_s$ in the shift partition has becomes dense in some basic set $H_R$. Furthermore, conditioning on $E$ most pairs $(X(u),X(v)), (X'(u),X'(v))$ remain almost-independent. Then running a simple rounding procedure on $H_R$ (as in~\cite{BBKSS}), we are able to satisfy a good fraction of the edges inside $H_R$. $H_R$ might be a $o(1)$ fraction of the graph though, therefore like~\cite{BBKSS} we repeat this procedure multiple times to get an $\Omega(1)$-valued solution for the whole graph. This gives an efficient algorithm for affine UG over the Johnson graphs as in Theorem~\ref{thm:main_johnson_close_to_1}.

To prove Theorem~\ref{thm:main_johnson}
(namely, the regime where $c$ is not close to $1$) more work is needed. Indeed, in the case that $c$ is 
close to 
$1$ we are able to conclude that essentially all edges 
stay within some part $F_s$ of the shift partition. Thus, 
as long as our sets $H_R$ cover a constant fraction of 
the edges that stay within some $F_s$, they are automatically guaranteed to cover a constant fraction of the edges that 
are satisfied by both $X$ and $X'$, and these are the edges our rounding procedure manages to satisfy. If $c$ is just 
bounded away from $0$ we can no longer make such an argument, 
and it is no longer even clear that the sets 
$H_R$ cover some edges that we have a hope of satisfying.

\subsubsection{Getting Small Completeness: Capturing all of the Non-expanding Edges}
To design our algorithm for the case when the completeness $c$ is just 
guaranteed to be bounded away from $0$ we must first argue that in the shift partition, we are able to capture almost all 
of the edges that stay within a part $F_s$  
using the basic sets $H_R$ (so as to ensure we are including 
the edges that $X$ and $X'$ both satisfy).

Towards this end we require a more
refined corollary of global hypercontractivity, asserting 
that if we have a small set of vertices $F$ in the Johnson 
graph that has edge expansion at most $1-\eta$, then we 
can find a collection $\mathcal{R}$ of basic sets such that:
\begin{enumerate}
    \item \textbf{Bounded and dense:} 
    each $R\in \mathcal{R}$ has size $|R| = O(1)$ 
    and $F$ is dense inside each $H_R$. That is, 
    $\delta(F\cap H_R)\geq \Omega_{\eta}(\delta(H_R))$ for 
    each $R\in\mathcal{R}$.
    \item \textbf{Maximally dense:} 
    For all $R\in\mathcal{R}$ and all $R'\subsetneq R$, 
    $F$ is not very dense in $H_{R'}$.
    \item \textbf{Capture almost all non-expanding edges:} 
    Almost all the edges that stay inside $F$ also stay inside 
    $H_R$ for some $R\in\mathcal{R}$.
\end{enumerate}
Indeed, we show that a global hypercontractive inequality such as the one 
in~\cite{KMMS} can be used to prove such a result (in a black-box manner).

Using this result, we are able to argue that that the edges that 
stay inside the subcubes $H_R$ for $R\in\mathcal{R}$  
cover most of the edges that stay within the same part in the shift
partition. There are several subtleties here that one has to deal with, for example, ``regularity issues'' such as, how many different $R$'s cover 
a given edge. The goal of the second item above is to handle such concerns, and it roughly says that no vertex nor 
edge gets over-counted by a lot. After that, we are able to 
condition on an event $E$, where as before $E$ indicates that some part 
$F_s$ becomes dense inside some basic set $H_R$, so that the resulting distribution has a large shift-partition inside $H_R$.
At this point, we are (morally) 
back to the problem of rounding the SoS solution 
on a set with a large shift-partition, except that now our solution 
has value $c'>0$ (as opposed to close to $1$). We remark that again, we use the ``elimination of global correlations'' 
idea presented earlier to retain near independence after conditioning. With more care, we use a similar analysis to the one presented for completeness close to $1$ to finish the proof when $c$ is arbitrarily small.

\subsection{Open Problems}
We end this introductory section by stating a few open
directions that are of interest for future research.
The first problem asks whether our results continue to 
hold for non-affine unique games:
\begin{problem}
For globally hypercontractive graphs $G$ 
such as the Johnson graph 
(with small intersection size) and the Grassmann graph, is there
a polynomial time algorithm that given a UG instance 
$\Psi$ over $G$ with ${\sf val}(\Psi)\geq 1-\eps$ (where $\eps>0$ is thought of as small), finds 
an assignment satisfying at least $\delta$ fraction of the
constraints in $\Psi$? How about the case that ${\sf val}(\Psi)\geq c$, where $c$ is bounded away from $1$?
\end{problem}

The second problem asks whether there are other combinatorial optimization problems for which our techniques may yield improved algorithms. Informally, we show how to round pseudodistributions with low-entropy, including for instance those that are supported over $O(1)$ assignments (and their perturbations). We do so by proving that after the
elimination of global correlations one can retain local independence properties even after conditioning. We believe that this technique should be useful outside the context of UG -- given any problem for which one can prove (in SoS) that there are only a ``few good solutions'', one can apply our rounding technique to obtain one such solution.

\begin{problem}
Can one use the low-entropy rounding framework to get improved run-time for other combinatorial optimization problems, such as coloring $3$-colorable graphs using as few colors as possible?
\end{problem}

Third, it would be interesting to see whether our techniques 
can help in designing improved sub-exponential time algorithms
for combinatorial optimization problems such as Max-Cut over general graphs (ideally) or over special classes of graphs.
\begin{problem}
Can one use the low-entropy rounding framework to get improved approximation algorithm for Max-Cut that runs in time $2^{n^{\eps}}$? 
\end{problem}

\section{Preliminaries}\label{sec:prelims-sos} 

\paragraph{Notations.} 
For a (weighted) graph $G=(V,E)$, we denote by $(u,v) \sim E$ the distribution over edges proportional to their weight, and by $A_G$ the transition matrix related to the random walk on $G$.
We define the Laplacian of $G$ by $L_G = I - A_G$, and denote the stationary distribution of $G$ by $\pi_G$. In this paper we will mostly be working with regular graphs and therefore the uniform distribution over $G$. We use $\ip{f,g}_G$ to denote $\E_{u \sim \pi_G}[f(u)g(u)]$. We often drop the subscript of $G$ when it is clear from context. 

If $A$ is some probabilistic event or condition, we use $\Ind(A)$ to denote the indicator random variable of $A$ (i.e., $\Ind(A)=1$ if $A$ occurs and $\Ind(A)=0$ otherwise).





\subsection{The Sum of Squares Relaxation of Unique-Games}
Our algorithm is based on the \sos semidefinite programming (SDP) relaxation, and in particular its view as optimizing over \emph{pseudo expectation} operators. In this section, we 
briefly present the necessary background as well as several 
basic facts that we will use, and we refer the reader to the
surveys \cite{BarakS14,RaghavendraSS18,TCS-086} for a more 
systematic presentation.

Given a Unique-Games instance $I= (G=(V,E),\Sigma,\Pi)$
the value of $I$ can be computed by the following integer program over zero-one variables $\{ X_{u,a} \}_{u \in V, a \in \Sigma}$:
\begin{align}
\max_{X}& \E_{(u,v) \in E} \sum_{a \in \Sigma} X_{u,a} X_{v,\pi_{uv}(a)} \label{eq:ip}\\
s.t. 
&\quad X_{u,a}^2 = X_{u,a} ~\qquad\quad \forall u \in V, a \in \Sigma \nonumber,\\
&\quad X_{u,a}X_{u,b} = 0 \qquad\quad \forall u \in V, a \neq b \in \Sigma \nonumber,\\
&\quad \sum_{a} X_{u,a} = 1 ~\qquad\quad \forall u \in V\nonumber.
\end{align}

Indeed, an assignment to the variables $X_{i,a}$ represent 
an assignment to the vertices of $G$, in which a vertex $v \in V$ takes label $a \in \Sigma$ if $X_{v,a} = 1$, and the objective function counts the fraction of constraints 
that are satisfied satisfied.

The degree $D$ \sos relaxation of the above program is a convex 
optimization program that relaxes the above program. This 
program can be written as an optimization problem in which 
the goal is to find a vector-valued assignment to the variables $X_{S,\sigma}$, 
where we have such variable for each set of vertices $S$ of size at most $D$ and labels for them $\sigma\in \Sigma^S$. 
The constraints of the program address the inner products between these vectors, attempting to ensure that for any two sets 
of vertices $S$ and $T$ such that $|S\cup T|\leq D$, 
the inner products 
$\langle X_{S,\sigma}, X_{T,\sigma'}\rangle$ represent 
a distribution over assignments to $S\cup T$, and that 
these distributions are locally consistent. 
It will be more convenient for us to take a different 
but equivalent view on the solutions to the \sos relaxation,
in the language of pseudoexpectations as presented next.

\subsection{Pseudoexpectations, Pseudodistributions and Pseudoprobabilities}
\paragraph{Pseudoexpectations.}
A degree $D$ pseudoexpectation $\pE$ is an operator $\pE:X^{\le D} \to \R$, where $X^{\le D}$ is the set of all monomials in the $X$ variables up to degree $D$, and $\pE$ satisfies the above equality constraints and the Booleanity constraints $\{X_{u,a}^2 = X_{u,a}\}$ as axioms.  For brevity, we will refer to this set of axioms as $\A_{I}$, dropping the subscript when $I$ is clear from context. More generally, given a polynomial optimization program $P = \{\max_{x} p(x) \ s.t. ~ q_i(x) = 0, \forall i \in [m] \}$, the degree-$D$ sum-of-squares semidefinite programming relaxation of $P$ is a semidefinite program of size $n^{O(D)}$ that returns a {\em pseudoexpectation operator} $\pE: x^{\le D} \to \R$. This operator can be uniquely extended to give a pseudoexpectation operator on the set of all polynomials of degree at most $D$ by linearity (defined precisely below).
This operator satisfies four properties:
\begin{itemize}
\item Scaling: $\pE[1] = 1$.
\item Linearity: $\pE[a \cdot f(x) + b \cdot g(x)] = a\cdot \pE[f(x)] + b \cdot \pE[g(x)]$, for all $a,b \in \R$ and all degree $\leq D$ polynomials $f,g$. 
\item Non-negativity of low-degree squares: $\pE[s(x)^2] \ge 0$ for all polynomials $s(x)$ with $\deg(s) \le \tfrac{D}{2}$.
\item Program constraints: $\pE[f(x) \cdot q_i(x)] = 0$ for all $i \in [m]$ and polynomials $f(x)$ such that $\deg(fq_i) \le D$.
\end{itemize}
Depending on the problem we are trying to solve, we also 
discuss the value achieves by the pseudoexpectation $\pE$, 
which is defined to be the pseudoexpectation of the objective 
function. In our case of interest, namely the case of Affine 
Unique-Games, the \emph{value} of $\pE$ on the instance $I$ is denoted by $\val_\mu(I)$ and is defined to be 
\[
\pE[\val_{I}(X)] = \pE \left[ \E_{(u,v) \in E} \sum_{a \in \Sigma} X_{u,a} X_{v,\pi_{uv}(a)}\right].
\]
Our pseudoexpectation will be guaranteed to achieve a value which matches the completeness guarantee (for example, in 
the context of Theorems~\ref{thm:main_johnson},~\ref{thm:main_grass} it will satisfy that $\pE[\val_{I}(X)]\geq c$).

\paragraph{Pseudodistribution.} 
We often refer to a pseudoexpectation operator $\pE$ as an operator corresponding to a pseudodistribution $\mu$ over variables $X$. This is analogous to the case where we have an actual distribution $\mu$ and its corresponding expectation operator $\E_\mu[\cdot]$. This notation makes our analysis using $\pE$ operators more intuitive, since many properties that are true of actual distributions also hold for pseudodistributions. Hence when we say that we are given a degree $D$ pseudodistribution $\mu$ we are referring to the degree $D$ pseudoexpectation operator $\pE_\mu$.


\paragraph{Pseudoprobabilities.}
\begin{definition}[Pseudoprobability of an event]\label{def:pseudoprob}
Let $\mu$ be a pseudodistribution of degree $D$.
For an event $\cE(X,X')$ such that $\Ind[\cE(X)]$ can be expressed as a degree-$D$ function of $X$, we define the {\em pseudoprobability of $\cE(X)$} to be
\[
\pPr_{\mu}[\cE(X)] = \pE_{\mu}[\Ind(\cE(X)].
\]
Similarly, if $\cF(X)$ is an event and $\deg(\Ind[\cF(X)]) + \deg(\Ind[\cE(X)]) \le D$,  we define the {\em pseudoprobability of $\cE(X)$ conditioned on $\cF(X)$} to be
\[
\pPr_\mu[\cE(X) \mid \cF(X)] = \pE_\mu[\Ind(\cE(X)) \mid \Ind(\cF(X))] = \frac{\pE[\Ind(\cE(X)) \cdot \Ind(\cF(X))]}{\pE[\Ind(\cF(X))]}.
\]
\end{definition}

\subsection{SoS-ing Mathematical Proofs}
Our argument will use several mathematical statements (such 
as global hypercontractivity), and we will need to be 
able to argue that these statements are also satisfied 
in the context of the SoS program and its variables. Thus, 
we will need to be able to ensure that we use tools that can 
be proved via sum of squares inequalities of low-degree (and 
axioms), and below we collect a few such standard tools that we will use. We use the following notation for SoS prooofs.

\paragraph{Sum of squares proofs notations.}
Given a set of axioms $\cA = \{q_i = 0\}_i \cup \{g_j \ge 0\}_j$ for polynomials $q_i, g_j \in \R[X]$, we say that ``there is a degree-$d$ sum-of-squares proof that $f \ge h$ modulo $\cA$'' if: $f = h + s + \sum_k c_k\cdot Q_k + \sum_t r_t \cdot G_t$ where each polynomial $Q_k$ and $G_t$ is a product of some polynomials from $\{q_i\}$ and $\{g_j\}_j$ respectively, $s,\{c_k\}_k,\{r_t\}_t \in \R[X]$ are real polynomials such that $s$ and $\{r_t\}_t$ are sums of squares, and the maximum degree among $s,\{c_k Q_k\}_k, \{r_t G_t\}_t$ is at most $d$. We will use the notation $\cA \vdash_d f(x) \ge h(x)$ to denote the existence of such an equality.
We also sometimes use $f(x) \succeq h(x)$ to denote that the inequality is a \sos inequality.

\subsubsection{Basic Inequalities}
The first of which is the following basic forms of the Cauchy-Schwarz and H\"{o}lder inequalities. 
The proofs are by now standard and can be found for example in~\cite{BarakKS14,ODonnellZ13}.
\begin{lemma}[Cauchy Schwarz]\label{CS1-prelim}
For for all $\epsilon \in \R_+$, 
$$
\vdash_2 YZ \leq \frac{\epsilon}{2}Y^2 + \frac{1}{2\epsilon}Z^2.
$$
\end{lemma}

\begin{lemma}[Cauchy Schwarz]\label{CS2-prelim}
A degree-$D$ pseudoexpectation operator where 
$D \geq  2\max(\deg(f),\deg(g))$ satisfies that
$$\pE[fg]^2 \leq \pE[f^2]\pE[g^2].$$
\end{lemma}

\begin{fact}[H\"{o}lder's Inequality]\label{fact:sos-hol}
For all real $\nu > 0$ we have that,
$$\vdash_4 Y^3 Z \leq \frac{3\nu}{4} Y^4 + \frac{1}{4\nu^3}Z^4.$$ 
\end{fact}

The theory of univariate sum-of-squares (in particular, Luk\'{a}cs Theorem) says that if a univariate polynomial is non-negative on an interval, this fact is also SoS-certifiable. The following  corollary of Luk\'{a}cs theorem is well-known.
\begin{corollary}[Corollary of Luk\'{a}cs Theorem]\label{cor:Lukacs}
Let $q$ be a degree-$d$ polynomial which is non-negative on $[a,b]$.
Then given the axioms $\cA = \{x \ge a\} \cup\{x \le b\}$, there is a degree-$2d$ \sos proof that $q$ is non-negative, $\cA \vdash_{2d} q(x) \ge 0$.
\end{corollary}
We will use the above multiple times to convert univariate inequalities into \sos inequalities in a blackbox manner.

Given multivariate inequalities the theorem above no longer holds. Nevertheless given a strictly positive polynomial $f$ that is bounded away from $0$, one can get an \sos proof of degree which is exponential in $\deg(f)$ using Theorem 3 in~\cite{schweighofer}. We state a corollary of~\cite{schweighofer} that is sufficient for our purposes. 

\begin{theorem}[Corollary of~\cite{schweighofer}]\label{thm:blackbox-sos}
Let $S$ be the set $[0,1]^k$. Let $f(x)$ be a polynomial of degree $d$ with $\nu = \min\{f(x) \mid x \in S\} > 0$ and $||f||$ denoting the maximum absolute value of $f$'s coefficients. Then $f$ has an \sos certificate of bounded degree:
\[\{x_i \in [0,1] | i \in [k]\} \vdash_D f(x) \geq 0,\]
with $D = \poly(k^{\deg(f)}, ||f||/\nu)$.
\end{theorem}

We will use the above fact in Section~\ref{sec:johnson_2} to convert inequalities involving polynomials of constant degree and on constantly many variables into \sos inequalities. We believe it should be possible to get a degree bound above which is $\poly(\deg(f))$ in cases where we have reasonable polynomials $f$, though as far as we know this has not been proved in generality.

\subsubsection{Approximating Indicators via Low-degree Polynomials}\label{sec:apx-ind}


Our argument will involve indicators of events such as 
$f(X)\geq \beta$ where $f$ is a low-degree polynomial, and we will want to condition on such events. Strictly speaking, the function $\Ind[f(X)\geq \beta]$ is not a low-degree polynomial and therefore we cannot condition on it. However, it is not difficult to show that such indicators can be approximated by low-degree polynomials, and we will need to use such ideas. 
Indeed, in this section we present such an approximation theorem that will be used throughout our proofs.

The following theorem, due to~\cite{diakonikolas2010bounded}, provides a low-degree approximation to a step function.

\begin{theorem}\label{thm:step-approx}
Let $\Ind[x \geq \beta]$ be the step function at $\beta \in (0,1)$. Then for each $0 < \nu < \beta$ there is a univariate polynomial $p_{\beta,\nu}$ of degree $O(\frac{1}{\nu}\log^2\frac{1}{\nu})$ such that:
\begin{enumerate}
\item $|p_{\beta,\nu}(x) - \Ind[x \geq \beta]| \le \nu$ for all $x \in [0, \beta] \cup [\beta +\nu, 1]$.
\item $p_{\beta,\nu}$ is monotonically increasing on $(\beta, \beta + \nu)$.
\item $0 \le p_{\beta,\nu}(x) \le 1$ for all $x \in [0,1]$.
\item All coefficients of $p_{\beta,\nu}$ are at most
$2^{O(\log(1/\nu)/\nu)}$ in absolute value.
\end{enumerate}
Further the first three facts are SoS-certifiable in degree $2\deg(p_{\beta,\nu})$.
\end{theorem}
The first three items follow from~\cite[Theorem 4.5]{diakonikolas2010bounded}, the fourth 
item is immediate by Markov brothers' inequality, and the SoS certifiability follows from Corollary~\ref{cor:Lukacs}.

The following fact provides convenient point-wise bounds on the approximating polynomials of indicators from above. 
\begin{fact}[Markov Inequality for Bounded Polynomials]
\label{fact:bdd-markov}
Let $p := p_{\beta,\nu}$ be the degree-$D = \tO(1/\nu)$ polynomial guaranteed by Theorem \ref{thm:step-approx}.
Then $p$ satisfies Markov's inequality:
\begin{equation*}
\{0 \le x \le 1\}
\vdash_{\deg(p)} 
\{ p(x) \ge 1 - \frac{1-x}{1-\beta-\nu} - \nu\} \cup \{p(x) \le \frac{x}{\beta-\nu} + \nu\}
\end{equation*}
\end{fact}
\begin{proof}
We perform case analysis on $x$ and then use Corollary~\ref{cor:Lukacs} to conclude the proof
is \sos.
For the first inequality, for $x \in [0,\beta+\nu)$ 
we have
\[
p(x) \succeq 0 \succeq 1 - \frac{1-x}{1-\beta - \nu},
\]
where we have used that $p(x) \succeq 0$ and $\frac{1-x}{1-\beta - \nu}\succeq 1$. For $x \in [\beta +\nu, 1]$,
\[
p(x) \succeq 1 - \eps \succeq 1 - \eps - \frac{1-x}{1-\alpha-\delta},
\]
where we have used that $x \in [0,1]$ so that we are subtracting a positive quantity.
Combining these claims concludes the proof of the first claim.

To see the second claim, notice that for $x \in [0,\beta-\nu]$, $p(x) \le \nu$, and for $x \in (\beta - \nu, 1]$, $p(x) \le 1 \le \frac{x}{\beta - \nu}$.
This concludes the proof.
\end{proof}

\subsection{Manipulating pseudoexpectations}

\subsubsection{Reweighing and conditioning:} 
We will sometimes {\em reweigh} or {\em condition} our degree-$D$ pseudodistribution by a polynomial $s(x)$ where $s(x)$ is non-negative under the program axioms, i.e. $\cA \vdash_{d} s(x)$ for $d < D$. Technically, this operation amounts to defining a new pseudoexpectation operator $\pE'$  of degree $D-d$ by taking, 
\[\widetilde{E}'[x^{\alpha}] = \frac{\pE[x^{\alpha} \cdot s(x)]}{\pE[s(x)]},\] 
for every monomial $x^{\alpha}$ of degree at most $D - d$. As an example, under the unique games axioms presented in \ref{eq:ip} one can prove that a variable $X_i$ is in $[0,1]$, hence one can reweigh the pseudodistribution by $X_i$. One can show that reweighing preserves the four properties of the pseudodistribution up to degree $D - d$.
Thus, we will also refer to this operation as ``conditioning'', and denote $\pE'$ by $\pE[\cdot ~|~ s(x)]$. Often times, the polynomial $s(x)$ we will ``condition'' on 
will be a smooth approximation of some event $E$, in which case the above operation takes the interpretation of conditioning our sample from the pseudodistribution to satisfy some properties specified by the event $E$. We refer the reader to \cite{BarakRS11,BKS17} for further discussion of reweighting.

\subsubsection{Independent Samples}\label{sec:ind_samples}
Recall that a given pseudoexpectation operator $\pE: X^{\leq D} \rightarrow \R$ has the interpretation as averaging of functions $f(X)$ over a pseudodistribution $X\sim\mu$.  
We will need to be able to mimic averaging over two 
independently chosen samples $X,X'\sim \mu$,\footnote{Similar constructs have been used in the literature, see e.g. \cite{BarakKS14}.} for that 
we define the product pseudoexpectation $\pE_{X,X'}$ 
as follows: let $X^{\alpha}(X')^{\beta}$ be a monomial 
of degree at most $D$ in variables $X,X'$; we define  $\pE_{X,X'}[X^{\alpha}(X')^{\beta}] := \pE_{X}[X^{\alpha}] \cdot \pE_X[X^{\beta}]$. It is easy to check that $\pE_{X,X'}$ 
is also a pseudoexpectation operator corresponding to two 
independent samples from the pseudodistribution $\mu$; 
see Fact~\ref{fact:indep}.

Given two independent samples $X,X'\sim \mu$, we will often be
interested in the variables $\{Z_{v,s}\}_{v\in V, s\in\Sigma}$ 
that are used to define the shift-partition discussed in the introduction. Formally, we define them as
\[
Z_{u,s} = \sum_{a \in \Sigma} X_{u,a} X'_{u,a+s}.
\]
In the rest of this section we present a few facts about polynomials in independent samples,
and the pseudodistribution $\pE_{X,X'}$ on them. 
Some of these facts will be general, and some of which 
will be specific to the shift partition variables $Z_{u,s}$.

The first fact asserts that $\pE_{X,X'}$ defined above is a 
legitimate pseudodistribution that inherits all of the 
constraints that $\pE$ satisfies.
\begin{fact}\label{fact:indep}
If $\pE_{X}$ is a valid pseudodistribution of degree $D$ in variables $X$, then $\pE_{X,X'} $ is a valid pseudodistribution of degree $D$. Furthermore, if there are additional \sos inequalities that are true for $\T_X$, they also hold for $\T_{X,X'}$.
\end{fact}


Now, we prove some properties specific to the $Z$ variables. 
The following fact asserts that the shift partition variables 
$Z_{v,s}$ indeed behave like a partition, in the sense 
that if $X$ satisfy the constraints in the program~\ref{eq:ip}, 
then the variables $Z_{v,s}$ indeed define a partition.
\begin{fact}\label{fact:z-vars}
Define the shift variable $Z_{u,s} = \sum_{a \in \Sigma} X_{u,a}X'_{u,a+s}$ to be the indicator that $X_u - X_u' = s$, for $X,X'$ degree-$8$ solutions to the \sos relaxation of the UG integer program (\ref{eq:ip}), and for each edge $(u,v)$ the variable $Y_{(u,v)} = \sum_{a} X_{u,a} X_{v, \pi_{uv}(a)}$ to be the indicator that the constraint on the edge $(u,v)$ is satisfied.

Then the $Z$ variables satisfy:
\begin{enumerate}
\item Booleanity: $Z_{v,s}^2 = Z_{v,a}$. 
\item Partition constraints: $Z_{v,s}Z_{v,s'} = 0$ for $s \neq s'$, and $\sum_{s} Z_{u,s} = 1$.
\item Crossing edges violate an assignment: $Z_{u,s} Z_{v,s'} Y_{(u,v)} Y'_{(u,v)} = 0$ for every edge $(u,v) \in E$ and $s \neq s'$. 
\end{enumerate}
\end{fact}
\begin{proof}
The first two items are easily verified via direct computation, using properties of the variables $X_{u,a}$.
We prove that the final property holds.
Since our UG instance is affine, we have that for each $i,j \in E$, $\pi_{ij}(a) = a + h_{ij}$ for some $h_{ij} \in \Sigma$.
 Therefore,
 \begin{align*}
 Z_{i,s} Z_{j,t}Y_{(i,j)}Y'_{(i,j)}
 &= \sum_{a,b,c,d \in \Sigma} X_{i,a}X'_{i,a+s} \cdot X_{j,b}X'_{j,b+t} \cdot X_{i,c}X_{j,c + h_{ij}} \cdot X'_{i,d} X'_{j,d+h_{ij}}\\
 &= 0,
 \end{align*}
 where we derive the final equality from the disjointness constraints (i.e. that $X_{i,a}X_{i,b} = 0$ whenever $a \neq b$), as for the above term to be nonzero we require $a = c$, $d = a + s$, $b = d - t + h_{ij} = a + s - t + h_{ij}$, and also $b = c + h_{ij}$, which implies $a + s - t = c$, a contradiction since $t \neq s$.
This establishes the final property.
\end{proof}

\subsubsection{Shift-Symmetry}\label{sec:symm-prelims}
Next, we define the notion of shift-symmetric functions and pseudodistributions and establish some properties that they
satisfies.
\begin{definition}[Shift-Symmetry]\label{def:shift-symm}
We say that a pseudodistribution $\mu$ is shift symmetric if 
for any monomial $\prod\limits_{i=1}^{D} X_{u_i, a_i}$ and $s\in\Sigma$ it holds that
\[
\pE_{\mu}\left[\prod\limits_{i=1}^{D} X_{u_i, a_i}\right]
=
\pE_{\mu}\left[\prod\limits_{i=1}^{D} X_{u_i, a_i+s}\right].
\]
\end{definition}
We say that a polynomial $p(X_{u_1,\sigma_1},\ldots)$ equivalently also thought of as a function $f: \Sigma^n \rightarrow \R$ is shift-symmetric if $f(X) = f(X+s)$ for all $s \in \Sigma$.

Given a pseudodistribution obtained by the \sos relaxation of sum of squares of an 
Affine Unique-Game, we are able to 
transform it into a shift-symmetric pseudodistributions 
with the same value by defining a new pseudodistribution $\mu^{\text{sym}}$ as follows
\[
\widetilde{E}_{\mu^{\text{sym}}}\left[\prod\limits_{i=1}^{D} X_{u_i, a_i}\right] 
= 
\frac{1}{|\Sigma|}\sum\limits_{s\in\Sigma}
\pE_{\mu}\left[\prod\limits_{i=1}^{D} X_{u_i, a_i+s}\right].
\]
It is easy to check that $\pE_{\mu^{\text{sym}}}$ is a valid pseudoexpectation, 
and that it's value is the same; indeed, this follows since if $(a,b)\in \Sigma$ satisfies a constraint in an affine Unique-Games instance, then $(a+s,b+s)$ also satisfy that 
constraint for any $s\in\Sigma$. More generally one can check that $\pE_{\mu^{\text{sym}}}[p(X)] = \pE_{\mu}[p(X)]$ for all shift-symmetric polynomials $p$.
Thus, we will assume henceforth that our pseudoexpectation and the pseudodistribution can be made shift symmetric without losing the value and in general preserving the pseudoexpectation of shift-symmetric polynomials.

\subsection{Information Theory}
We will use $\mu|_R$ to denote the marginal distribution of a random variable $R \sim \mu$. We use $TV(A,B)$ to denote the total-variation distance between two distributions $A,B$.

\begin{definition}[Mutual Information]
Given a distribution $\mu$ over $(X,Y)$, the mutual information between $X,Y$ is defined as:
\[I_\mu(X;Y) = D_{KL}(\mu~ || ~\mu|_X \times \mu|_Y),\]
where $D_{KL}$ is the Kullback-Leibler divergence. The conditional mutual information between $(X,Y)$ with respect to a random variable $Z$ is defined as:
\[I(X;Y | Z) = \E_{z \sim Z}[I_{\mu | Z = z}(X;Y)].\]
\end{definition}

\begin{lemma}[Pinsker's inequality]
Given any two distributions $D_1,D_2$:
\[TV(D_1, D_2) \leq \sqrt{\frac{1}{2}D_{KL}(D_1 || D_2)}.\]
Using this we get that for a distribution $\mu$ over $(Y,Y')$:
\[TV(\mu, \mu_Y \times \mu_{Y'}) \leq O(\sqrt{I(Y;Y')}).\]
\end{lemma}

\begin{lemma}[Data processing inequality]
Let $X,Y,A,B$ be random variables such that $H(A|X) = 0$ and $H(B|Y) = 0$, i.e. $A$ is fully determined by $X$ and $B$ is fully determined by $Y$. Then:
\[I(A;B) \leq I(X;Y).\]
\end{lemma}

\section{Proof of Theorem~\ref{thm:main_johnson_close_to_1}}
\label{sec:johnson_1}
In this section we prove Theorem~\ref{thm:main_johnson} 
in the case that $c = 1-\eps$ where $\eps>0$ is small 
so as to isolate the ``conditioning on an event'' challenge as explained in the introduction. In the next section, we explain the modifications that are necessary to prove Theorem~\ref{thm:main_johnson} in full generality. 

We begin by formally defining the Johnson graph.

\begin{definition}[Johnson Graph] \label{def:johnson}
For any $\alpha \in (0,1)$ and $n,\ell \in \N$ with $\alpha\ell \in \N$ and $n > \ell$, we define the \emph{$(n,\ell,\alpha)$-Johnson graph $J(n,\ell,\alpha\ell )$} to be the graph whose vertex set is $\binom{[n]}{\ell}$ and where edges are between pairs of vertices $U,V \in \binom{[n]}{\ell}$ if and only if $|U \cap V| = (1-\alpha)\ell$. We will drop the $(n,\ell,\alpha\ell)$ when clear from context,  
and use $[N]$ to denote the set of vertices of $J$. 

We will often refer to $\alpha$ as the noise parameter 
of the graph.
\end{definition}

\paragraph{Notation:}
The $(n,\ell,\alpha\ell)$-Johnson graph contains other Johnson graphs as subgraphs, corresponding to the basic sets from 
the introduction. Indeed, for any $S \subset [n]$ of size 
smaller than $\ell$ we may consider the induced subgraph on 
all vertices $A$ such that $A\supseteq S$. We will often 
refer to such subgraphs as $|S|$-restrictions and denote it by $J|_S$. The 
motivation for this name is that for given a function $F\colon V(J)\to \mathbb{R}$ 
and $S$, we can define the restricted function 
$F|_{S}\colon V(J|_S)\to\mathbb{R}$ by $F|_{S}[A] = F[A\cup S]$. We will use the notation $\delta(F|_S)$ to denote $\E_{u \in J|_S}[F(u)]$.


\subsection{The Algorithm for Affine Unique Games over the Johnson Graph}
In this section, we describe our algorithm for Affine Unique-Games on $J(n,\ell,\alpha)$. Throguhout, when 
we say $q = \exp(r)$ we mean that there exist universal positive constants $c_1, c_2$ such that $2^{c_1 r} \leq q \leq 2^{c_2 r}$. Further if we set a parameter $q$ to be $\exp(r)$ we mean that we set $q = 2^{cr}$ for some universal constant $c$. We also have
an error parameter $\nu>0$ which we determine
in the end but should be thought of as small but bounded away from $0$.

\paragraph{High level description of the algorithm}
\begin{enumerate}
    \item First, our algorithm solves the degree $D$ sum of squares relaxation of program~\eqref{eq:ip} to find a pseudo-expectation $\pE_{\mu}$ corresponding to a pseudo-distribution 
    whose value is at least $1-\eps$.
    \item Next, using the manipulations that were described above, we produce a new one corresponding to two independent samples, 
    as in Sections~\ref{sec:ind_samples}.
    \item In the new pseudo-distribution, it holds that the expected objective function for both the samples $X$ and $X'$ is at least $1-\eps$, hence morally the fraction of edges that are simultaneously satisfied by $X$ and $X'$ is $1-2\eps$. We argue that in this case, there exists a subcube $C$ corresponding to a restriction of constant size on 
    which we have a large \emph{shift-partition potential} \footnote{The existence 
    of such a subcube is argued via global hypercontractivity,
    and to ensure this applies for our pseudodistribution 
    we have to make sure that the global hypercontractivity 
    result we used is proved within the sum-of-squares 
    system.}. We formally define this notion soon, but remark for now that it guarantees that a simple rounding procedure as described in the introduction manages to satisfy a constant fraction of the constraints inside the subcube.
    \item We use a subprocedure, that we refer to as SubRound, to find a subcube $C$ with large shift-partition potential and then an $\Omega(1)$-valued assignment for the vertices in the subcube $C$.
    \item We randomize the edges incident on the vertices of $C$ to get a new instance $I'$ (with possibly lower value) 
    and iterate the algorithm again. The goal of this step is 
    that in the next iteration, we will find a subcube $C'$ 
    which only has very small overlap with the vertices we have assigned thus far. Satisfying a large fraction of $C'$ in the same manner as we did for $C$ would then ensure that we have satisfied new edges and made progress. Iterating thus, we manage to satisfy an $\Omega(1)$-fraction of the graph. This kind of iteration was also used in~\cite{BBKSS}, but at each step they only managed to satisfy an $\Omega(1/\ell^r)$-fraction of the constraints in an $r$-subcube.
\end{enumerate}

\subsubsection{Shift-partition Potential~\cite{BBKSS}}
Our analysis uses the definition of the shift-partition potential from~\cite{BBKSS}, but we explain it here for completeness. 
Recall that given two assignments $X$ and 
$X'$ to an affine Unique-Games instance $I$, 
one may define the shift partition 
$F_s$ which consists of vertices $v$ for 
which $X(v) - X'(v) = s$. As explained 
in the introduction we will heavily use the shift partition, but for 
technical reasons we only want to work with vertices $u$ on which many of the constraints adjacent to them are satisfied. For an assignment $X$ and a vertex $v \in V(G)$, let ${\sf val}_v^G(X)$ denote the value of $v$, i.e. the fraction of the edges of $G$ that are incident on $v$ and are also satisfied by $X$. We drop the superscript $G$ when clear from context.

Naturally, 
this means that we want to include $v$ in $F_s$ only if $X(v) - X'(v) = s$ and ${\sf val}^G_v(X), {\sf val}^G_v(X')$ are at least 
somewhat large. To stay within the realm of
\sos though, we replace these indicator constraints with approximating polynomials 
as in Theorem~\ref{thm:step-approx}.

\begin{definition}[Shift-Partition~\cite{BBKSS}]
Given two affine UG assignments $X$ and $X'$ to $G$ define the following functions $\{F_s:V(G) \to \R\}_{s \in \Sigma}$ \begin{equation}
F_{s}^G(u) = \Ind(X_u - X'_u = s) \cdot p_{\beta,\nu}(\val^{G}_u(X))p_{\beta,\nu}(\val_u^{G}(X')),\label{eq:fs-j}
\end{equation}
where $p_{\beta,\nu}(x)$ is the degree-$\tilde{O}(1/\nu)$ polynomial which is a polynomial approximation to $\Ind[x \geq \beta]$ with accuracy parameter $\nu$ as in Theorem~\ref{thm:step-approx}. We drop the superscript $G$ in $F_s^G$ when the graph is clear from context.
\end{definition}



Equipped with the definition of shift partition, we can now define the shift-partition potential that governs 
our most basic rounding procedures inside 
subcubes.
\begin{definition}[Shift-Partition Potential~\cite{BBKSS}]\label{def:shift-partition}
The shift-partition size given two UG assignments $X$ and $X'$ on a graph $G$ is defined as:
\[\Phi^G_{\beta,\nu}(X,X') = \sum_{s \in \Sigma} \E_{u \in G}[F_s^G(u)]^2.\]
The shift-partition potential for a pseudodistribution $\cD$ over $(X,X')$ is defined as:
\[\Phi^G_{\beta,\nu}(\cD) = \pE_{(X,X') \sim \cD}[\Phi^G_{\beta,\nu}(X,X')].\]
\end{definition}
Intuitively, $\E_{u \in G}[F_s^G(u)]$ measures the 
fractional size of $F_s$ in the shift partition, 
hence the shift partition potential measures the 
collision probability of the shift partition. In
particular, it can be $\Omega(1)$ if and inly if there
is $s$ such that the part $F_s$ in the partition 
has constant density.

Throughout our arguments, we will consider the shift-partition potential with respect 
to the whole Johnson graph, as well as with 
respect to subcubes inside it. When we consider the shift-partition potential with respect to a subcube $J|_a$ we abbreviate it as $\Phi^a_{\beta,\nu}(\cD)$.

\subsubsection{The Main Routine}
We now give a formal description of our algorithm, modulo the procedure SubRound which we present later.
\begin{algorithm-thm}[Unique Games on the Johnson Graph]\label{alg:j} 
The input to the algorithm is an instance $I$ of affine Unique-Games over a Johnson graph $J(n,\ell,\alpha)$ such that 
${\sf val}(I) \geq 1-\eps$. The output of the algorithm is 
an assignment to $V(J)$.
\begin{enumerate}
\item Let $r =\left\lfloor\frac{64\sqrt{\eta}}{\alpha}\right\rfloor$, $\delta := \frac{\eps}{\exp(r)}$, $\gamma = \eps$ and $D = \frac{|\Sigma|^3\ell^{O(r)}}{\eps^3}$.
Fix $\cA_I$ to be the set of unique games axioms/ integer program over the instance $I$ (Program~\eqref{eq:ip}).

\item Set $j = 1, I_0 = I, R_0 = \phi$.
\item While $|R_{j-1}| < \frac{\gamma}{2}|V(J)|$ do the following:
\begin{compactenum}
\item Solve the degree-$D$ \sos SDP relaxation for the integer program $\cA_{I_{j-1}}$ and apply the transformation 
from Sections~\ref{sec:ind_samples} to get a pseudo-expectation corresponding to two independent samples. Let $\mu_j \times \mu_j$ be the corresponding pseudo-distribution.
\item Find $r' \le r$ and an $r'$-restricted subcube $C_j$ with high SubRound-value\footnote{This is the value of the assignment returned by the subroutine SubRound in~\eqref{alg:j-partial}.}, namely such that $\text{SubRound-val}_{\mu_j}(C_j) \ge \delta$. 

\item Run SubRound (Algorithm~\ref{alg:j-partial}) on $C_j$ with pseudodistribution $\mu_{j}$ to get an assignment $f_j$ to $V(C_j)$.

\item Let $S_j$ be a subgraph of $C_j$ induced by the set of vertices that have not been previously assigned by any partial assignment $f_{k}, k < j$ and assign them using $f_j$. Set $R_j = R_{j-1} \cup H_j$. 

\item Choose random affine-constraints on edges incident on $V(R_j)$ and let the new instance by $I_j$. 

\item Increment $j$.
\end{compactenum}
\item Output any assignment to $V(J)$ that agrees with all partial assignments $f_j$ (assigned to $S_j$) considered above.
\end{enumerate}
\end{algorithm-thm}

\subsubsection{The Subroutine SubRound and Condition\&Round}
Before describing the subroutine SubRound, 
we present several notations that are necessary for the analysis. For a vertex $u\in V(G)$ we will introduce new auxiliary Boolean variables $p_u$ and $p'_u$.  Morally speaking, $p_u$ are the indicators
that many of the constraints adjacent to $u$ are satisfied by the assignment $X$, but formally 
we approximate it by polynomials. 

First note that given a valid degree $D$ pseudodistribution $\cD$ one can extract valid and consistent local distributions over any $D$ variables from $X$, for e.g. looking at the marginal of $(X_u,X_v)$ in $\cD$ we get a distribution over $\Sigma \times \Sigma$. Similarly for any indicator variables $R$ in $\cD$ (where $R_u = R_u^2$ under the axioms of $\cD$) one could analogously define local distributions over $D$ variables from $X,X',R,R'$. Unfortunately the indicator that the value of $u$ is large is not a polynomial. Nevertheless using $\cD$ we can extract the following collection of local distributions over $X,X',p,p'$, where $p$ is a new set of variables we introduce and are not present in $\mu$. We denote this collection by $\cL(\cD, p_{\beta,\nu})$.

Formally, given a pseudodistribution $\cD$ of degree $D > \tO(1/\nu)$ over $(X,X')$ assignments to $\cA_I$ so that the corresponding pseudo-expectation has a high value, 
we define the local distributions $\cL(\cD, p_{\beta,\nu})$ over the variables of $p,p', X, X'$ in the following way:
\begin{definition}[Collection of Local Distributions $\cL(\cD,p_{\beta,\nu})$]\label{defn:local-dist}
Let $\cD$ be a pseudodistribution of degree $D\geq \tilde{\Omega}(d/\nu)$ over $(X,X')$ assignments to $\cA_{I}$, and let the polynomial $p_{\beta,\nu}$ be from Theorem~\ref{thm:step-approx} for $\eps=\delta = \nu$ 
and $\alpha = \beta$.
We define joint distributions over collections of $d$ variables from $p,p',X,X'$ as:
\[\pPr_{{\cL}}[p_u = 1] := \pE_{\cD}[p_{\beta,\nu}(\val_u(X))], ~~~ \pPr_{\cL}[p_u = 0] = 1- \pE_{\cD}[p_{\beta,\nu}(\val_u(X))],~~~ \pPr_{\cL}[X_u = \sigma] = \pE[X_{u,\sigma}],\]
and analogously for $X'_u,p'_u$. More generally we can extend the above definition to define the probability of conjunction of events in the variables $X,X',p,p'$. For all subsets $S,T,A,B \subseteq [N]$ with $|S|+|T|+|A|+|B| \leq d$, $b_1,\ldots,b_{|S|},b'_1,\ldots,b'_{|T|} \in \{0,1\}$ and $\sigma_1,\ldots,\sigma_{|A|},\sigma'_1,\ldots,\sigma'_{|B|} \in \Sigma$ define:
\begin{align*}
&\pPr_{{\cL}}\left[\bigcap_{u \in S} p_u = b_u, \bigcap_{v \in T} p'_v = b'_v, \bigcap_{a \in A} X_{a} = \sigma_a,\bigcap_{b \in B} X'_{b} = \sigma'_b\right] \\= &\pE_{\cD}\left[\prod_{u \in S} (1-b_u + (-1)^{1-b_u}p_{\beta,\nu}(\val_u(X))) \prod_{v \in T} (1-b_v + (-1)^{1-b_v}p_{\beta,\nu}(\val_v(X'))) \prod_{a \in A}X_{a,\sigma_a} \prod_{b \in B}X'_{b,\sigma_b}\right].
\end{align*}
\end{definition}
\begin{remark}
Note that since $\cD$ has large enough degree and $\cA_I \vdash_{\tO(1/\nu)} p_{\beta,\nu}(\val_u(X)) \in [0,1]$, it is easy to check that the above collection of local distributions form valid probability distributions over $d$ variables at a time, and are consistent with each other.
\end{remark}

Using the variables $p,p'$ we further define auxiliary random variables $Y_{u,v}$ and $Y'_{u,v}$ for $u \neq v \in [N]$ that will be particularly useful in the description and analysis of our algorithm: 
\begin{equation}
Y_{u,v} = (X_u,X_v,p_u,p_v), ~~~~ Y'_{u,v} = (X'_u,X'_v,p'_u,p'_v).
\end{equation}
Additionally, analogous to the notation $Y_{u,v}$, for all $S = \{i_1,\ldots,i_{|S|}\} \subseteq [N]$, define the random variables $Y_S, Y'_S$ as: 
\[Y_S = (X_{i_1}, \ldots, X_{i_{|S|}}, p_{i_1}, \ldots, p_{i_{|S|}}), ~~~ Y'_S = (X'_{i_1}, \ldots, X'_{i_{|S|}}, p'_{i_1}, \ldots, p'_{i_{|S|}}).\]
From Definition~\ref{defn:local-dist} we know that the event $Y_S = y_S$ naturally corresponds to an \sos polynomial $Q_S$ such that $\pPr_{{\cL}(\mu \times \mu, p_{\beta,\nu})}[Y_S = y_S] = \pE_{\mu}[Q_S]$ and similarly for events on $Y'_S$. Therefore we will use the notation ``condition $\mu$ on $Y_S = y_S$'' to mean reweighting $\mu$ on the corresponding polynomial $Q_S$. We will frequently look at quantities like $\pPr_{\cL}[\Ind[Y_S = y_S] \wedge \Ind[Y_T = y_T]]$ and $I(Y_S;Y_T)$ where $|S \cup T| \leq d$. In such places we emphasize that we're thinking of the local distribution on $Y_{S \cup T}$ and computing probabilities or mutual information on this distribution.

\skipi
\paragraph{The Subroutine Condition\&Round.}
To formally state SubRound, we first need to present
the subroutine Condition\&Round from~\cite{BBKSS} which is similar to 
the basic rounding procedure from the introduction.
\begin{algorithm-thm}[Condition\&Round~\cite{BBKSS}]\label{alg:low-ent-restated} ~\\ 
\textbf{Input:} A degree-$D$ (for $D \geq 2$) 
shift-symmetric pseudodistribution
$\mu$ for a UG  instance $I'=(G=(V,E),\Pi)$ over alphabet $\Sigma$.  \\
\textbf{Output:} Returns an assignment $x \in \Sigma^V$. 
\\\\
Sample a random solution $Z$:
\begin{compactenum}
\item Sample a vertex $u \sim \pi$ and condition on $X_{u} = 0$ to obtain the new marginals $\pE_\mu[\cdot ~|~ X_{u} = 0]$.
\item Sample a solution $Z$ by choosing each collapsed variable's labels independently according to its marginals: $Z_v \sim \pE_\mu[X_v ~|~ X_{u} = 0]$.
\end{compactenum}
\end{algorithm-thm}
\paragraph{Derandomized Condition\&Round.}
For future reference, it will be more convenient for us to analyze
a derandomized version of Condition\&Round, which makes the 
above procedure deterministic using standard methods such as 
the method of conditional expectations. We omit the straightforward details, and for us it suffices that this derandomized 
procedure always achieves value which is at least the 
expected value of what Condition\&Round gives.
\skipi
We can now describe SubRound that obtains a partial assignment for the Johnson graph. 

\paragraph{High level description of SubRound.}
In this routine, we are given a pseudoexpectation and a 
subcube $J|_a$ (on which we hope the value is large), and we wish to 
find a good assignment to the nodes in the given subcube, 
or quit. This method successfully assigns values to $V(J|_a)$ 
in the case that the shift-partition potential function on the subcube is large.

Towards this end, we find an appropriate conditioning on the random variables $Y_{u,v}$ so that after performing it, 
they are nearly pairwise independent (we use mutual information 
to denote this; this is the step that we referred to in 
the introduction as eliminating global correlations). Denote the new distribution obtained as $\mu_1 \times \mu_2$.

Simultaneously we find an event 
(which we state 
as a polynomial to stay within the realm of \sos) $P_a(X,X')$, which roughly says that 
the part $\{v~:~X(v)-X'(v) = s\}$ is dense inside $J|_a$, 
such that this event has significant probability 
and conditioned on it the shift-partition inside $J|_a$
is large, where both these facts are measured with respect to $\mu_1 \times \mu_2$. We give a formal definition of $P_a(X,X')$ in Lemma~\ref{lem:sp-j}.

Finally, we run the procedure Condition\&Round using either the symmetrized versions of $\mu_1$ or $\mu_2$ which 
returns a good assignment to the given subcube provided that the previous steps succeeded.

\skipi
We now move on to the formal description of SubRound.
\begin{algorithm-thm}[SubRound: Rounding for a subgraph]\label{alg:j-partial} ~\\
\textbf{Input:} Takes as input an affine Unique-Games instance $I = (G, \Pi, \Sigma)$ over an $(n,\ell,\alpha\ell)$-Johnson graph $G$ and alphabet $\Sigma$, a pseudo-distribution $\mu$ with $\val_\mu(I) \geq 1- \eta$, a $j$-restricted subcube $J|_a$, with $j \leq r = \left\lfloor\frac{32\sqrt{\eta}}{\alpha}\right\rfloor$. 

\noindent\textbf{Output:} Returns an assignment in $\Sigma^{|V(J|_a)|}$ to the vertices of $J|_a$. 
\begin{enumerate}
\item Fix $t = \frac{|\Sigma|^3\ell^{\Theta(r)}}{\eta^3}, \tau = \Theta(\frac{\eta}{|\Sigma|^2\ell^{2r}\exp(r)}), \beta = \sqrt{\eta}, \nu = \eta\exp(-r)$.
\item Find a polynomial $P_a(X,X')$ from the set of polynomials given by Lemma~\ref{lem:sp-j}, subsets $A, B \subseteq V(J|_a)$ of size at most $t$ and strings $y_A, y_B$ so that conditioning $\mu$ on the events $Y_A = y_A$ and $Y_B = y_B$ gives pseudodistributions $\mu_1, \mu_2$ satisfying:
\begin{enumerate}
\item 
Shift partition potential inside $J|_{a}$ is significant: 
\[
\pE_{\mu_1 \times \mu_2}[\Phi^a_{\beta,\nu}(X,X')P_a(X,X')] \geq \exp(-r) \pE[P_a(X,X')].
\]
\item The probability of $P_a$ is significant:
\[
\pE_{\mu_1 \times \mu_2}[P_a(X,X')] \geq \Omega\left(\frac{\sqrt{\eta}}{|\Sigma|\ell^{j}\exp(r)}\right).
\]
\item The random variables $Y_{u,v}$ have small global correlation inside $J|_a$: 
\[\E_{\substack{u_1,v_1 \sim S: u_1 \neq v_1\\ u_2,v_2 \sim S: u_2 \neq v_2}}[I(Y_{u_1,v_1};Y_{u_2,v_2})] \leq \tau.
\]
\item The random variables $Y_{u,v}'$ have small global correlation inside $J|_a$: \[
\E_{\substack{u_1,v_1 \sim S: u_1 \neq v_1\\ u_2,v_2 \sim S: u_2 \neq v_2}}[I(Y'_{u_1,v_1};Y'_{u_2,v_2})] \leq \tau.
\]
\end{enumerate}
Here $S = V(J|_a)$, and the mutual information is taken with respect to the collection of local distributions $\cL(\mu_1 \times \mu_2,p_{\beta,\nu})$. If no such polynomial or conditioning exists then output the all zeros assignment on $V(J|_a)$ and exit.
\item Symmetrize $\mu_1, \mu_2$ using the procedure in Section~\ref{sec:symm-prelims} to obtain $\mu^{\sym}_1, \mu^{\sym}_2$. Perform derandomized Condition\&Round on $\mu^{\sym}_1$ or $\mu^{\sym}_2$ and output the higher valued assignment obtained for $V(J|_a)$.
\end{enumerate}
\end{algorithm-thm}

\subsection{Analysis of Algorithm~\ref{alg:j}}
\subsubsection{High Level Description of the Analysis}
We prove that Algorithm~\ref{alg:j} returns a solution with value independent of the alphabet size. We begin by explaining the ideas of the analysis~\cite{BBKSS} for small-set expanders, and then 
explain how this analysis is adapted to our case.

\paragraph{Small-set expanders~\cite{BBKSS}:} 
Consider the shift-partition potential, 
and suppose for simplicity of presentation that the $p$'s are replaced with the actual indicator functions. 
Then the functions $F_s$ cover 
almost all the vertices of the graph, 
and any edge that goes across parts in the shift-partition must be violated by either $X$ or $X'$. Hence if $X,X'$ were $(1-\eps)$-satisfying assignments then the shift-partition is non-expanding, in the sense that at most $2\eps$-fraction of the edges go across the parts (those that are violated 
by either $X$ or $X'$). 

In the case of $(2\eps,\delta)$-small-set expanders,~\cite{BBKSS} then conclude that the fractional-size of one of the parts must be at least $\Omega(\delta)$, and in fact the shift-partition size: $\sum_s \E_{u \sim G}[F_s(u)]^2$, must be $\geq \Omega(\delta)$. On the other hand they prove that if the shift-partition potential is large on $\mu \times \mu$ then Condition\&Round on $\mu$ succeeds in rounding to an $\Omega_{\eps,\delta}(1)$-assignment, 
and formally they show:
\begin{enumerate}
\item Given a degree $D$ certificate of $(O(\eps),\delta)$-small-set expansion for a graph $G$ and a degree $\Omega(D)$ pseudodistribution $\mu$ with value $1-\eps$ on any affine UG instance $I = (G,\Pi)$, conclude that $\Phi^G_{\beta,\nu}(\mu \times \mu) \geq \delta$.
\item Show that if $\Phi^G_{\beta,\nu}(\mu \times \mu)$ is large then Condition\&Round outputs an assignment with large value.
\end{enumerate}

\paragraph{Our analysis for Johnson graphs:} 
Let us see how to adapt this analysis to the case of Johnson graphs. Step (1) is far from being true since Johnson graphs are not small-set expanders. But we know that all non-expanding sets in these graphs must have large size when restricted to some subcube, although the set might not be large in the whole graph. We use this property to conclude that the shift-partition potential must be large on a subcube. This part of the analysis turns out to be much more non-trivial than the certifiable-SSE case and is where the bulk of our technical work lies. 

Once we have this we use a suitably modified version of step (2) to conclude that Condition\&Round when applied to the subcube outputs an $\Omega(1)$-assignment to the subcube. We then use an iteration lemma to run this algorithm multiple times to output an $\Omega(1)$-satisfying assignment to the whole graph. 


\subsubsection{Lemmas to be Proven Later}
In this section, we state a few lemmas that are necessary for the analysis of 
Algorithm~\ref{alg:j}, whose proofs are deferred to later sections. In Section~\ref{sec:combine_analysis} we 
use these lemmas to analyze the performance 
of Algorithm~\ref{alg:j}.
\skipi

The analysis of Algorithm~\ref{alg:j} proceeds by showing that given a pseudodistribution for $I$ with high value, there exists a subcube with high SubRound value, i.e. showing the success of Step 3(a). To do so, we first prove that the structure theorem for Johnson graphs (Theorem~\ref{thm:structure-johnson}) implies that given a pseudodistribution $\mu$ with large value, there exists a subcube $J|_a$ with large shift-partition potential denoted by $\Phi^a_{\beta,\nu}$. Since the Johnson graph is not a small-set expander though, we have the following more subtle statement: There exists an $r = O(\sqrt{\eps}/\alpha)$-restricted subcube $a$ and an event $P_a$ that has large probability, such that if we condition $\mu \times \mu$ on $P_a$, the induced shift-partition potential is large on the subcube $J|_a$. The event $P_a$ roughly corresponds to the indicator that the size of one of the shift-partition components is large inside the subcube $J|_a$, i.e. $\Ind(\delta(F_s|_a) > \exp(-r))$ and has probability $\geq \ell^{-O(r)}$.



\begin{lemma}\label{lem:sp-j}
There exists a constant $\eps_0 \in (0,1)$, such that for all positive constants $\eps \leq \eps_0$, $\alpha, \tau \leq 1$, $\nu \leq \eps \exp(-r)$, all integers $\ell \geq \Omega(r), n \geq \ell$, where $r = \lfloor \frac{32\sqrt{\eps}}{\alpha}\rfloor$ the following holds. 

Let $I$ be an affine UG instance on $J(n,\ell,\alpha\ell)$ and $\mu$ be a pseudodistribution over assignments for $I$ with $\val_\mu(I) \geq 1-\eps$ and degree at least $D = \widetilde{\Omega}\left(\frac{|\Sigma|\ell^r\exp(r)}{\nu\tau\sqrt{\eps}}\right)$. Then there exists a restriction $a \subseteq [n]$ of size $j \leq r$, a degree $\widetilde{O}(1/\nu)$ polynomial $P_a(X,X')$ from a fixed set of $|\Sigma|$ polynomials, subsets $A, B \subseteq V(J|_a)$ of size at most $\tO(\frac{|\Sigma|\ell^j\exp(r)}{\tau\sqrt{\eps}})$ and strings $y_A, y_B$ such that conditioning $\mu$ on the events $Y_A = y_A$ and $Y_B = y_B$ gives degree $\widetilde{\Omega}(1/\nu)$ pseudodistributions $\mu_1$ and $\mu_2$ such that:
\begin{enumerate}
\item $\cA_I  ~~\vdash_{\widetilde{O}(1/\nu)} ~~ P_a(X,X') \in [0,1]$.
\item $\pE_{\mu_1 \times \mu_2}[\Phi^a_{\sqrt{\eps},\nu}(X,X')P_a(X,X')] \geq \exp(-r) \pE[P_a(X,X')]$.
\item $\pE_{\mu_1 \times \mu_2}[P_a(X,X')] \geq \Omega\left(\frac{\sqrt{\eps}}{|\Sigma|\ell^{j}\exp(r)}\right)$.
\item $\E_{\substack{u_1,v_1 \sim S: u_1 \neq v_1\\ u_2,v_2 \sim S: u_2 \neq v_2}}[I(Y_{u_1,v_1};Y_{u_2,v_2})] \leq \tau.$
\item $\E_{\substack{u_1,v_1 \sim S: u_1 \neq v_1\\ u_2,v_2 \sim S: u_2 \neq v_2}}[I(Y'_{u_1,v_1};Y'_{u_2,v_2})] \leq \tau,$
\end{enumerate}
where $S = V(J|_a)$, $Y_{u,v} = (X_u,X_v,p_u,p_v)$, $Y'_{u,v} = (X'_u,X'_v,p'_u,p'_v)$ and the mutual information is taken with respect to the collection of local distributions $\cL(\mu_1 \times \mu_2,p_{\sqrt{\eps},\nu})$. 
\end{lemma}


The next lemma captures the intuition that 
given a collection of random variables $Z_{u,v}$ that are roughly pairwise independent (that are functions of assignments $X$ and $X'$), and an event $E(X,X')$ with significant 
probability, then typically $Z_{u,v}~|~E$ 
is close to $Z_{u,v}$ in statistical distance. Informally, this asserts that 
conditioning on the event $E$ does not 
change the distributions of $(X, X')$ locally.

\begin{lemma}\label{lem:correlation}
For all $\tau,p,\beta,\nu,\delta \in (0,1)$, integers $D,N$, $m \leq N$ the following holds: Let $\mu_1 \times \mu_2$ be a degree $D+\widetilde{\Omega}(1/\nu)$ pseudodistribution over $(X,X')$ satisfying $\cA_I(X,X')$ and $E(X,X')$ be a polynomial such that, $\cA_I \vdash_D E(X,X') \in [0,1]$ and $\pE_{\mu_1 \times \mu_2}[E(X,X')] \geq p$. Suppose for $S \subseteq [N]$ we have that, 
\begin{enumerate}
\item $\E_{\substack{u_1,v_1 \sim S: u_1 \neq v_1\\ u_2,v_2 \sim S: u_2 \neq v_2}}[I(Y_{u_1,v_1};Y_{u_2,v_2})] \leq \tau.$
\item $\E_{\substack{u_1,v_1 \sim S: u_1 \neq v_1\\ u_2,v_2 \sim S: u_2 \neq v_2}}[I(Y'_{u_1,v_1};Y'_{u_2,v_2})] \leq \tau$,
\end{enumerate}
where $Y_{u,v} = (X_u,X_v,p_u,p_v), Y'_{u,v} = (X'_u,X'_v,p'_u,p'_v)$ and the mutual information is with respect to the collection of local distributions $\cL(\mu_1 \times \mu_2,p_{\beta,\nu})$ (Definition~\ref{defn:local-dist}). Then we have that, 
\[\Pr_{u,v \sim S: u \neq v}[ TV((Y_{u,v},Y'_{u,v})|E , (Y_{u,v},Y'_{u,v})) \geq \delta] \leq O\left(\frac{\sqrt{\tau}+ 1/|S|}{p\delta^2}\right),\]
where the distribution $(Y_{u,v},Y'_{u,v})|E$ refers to the joint distribution on these variables defined by the collection of local distributions $\cL(\mu_1 \times \mu_2|E, p_{\beta,\nu})$ and similarly $(Y_{u,v},Y'_{u,v})$ refers to the distribution defined by $\cL(\mu_1 \times \mu_2,p_{\beta,\nu})$.
\end{lemma}


We now restrict our attention to the nice subcube $J|_a$ obtained from Lemma~\ref{lem:sp-j} and we henceforth only care about the relevant pseudodistributions (e.g. $\mu_1 \times \mu_2|P_a$) when marginalized to $J|_a$. For simplicity of notation we still refer to the marginalized pseudodistributions using their original notation.

We show that large shift-potential $\Phi^a(\mu_1 \times \mu_2 | P_a) \geq \delta$ with the additional property that an average $(Y_{u,v}, Y'_{u,v})$-pair when drawn from $\mu_1 \times \mu_2 |P_a$ is close to its distribution when drawn from $\mu_1 \times \mu_2$, implies that Condition\&Round (applied on $J|_a$) succeeds on either $\mu_1$ or $\mu_2$.

\begin{lemma}\label{lem:round-j}
Let $I = (G, \Pi)$ be an affine instance of Unique Games over the alphabet $\Sigma$ and $\mu_1 \times \mu_2$ be a degree $D+\widetilde{O}(1/\nu)$ pseudodistribution over assignments $(X,X')$ to $I$. Let $E(X,X')$ be a polynomial such that $\cA_I(X,X') \vdash_D E(X,X') \in [0,1]$. Suppose we have that: 
\[\Pr_{u,v \sim V(G): u \neq v}[TV((Y_{u,v},Y'_{u,v})|E , (Y_{u,v},Y'_{u,v})) > \delta] \leq \zeta,\]
where the distribution $(Y_{u,v},Y'_{u,v})|E$ refers to the joint distribution on these variables defined by the collection of local distributions $\cL(\mu_1 \times \mu_2|E, p_{\beta,\nu})$ and similarly $(Y_{u,v},Y'_{u,v})$ refers to the distribution defined by $\cL(\mu_1 \times \mu_2,p_{\beta,\nu})$ (Definition~\ref{defn:local-dist}). 

If $\Phi^G_{\beta,\nu}(\mu_1 \times \mu_2 | E(X,X')) \geq \gamma$, then on at least one of the pseudodistributions $\mu^\sym_1$ or $\mu^\sym_2$ Algorithm~\ref{alg:low-ent-restated} runs in time $\poly(|V(G)|)$ and returns an assignment of expected value at least 
\[
(\beta - \nu)^2\left(\gamma - O(\delta+\zeta) - \frac{1}{|V(G)|}\right) - 3\nu(\beta - \nu)
\]
for $I$.
\end{lemma}

Using the lemmas above for $G = J|_a$ one can analyze a single iteration of SubRound and show that it manages to assign a subcube of $I$ and within it satisfy a constant fraction of the edges. This subcube might be of size $o(1)$ of the whole graph though, therefore to complete the analysis of our algorithm we need the following lemma, asserting that we can iterate this procedure (as done in Algorithm~\ref{alg:j}) to satisfy a constant fraction of constraints of $I$. 

\begin{lemma}\label{lem:subroutine}
Let $c,\gamma,\delta \in (0,1]$ and $r \in \nn$. Let $I$ be an affine UG instance on alphabet $\Sigma$ on $J_{n,\ell,\alpha}$ with $\ell, n$ large enough  and value at least $c$. 
Suppose we have a subroutine $\cA$ which given as input any affine UG instance $I'$ on $J_{n,\ell,\alpha}$ with $\val_\mu(I') \ge c-\gamma$, returns an $\leq r$-restricted subcube $H$ on $J$ and a partial assignment $f$ such that, $\text{val}_{f}(H) \geq \delta$.
Then if $\cA$ runs in time $T(\cA)$, there is a $O(|V(J)|T(\cA)+|V(J)|^3)$-time algorithm which finds a solution for $I$ that satisfies an $\Omega(\delta\gamma(1-\alpha)^r)$-fraction of the  edges of $J$. 
\end{lemma}
\begin{proof}
Deferred to Section~\ref{sec:pf_it}.
\end{proof}




\subsubsection{The Analysis of Algorithm~\ref{alg:j}}\label{sec:combine_analysis}

\begin{theorem}\label{thm:main-johnson}
There exists a constant $\eps_0 \in (0,1)$, such that for all positive constants $\eps \leq \eps_0$, $\alpha \in \Q$ with $\alpha \leq 1/2$, all integers $\ell \geq \Omega(r)$ with $\alpha \ell \in \N$ and $r = \Theta(\sqrt{\eps}/\alpha)$ and $n$ large enough, Algorithm~\ref{alg:j} has the following guarantee: If $I$ is an instance of affine Unique Games on the $(n,\ell,\alpha)$-Johnson graph $G$ with alphabet $\Sigma$ and $\val(I) = 1 - \eps$, then in time $|V(J)|^{\poly(\ell^r,|\Sigma|,1/\eps)}$ 
Algorithm~\ref{alg:j} returns an $\frac{\poly(\eps)}{\exp(r)}$-satisfying assignment for $I$.
\end{theorem}

\begin{proof}
Let $r = \lfloor 32\eps/\alpha\rfloor$. Given the $\leq r$-restricted subcube $a \subseteq [n], P_a(X,X')$ and $\mu_1 \times \mu_2$ from Lemma~\ref{lem:sp-j} we have that $\Phi_{\sqrt{\eps},\nu}^a(\mu_1 \times \mu_2| P_a) \geq \exp(-r)$, therefore we apply Lemma~\ref{lem:correlation} with $S = J|_a$, the polynomial  $E(X,X') = P_a(X,X')$ and the parameters $\beta = \sqrt{\eps}, \nu = \exp(-r)\eps, \delta = \exp(-r), p = \Omega\left(\frac{\sqrt{\eps}}{|\Sigma|\ell^{r}\exp(r)}\right)$ and $\tau = \Theta(\frac{\eps}{|\Sigma|^2\ell^{2r}\exp(r)})$. The parameter $\tau$ has been chosen so that $O(\frac{\sqrt{\tau}+1/|S|}{p\delta^2}) = \exp(-r)$ and therefore we can apply the rounding lemma (Lemma~\ref{lem:round-j}) with $\zeta = \exp(-r)$ and $\gamma = \exp(-r)$ and the same settings of $\beta, \nu, \delta$. This shows that there exists a subcube $a \in {[n] \choose j}$ with SubRound value that is at least $\Omega(\eps/\exp(r))$ if $\deg(\mu) \geq \widetilde{\Omega}\left(\frac{|\Sigma|\ell^r\exp(r)}{\nu\tau\sqrt{\eps}}\right)$. Hence it suffices to have degree of $\mu$ equal to $\frac{|\Sigma|^3\ell^{O(r)}}{\eps^3}$.

By using SubRound as a subroutine, we finish the proof of this theorem by applying the iteration Lemma~\ref{lem:subroutine} with $c=1-\eps$, $\gamma = \eps$, $\delta = \Omega(\eps/\exp(r))$ and $r = \lfloor 64\sqrt{\eps}/\alpha \rfloor$ so that $(1-\alpha)^r = O(\sqrt{\eps})$ (Claim~\ref{claim:subcube-expanse}). We can check that SubRound satisfies the hypotheses of the lemma: it finds a $\le r$-restricted subcube and an assignment to it with value $\Omega(\delta)$ in time $|V(J)|^{\frac{|\Sigma|^3\ell^{O(r)}}{\eps^3}}$ (this follows from the degree upper bound on $\mu$).
This gives us that Algorithm~\ref{alg:j} outputs an assignment of value at least $\Omega(\eps^3/\exp(r))$ in time $|V(G)|^{\frac{|\Sigma|^3\ell^{O(r)}}{\eps^3}}$.
\end{proof}

\subsection{Proof of Lemma~\ref{lem:sp-j}: Finding Subcube with Large Shift-partition}\label{sec:large-potential}

In this section we use the structure theorem for Johnson graphs
to prove the existence of a restriction $a \subseteq [n]$ of constant size (corresponding to the induced subgraph $J|_a$) with large induced shift-partition potential when conditioned on a non-negative polynomial $P_a$, i.e. $\Phi^a_{\beta,\nu}(\mu \times \mu | P_a)$ is large. Recall that to analyze our final Condition\&Round step though we need to show conditioning $\mu \times \mu$ on $P_a$ does not introduce too many correlations between $X, X'$. Therefore we first perform a global correlation reduction procedure on $\mu \times \mu$ to get the product pseudodistribution $\mu_1 \times \mu_2$ while preserving the property that $\Phi^a_{\beta,\nu}(\mu_1 \times \mu_2 | P_a)$ is large.
The exact quantity that aids our rounding analysis is: 
\[I(Y_{u_1,v_1};Y_{u_2,v_2}) =  I(X_{u_1},X_{v_1},p_{u_1},p_{v_1};X_{u_2},X_{v_2},p_{u_2},p_{v_2}),\] 
defined according to the collection of local distributions in Definition~\ref{defn:local-dist}. We upper bound this by using the global correlation reduction procedure of Raghavendra-Tan~\cite{RT12}. 

\subsubsection{Restricting the Shift Partition to a Subgraph}
To prove Lemma~\ref{lem:sp-j} we will need the notion of
\emph{global shift-partition potential restricted to subcube $J|_a$}, which is almost the same as shift-partition potential on $J|_a$ except for one key difference:
\begin{definition}[Global shift-potential restricted to Subgraphs~\cite{BBKSS}]\label{def:res-shift}
For any $\nu,\beta \in (0,1)$ and subgraph $H$ of $G$, define the {\em global shift-partition size restricted to the subgraph $H$} to be the quantity:
\[
\Phi_{\beta,\nu}(X,X')|_H = \sum_{s \in \Sigma} \E_{u \in H} [F_s^{G}(u)]^2,
\]
where $F_s$'s are the functions defining the shift-partition (Definition~\ref{def:shift-partition}). 
Let the global shift-partition potential with respect to a pseudodistribution $\cD$ over pairs of assignments $(X,X')$ to $G$ be:
\[
\Phi_{\beta,\nu}(\cD)|_H = \pE_{(X,X') \sim \cD}[\Phi_{\beta,\nu}(X,X')|_H].
\]
\end{definition}

Note the difference between the global potential and the shift-partition potential on $H$: the global shift-partition potential $\Phi_{\beta,\nu}(X,X')|_{H}$ measures the size of the global partition inside $H$, i.e. $\val_u^G(X)$ is a function of all the edges of the graph that are incident on $u$, not just the edges in the subgraph $H$. The potential 
$\Phi_{\beta,\nu}^{H}(X,X')$, on 
the other hand, measures the value of a vertex $u$ only \emph{inside} the subgraph $H$.

\subsubsection{Global Hypercontractivity}
We will need the following definition to describe the structure theorem for Johnson graphs:
\begin{definition}[Restrictions of Functions]
For the $(n,\ell,\alpha)$-Johnson graph $J$, given a function $F:V(J) \rightarrow \R$ and a set $a \subseteq [n]$ with $a = r$, such that $0 \leq r \leq \ell-1$, we define the restricted function $F|_a: \binom{[n] \setminus a}{\ell - r} \rightarrow \R$ as, 
$$F|_a(X) = F(a \cup X).$$
Further, let $\delta(F|_a)$ denote the fractional size of the function restricted to the subcube $J|_a$, that is,
$$\delta(F|_a) := \E\limits_{X \sim \binom{[n] \setminus a}{\ell - r}}[F|_a(X)].$$
When $a = \phi$ and $r = 0$, we have that $F|_a(X) = F(X)$ for all $X \in \binom{[n]}{\ell}$ and $\delta(F|_a) = \delta(F) = \E[F]$.
\end{definition}

In~\cite{KMMS} it is shown that pseudorandom sets expand. Formally, we define pseudorandom sets and in general pseudorandom functions as follows:
\begin{definition}[Pseudorandom functions]\label{def:p.r.}
A set $S$ is called $(r,\gamma)$-pseudorandom if for all $\leq r$-restrictions $a$, $\delta(S|_a) \leq \gamma$. Similarly a function $F$ bounded in $[0,1]$ is called $(r,\gamma)$-pseudorandom if for all $\leq r$-restrictions $a$, $\delta(F^2|_a) \leq \gamma$.
\end{definition}

We will need a version of this 
result for general Boolean and pseudorandom functions $F$, and moreover to show that it is proved in the 
\sos proof system of constant degree. We get an \sos proof that shows that if $\delta(F^2|_a) \leq \gamma$ for all $a$ of size $\leq r$ and $F$ is Boolean, then $\ip{F,LF}$ is large, with $q_a(F)$ being the SoS multipliers of the axioms $\{\gamma - \delta(F^2|_a)\}_a$ and $(\Pi_{\geq \lambda_r}F)(X)$ being the multiplier for the axioms $\{F(X)^3 = F(X)\}$. We have the following statement: 


\begin{theorem}[Expansion Theorem for Johnson Graphs]\label{thm:structure-johnson}
For all $\alpha \in (0,1)$, all integers $\ell \geq 1/\alpha$ and $n \geq \ell$, the following holds: Let $J(n,\ell,\alpha \ell)$ be the $\alpha$-noisy Johnson graph. For every constant $\gamma \in (0,1)$ and positive integer $r \leq O(\ell)$, every function $F: V(J) \rightarrow \R$ that is $(r,\gamma)$-pseudorandom has high expansion:
\begin{align}
&\{F(X) \in [0,1]\}_{X \in V(J)}  \vdash_{O(1)}\, \notag\\
&\ip{F,LF} \geq \delta(F)(1 - (1-\alpha)^{r+1})(1- \gamma^{1/3}\exp(r)) - \sum_{j = 0}^r \frac{ c_j\ell^{j}}{\gamma}\E_{a \sim {[n] \choose j}}[q_a(F)({\delta(F^{2}|_a) - \gamma)}]+ B(F) \\
&\{F(X) \in [0,1]\} \,\,\, \vdash_{2} \,\,\, 0 \leq q_a(F) \leq 1 
\end{align}
where for all $j \leq r$, $c_j$'s are positive constants of size at most
$\exp(r)$, for all size $j$ subsets $a \subseteq [n]$, $q_a(F)$ are degree $2$  polynomials and $B(F)=\frac{4}{3}\E_X[(F^{3} - F)\Pi_{\geq \lambda_r}F ]$ for $\lambda_r = (1-\alpha)^r$ and $\Pi_{\geq \lambda}$ denoting the projection operator to the top-eigenspace of $J$ of eigenvalues $\geq \lambda$.
\end{theorem}
\begin{proof}
Deferred to Section~\ref{sec:fourier}.
\end{proof}

One can equivalently view the expansion result in~\cite{KMMS} as asserting that if $F$ is a Boolean function which is the
indicator of a set $\mathcal{F}\subseteq \binom{[n]}{\ell}$ whose edge 
expansion is bounded away from $1$, then $\delta(F|_{a})\geq \Omega(1)$ for 
some $r$-restricted subcube $a$ with constant $r$. In fact, if the expansion of $F$ is at most $1-\eps$ on $J_{n,\ell,\alpha}$, then $\delta(F|_a) \geq \exp(-r)$ for $r = O(\eps/\alpha)$. Similarly, we can use the above \sos statement to conclude that even if $F$ is ``almost-Boolean'' and non-expanding, then $q_a(F)(\delta(F^2|_a) - \gamma) > 0$ for some $O(1)$-restriction $a$. In the context of a pseudodistribution $\mu$ over non-expanding sets $F$, one can conclude that conditioning $\mu$ on $q_a(F)$ results in a new pseudodistribution $\mu'$ where $\delta(F^2|_a) > \gamma$ on average, and therefore $q_a(F)$ roughly corresponds to the ``event'' that $F$ is dense on $J|_a$.

We use the observation that the shift-partition defined with respect to $(1-\eps)$-satisfying assignments $X$ and $X'$ has expansion at most $2\eps$, therefore is non-expanding. Therefore as above, using the structure theorem we conclude that at least one of the sets in the shift-partition is not $(r,\gamma)$-pseudorandom for $r = O(\eps/\alpha)$ and $\gamma = \exp(-r)$, i.e. $\delta(F_s|_a) \geq \gamma$ for some $s \in \Sigma$ and $\leq r$-restriction $a$. Let $\viol(X)$ denote the fraction of edges that an assignment $X$ violates in the instance $I$. We frame this fact in \sos using the following lemma:

\begin{lemma}\label{lem:structure-fs-j}
Under the conditions of Lemma~\ref{lem:sp-j}, for all $\eta > 0$ and $r \leq \ell$ we get:
\begin{align}
\A_I  
&\vdash_{\tilde{O}(1/\nu)} \\ 
&\sum_{j = 0}^r c_j \exp(r) \ell^j \E_{a \in {[n] \choose j}}[\sum_{s}q_{a}(F_s)(\delta(F_s|_a) - \exp(-r))] \notag\\
&\geq \frac{1 - \lambda_{r+1}}{2} 
- (\viol(X)+\viol(X'))\left(1+\frac{3(1-\lambda_{r+1})+ 12 +8\eta}{6(1-\beta-\nu)}\right) - \frac{2}{3\eta} - \nu\left(5 - \lambda_{r+1} + \frac{8\eta}{3}\right)\label{eq:sp-j-1},
\end{align}
where $\lambda_{r+1} = (1-\alpha)^{r+1}$, $q_a$ are the polynomials in Theorem~\ref{thm:structure-johnson}, $\A_I$ is the set of axioms defined for $I$ by program (\ref{eq:ip}).
\end{lemma}
\begin{proof}
This proof proceeds exactly as the proof of Lemma 4.4 for certifiable small-set expanders in~\cite{BBKSS}:
we apply Theorem~\ref{thm:structure-johnson} to each function $F_s$ and sum up over $s$. Doing so, we get the following inequality:
\begin{align}
\sum_s \ip{F_s, LF_s} \geq &(1 - \lambda_{r+1})(1- \gamma^{1/3}\exp(r))\sum_s \delta(F_s) - \sum_{j = 0}^r \frac{c_j \ell^j}{\gamma} \E_{a \in [n]^j}[\sum_{s}q_{a}(F_s)(\delta(F_s^{2}|_a) - \gamma)] \notag\\
&+ \frac{4}{3}\sum_s \E[(F_s^3 - F_s)\Pi_{\geq \lambda_r} F_{s}] \label{eq:sp-intermediate}
\end{align}

We now set $\gamma = \exp(-r)$ such that $(1-\gamma^{1/3}\exp(r)) = 1/2$. Further we bound each of the terms $\sum_s \ip{F_s, LF_s}, \sum_s \delta(F_s)$ and the Booleanity error: $\frac{4}{3}\sum_s \E[(F_s^3 - F_s)\Pi_{\geq \lambda_r} F_{s}]$.

First note that the fraction of low-valued vertices is small and in particular is at most $\frac{\viol(X)+\viol(X')}{1-\beta}$. Since the functions cover all of the high-valued vertices we get $\sum \delta(F_s) \gtrsim 1 - \frac{\viol(X) + \viol(X')}{1-\beta}$. Next, $\sum_s \ip{F_s, LF_s}$ counts the fraction of edges crossing the shift-partition. Every such edge must be violated by either $X$ or $X'$ or must be incident on a vertex with low value, therefore we get:  $\sum_s \ip{F_s, LF_s} \lesssim \viol(X)+\viol(X')+ \frac{2\viol(X)+2\viol(X')}{1-\beta}$.  We have put approximate inequalities here since there are some error terms generated because $F_s$'s are not exact indicator functions. 
The Booleanity term $B(F)$ in Theorem~\ref{thm:structure-johnson} is also small (after summing up) because $F_s$'s are approximate-indicators (note that it is $0$ for $0/1$ functions since $F^3 = F$)\footnote{Note that if the functions $F_s$ were defined using $\Ind[\val_u(X)\geq \beta]$ then we would obtain the lemma statement without the terms involving $\eta$ and with $\nu = 0$}. These statements have been made formal in the claims from Section 4.2 of~\cite{BBKSS}, specifically Claim 4.2, 4.7 and 4.9 therein. They are simple to prove given the properties of the approximate-indicator polynomial $p_{\beta,\nu}$ hence we omit them here. 

Using the above claims to bound each sum, plugging in the bounds in \eqref{eq:sp-intermediate} and rearranging we get the desired inequality.

\end{proof}

\subsubsection{Reducing Global Correlation}
The last ingredient we need for the proof of Lemma~\ref{lem:sp-j} is 
that given a pseudo-expectation of sufficiently high degree, one can 
construct different pseudo-expectations (which are conditionals 
of the initial pseudo-expectation) that have no global correlations.
More precisely:
\begin{lemma}\label{lem:ragh-tan-appln}
For all $\tau,p,\beta,\nu \in (0,1)$ and $|\Sigma|, D, D' \in \Z$ 
such that $D' \geq \max(D,\widetilde{\Omega}(1/\nu))$ the following holds.

Suppose there is a degree $D' + \widetilde{\Omega}(\frac{\log |\Sigma|}{p \nu \tau})$ pseudodistribution $\mu$ over UG assignments that satisfies $\cA_I$, and a polynomial $E(X,X')$ satisfying $\pE_{(X,X') \sim \mu \times \mu}[E(X,X')] \geq p$ and $\cA_I \vdash_D E(X,X') \leq 1$. Then for all subsets $S \subseteq [N]$, there exist subsets $A, B \subseteq S$ of size at most $O(\frac{\log|\Sigma|}{p\tau})$ and strings $y_A, y_B$ such that conditioning $\mu$ on the events $Y_A = y_A$ and $Y_B = y_B$ gives pseudodistributions $\mu_1$ and $\mu_2$ of degree at least $D'$ such that: 
\begin{enumerate}
\item $\pE_{X,X'\sim \mu_1\times\mu_2}[E(X,X')] \geq \frac{p}{2}$.
\item $\E_{\substack{u_1,v_1 \sim S: u_1 \neq v_1\\ u_2,v_2 \sim S: u_2 \neq v_2}}[I(Y_{u_1,v_1};Y_{u_2,v_2})] \leq \tau.$
\item $\E_{\substack{u_1,v_1 \sim S: u_1 \neq v_1\\ u_2,v_2 \sim S: u_2 \neq v_2}}[I(Y'_{u_1,v_1};Y'_{u_2,v_2})] \leq \tau$,
\end{enumerate}
where for $A \subseteq [N]$, $Y_{A} = (X_{i_1},\ldots, X_{i_{|A|}},p_{i_1},\ldots, p_{i_{|A|}})$ (same for $Y'$) and the mutual information is with respect to the collection of local distributions $\cL(\mu_1 \times \mu_2,p_{\beta,\nu})$.
\end{lemma}
\begin{proof}
Deferred to Section~\ref{sec:ragh}.
\end{proof}

\subsubsection{Proof of Lemma~\ref{lem:sp-j}}
We combine the lemmas stated in the previous section and set parameters to complete the proof of Lemma~\ref{lem:sp-j}.

\begin{lemma}[Restatement of Lemma~\ref{lem:sp-j}]\label{lem:sp-j-restated}
There exists a constant $\eps_0 \in (0,1)$, such that for all positive constants $\eps \leq \eps_0$, $\alpha, \tau \leq 1$, $\nu \leq \eps \exp(-r)$, all integers $\ell \geq \Omega(r), n \geq \ell$, where $r = \lfloor \frac{32\sqrt{\eps}}{\alpha}\rfloor$ the following holds: Let $I$ be an affine UG instance on $J_{n,\ell,\alpha}$ and $\mu$ be a pseudodistribution over assignments for $I$ with $\val_\mu(I) \geq 1-\eps$ and degree at least $\widetilde{\Omega}\left(\frac{|\Sigma|\ell^r\exp(r)}{\nu\tau\sqrt{\eps}}\right)$. There exists a restriction $a \subseteq [n]$ of size $j \leq r$, a degree $\widetilde{O}(1/\nu)$ polynomial $P_a(X,X')$ in a fixed set of $|\Sigma|$ polynomials, subsets $A, B \subseteq V(J|_a)$ of size at most $\tO(\frac{|\Sigma|\ell^j\exp(r)}{\tau\sqrt{\eps}})$ and strings $y_A, y_B$ such that conditioning $\mu$ on the events $Y_A = y_A$ and $Y_B = y_B$ gives degree $\widetilde{\Omega}(1/\nu)$ pseudodistributions $\mu_1$ and $\mu_2$ such that:
\begin{enumerate}
\item $\cA_I  ~~\vdash_{\widetilde{O}(1/\nu)} ~~ P_a(X,X') \in [0,1]$.
\item $\pE_{\mu_1 \times \mu_2}[\Phi^a_{\sqrt{\eps},\nu}(X,X')P_a(X,X')] \geq \exp(-r) \pE[P_a(X,X')]$.
\item $\pE_{\mu_1 \times \mu_2}[P_a(X,X')] \geq \Omega\left(\frac{\sqrt{\eps}}{|\Sigma|\ell^{j}\exp(r)}\right)$.
\item $\E_{\substack{u_1,v_1 \sim S: u_1 \neq v_1\\ u_2,v_2 \sim S: u_2 \neq v_2}}[I(Y_{u_1,v_1};Y_{u_2,v_2})] \leq \tau.$
\item $\E_{\substack{u_1,v_1 \sim S: u_1 \neq v_1\\ u_2,v_2 \sim S: u_2 \neq v_2}}[I(Y'_{u_1,v_1};Y'_{u_2,v_2})] \leq \tau,$
\end{enumerate}
where $S = V(J|_a)$, $Y_{u,v} = (X_u,X_v,p_u,p_v)$, $Y'_{u,v} = (X'_u,X'_v,p'_u,p'_v)$ and the mutual information is taken with respect to the collection of local distributions $\cL(\mu_1 \times \mu_2,p_{\sqrt{\eps},\nu})$. 
\end{lemma}
\begin{proof}


We will apply Lemma~\ref{lem:structure-fs-j} with the following parameters: $r = \lfloor 32\sqrt{\eps}/\alpha \rfloor$, so that $1 - \lambda_{r+1} \geq 16\sqrt{\eps}$,  $\eta = 1/\sqrt{\eps}$ and $\beta = 201\sqrt{\eps}$.
Since $\eps$ is sufficiently small and $\nu \leq \eps/2$ we get that $1/(1 - \beta - \nu) < 2$. By assumption $\pE_\mu[\viol(X)] = \pE_\mu[\viol(X')] \le \eps$, $\pE_\mu$ has degree $\widetilde{O}(1/\nu)$ and $\pE_\mu$ satisfies $\cA_I$, therefore taking the pseudoexpectation of equation~\ref{eq:sp-j-1} with respect to the pseudodistribution $\mu \times \mu$ we get:
\begin{align}
\sum_{j = 0}^r c_j \exp(r) \ell^j \E_{a \in {[n] \choose j}}\pE_{\mu \times \mu}[\sum_{s}q_{a}(F_s)(\delta(F_s|_a) - \exp(-r))] \geq \frac{\sqrt{\eps}}{3}.    
\end{align}

Since $c_j$'s are smaller than $\exp(r)$, an averaging argument gives us a size $j \leq r$ restriction $a \subseteq [n]$ and an $s \in \Sigma$ such that:
\begin{equation}\label{eq:sp-j-2}
\pE_{\mu \times \mu}[q_{a}(F_s)(\delta(F_s|_a) - \exp(-r))] \geq \Omega\left(\frac{\sqrt{\eps}}{|\Sigma|\ell^{j}\exp(r)}\right).
\end{equation}

Let $P_a(X,X')$ be the polynomial $q_a(F_s(X,X'))$ from above, abbreviated henceforth as $P_a$. First note that since $q_a$ is a degree $O(1)$-SoS polynomial we immediately get that $P_a$ is a degree $\widetilde{O}(1/\nu)$ SoS polynomial such that:
\[\cA_I ~~ \vdash_{\widetilde{O}(1/\nu)} F_s(u) \in [0,1] ~~\vdash_{O(1)} ~~ q_a(F_s(X,X')) = P_a(X,X') \in [0,1],\]
thus proving the first statement in the theorem.
We now reduce the global correlations using Lemma~\ref{lem:ragh-tan-appln}. 
Namely, we take $\mu$ with $S$ being the set of variables corresponding to the subcube $a$, the polynomial $E(X,X') = P_a(\delta(F_s|_a) - \exp(-r))$ (for which we know that $\cA \vdash_{\tO(1/\nu)} E(X,X') \leq 1$), $\beta = \sqrt{\eps}$ and $\tau,\nu$ as in the lemma statement. Thus, we get a 
pseudodistribution $\mu_1 \times \mu_2$ over UG assignments to the subcube $a$, such that: 
\begin{equation}\label{eq:sp-large-mu1-mu2}
\pE_{\mu_1 \times \mu_2}[P_a(\delta(F_s|_a) - \exp(-r))] \geq \Omega\left(\frac{\sqrt{\eps}}{|\Sigma|\ell^{j}\exp(r)}\right),
\end{equation}
 along with the conditions on the mutual information of $\mu_1,\mu_2$. 
This proves point (4),(5) from the lemma statement. 

\paragraph{Establishing point (2) of the lemma.}
Let us now derive the fact that the global shift-partition potential $\Phi(\mu_1 \times \mu_2)|_a$ is large conditioned on $P_a$.  
Rearranging~\eqref{eq:sp-large-mu1-mu2} we get that:
\begin{equation}\label{eq1}
    \pE_{\mu_1 \times \mu_2}[P_{a}\delta(F_s|_a)] \geq \exp(-r)\pE_{\mu_1 \times \mu_2}[P_a] + \Omega\left(\frac{\sqrt{\eps}}{|\Sigma|\ell^{j}\exp(r)}\right) \geq \exp(-r) \pE_{\mu_1 \times \mu_2}[P_a].
\end{equation}

Since $P_a$ is an SoS polynomial we can reweight $\pE_{\mu_1 \times \mu_2}[\cdot]$ by $P_a$ to get $\widetilde{\mathbb{E}'}$ and apply Cauchy-Schwarz to get:
\[\exp(-r) \leq \widetilde{\mathbb{E}'}[\delta(F_s|_a)] \leq \sqrt{\widetilde{\mathbb{E}'}[\delta(F_s|_a)^2]} = 
\sqrt{\frac{\pE_{\mu_1 \times \mu_2}[\delta(F_s|_a)^2 P_a]}{\pE_{\mu_1 \times \mu_2}[P_a]}},\]
rearranging which we get that: $\pE_{\mu_1 \times \mu_2}[\delta(F_s|_a)^2P_a] \geq \exp(-r)\pE_{\mu_1 \times \mu_2}[P_a]$. By definition, the global potential restricted to $a$ is equal to $\sum_{s \in \Sigma} \delta(F_s|_a)^2$. Therefore adding the terms $\pE_{\mu_1 \times \mu_2}[\delta(F_{s'}|_a)^2 P_a]$ for $s' \in \Sigma, s' \neq s$ (which are all non-negative) to the LHS of~\eqref{eq1} we get:
\begin{equation}\label{eq:global-large}
\pE_{\mu_1 \times \mu_2}[\Phi_{\beta,\nu}(X,X')|_a P_a] \geq \exp(-r)\pE_{\mu_1 \times \mu_2}[P_a].
\end{equation}

We will now relate the global potential to the shift-partition potential on $a$, where the only difference between the quantities is that in the former the value of a vertex is the fraction of edges satisfied in the whole graph, whereas in the latter it is the value calculated according to only the edges inside the subcube $a$ (in the terms $p_{\beta,\nu}(\val(u))$).

First note that that an $r$-restricted subcube has bounded expansion when $r$ is not too large: $\phi(J|_a) \leq O(\sqrt{\eps})$ (Claim~\ref{claim:subcube-expanse}). Using this we get that,
\[
\Phi^a_{\beta -200\sqrt{\eps},\nu}(X,X') \ge \Phi_{\beta,\nu}(X,X')|_a - 4\nu,
\]
which we prove formally in Claim~\ref{claim:potentials}.


Overall, we get:
\[\pE_{\mu_1 \times \mu_2}[\Phi^a_{\sqrt{\eps},\nu}(X,X')P_a] \ge \pE_{\mu_1 \times \mu_2}[(\Phi_{201\sqrt{\eps},\nu}(X,X')|_a - 4\nu) P_a] \geq \exp(-r)\pE_{\mu_1 \times \mu_2}[P_a],\]
since $\nu \leq \exp(-r)/8$, which establishes point (2) of the 
lemma.

Point (3) follows from~\eqref{eq:sp-large-mu1-mu2}. Indeed, 
since $\delta(F_s|_a) \leq 1$ and $P_a(X,X') \geq 0$ (both facts are certifiable in degree $\widetilde{O}(1/\nu)$) we get that: $\cA_I \vdash_{\widetilde{O}(1/\nu)} P_{a}(X,X')(\delta(F_s|_a)-\exp(-r))] \leq (1-\exp(-r))P_a(X,X')$. Taking $\pE_{\mu_1 \times \mu_2}$ and plugging this into equation~\eqref{eq:sp-large-mu1-mu2} we get that:
\[\pE_{\mu_1 \times \mu_2}[P_a] \geq \Omega\left(\frac{\sqrt{\eps}}{|\Sigma|\ell^{j}\exp(r)}\right)
\qedhere\]
\end{proof}

\subsubsection{Auxiliary Claims}\label{sec:aux_claims}
We end this section by giving the proofs of Claims~\ref{claim:subcube-expanse} and \ref{claim:potentials}. These were also used in~\cite{BBKSS} but we include the proofs here for completeness.

\begin{claim}[Claim 6.10~\cite{BBKSS}]\label{claim:subcube-expanse}
If $r = \left\lfloor\frac{32\sqrt{\eps}}{\alpha}\right\rfloor < \frac{\ell}{4}$ and $s < r$, an $s$-restricted subcube of $J(n,\ell,\alpha\ell)$ has expansion at most $200\sqrt{\eps}$.
\end{claim}

\begin{proof}
Let $J|_Y$ be an $s$-restricted subcube.
We have that,
\[1 - \phi(J|_Y) = \frac{\binom{\ell - |Y|}{\alpha\ell}}{\binom{\ell}{\alpha\ell}} \geq \frac{\binom{\ell - r}{\alpha\ell}}{\binom{\ell}{\alpha\ell}} = \left(\frac{\ell - \alpha\ell}{\ell}\right)\left(\frac{\ell - \alpha\ell - 1}{\ell - 1}\right)\ldots \left(\frac{\ell - \alpha\ell - r + 1}{\ell - r + 1}\right).\]
Now since $r \leq \ell/4$ by assumption, each of the parenthesized terms is at least $\left(\frac{3\ell/4 - \alpha\ell}{3\ell/4}\right) = (1 - 4\alpha/3)$, so
\[1 - \phi(J|_Y) \geq \left(1-\frac{4\alpha}{3}\right)^r \geq 1 - \frac{4r\alpha}{3}.\]

Since $r = \left\lfloor\frac{32\sqrt{\eps}}{\alpha}\right\rfloor < \frac{75\sqrt{\eps}}{\alpha}$, we get that $\phi(J|_Y) < 200\sqrt{\eps}$
as desired.
\end{proof}

\begin{claim}[Claim 6.11~\cite{BBKSS}]\label{claim:potentials}
Suppose that $C$ is an $r$-restricted subcube of $J_{n,\ell,\alpha}$ with $r = \left\lfloor\frac{32\sqrt{\eps}}{\alpha}\right\rfloor$.
Then if $\Phi^C$ is the shift-partition potential restricted to $C$, for any $\beta \ge 201\sqrt{\eps}$ and $\nu < \eps$,
\[
\Phi^C_{\beta -200\sqrt{\eps},\nu}(X,X') \ge \Phi_{\beta,\nu}(X,X')|_C - 4\nu,
\]
and furthermore this is certifiable in degree $\tO(1/\nu)$ \sos.
\end{claim}

\begin{proof}
When $r = \left\lfloor\frac{32\sqrt{\eps}}{\alpha}\right\rfloor$, the expansion of $C$ is at most $1 - (1-4\alpha/3)^r \le 200\sqrt{\eps}$ by Claim~\ref{claim:subcube-expanse}.
Furthermore, from the definition of the Johnson graph this holds vertex-by-vertex; every $v \in C$ has at most a $200\sqrt{\eps}$-fraction of its neighbors outgoing.
Therefore,
\[
\Ind[\val_u^C(X) \ge \beta - 200\sqrt{\eps}] \ge \Ind[\val_u(X) \ge \beta],
\]
and furthermore since $\nu < \eps$,
\[
p_{\beta - 200 \sqrt{\eps},\nu}(\val_u^C(X)) + \nu \ge p_{\beta,\nu}(\val_u(X)) - \nu.
\]
The same holds for $p(\val_u(X'))$. Let $Z_{u,s} = \Ind[X_u - X'_u = s]$. Therefore, by definition,
\begin{align*}
    \Phi_{\beta - 200\sqrt{\eps},\nu}(X|_C,X'|_C)
    &= \sum_{s \in \Sigma} \E_{u \in C}\left[Z_{u,s}\cdot p_{\beta-200\sqrt{\eps},\nu}(\val_u^C(X))\cdot p_{\beta-200\sqrt{\eps},\nu}(\val_u^C(X'))\right]^2\\
    &\ge \sum_{s \in \Sigma} \E_{u \in C}\left[Z_{u,s}\cdot (p_{\beta,\nu}(\val_u(X)) - 2\nu)\cdot (p_{\beta,\nu}(\val_u(X')) - 2\nu)\right]^2\\
    &\ge \Phi_{\beta,\nu}(X,X')|_C - 4\nu,
\end{align*}
where each inequality is a sum-of-squares inequality of degree at most $2\deg(p)$.
\end{proof}

\subsection{Proof of Lemma~\ref{lem:ragh-tan-appln}: Reducing Global Correlation}\label{sec:ragh}
We will use the following lemma from~\cite{RT12}.

\begin{lemma}\label{lem:ragh-tan}
There exists $t \leq r$ such that: 
\[\E_{i_1,\ldots,i_t \sim [M]}\E_{i,j \sim [M]}[I(Y_i;Y_j \mid Y_{i_1}, \ldots,Y_{i_t})] \leq \frac{\log q}{r - 1},\]
where $q$ is the size of the domain of $Y_i$.
\end{lemma}

Note that the above lemma holds as long as there is a local collection of distributions over the variables $(Y_1,\ldots,Y_M)$ that are valid probability distributions over all collections of $t+2$ variables and are consistent with each other. We will apply the above lemma to reduce the global correlation between $(Y_{u_1,v_1};Y_{u_2,v_2})$ (similarly for $Y'$), for the local collection of distributions $\cL$ (Definition~\ref{defn:local-dist}).

\begin{lemma}[Lemma~\ref{lem:ragh-tan-appln} restated]
For all $\tau,p,\beta,\nu \in (0,1)$ and $|\Sigma|, D, D' \in \Z$, $D' \geq \max(D,\widetilde{\Omega}(1/\nu))$ the following holds: Suppose there is a degree $D' + \widetilde{\Omega}(\frac{\log |\Sigma|}{p \nu \tau})$ pseudodistribution $\mu$ over UG assignments that satisfies $\cA_I$, and a polynomial $E(X,X')$ satisfying $\pE_{(X,X') \sim \mu \times \mu}[E(X,X')] \geq p$ and $\cA_I \vdash_D E(X,X') \leq 1$. Then for all subsets $S \subseteq [N]$, there exist subsets $A, B \subseteq S$ of size at most $O(\frac{\log|\Sigma|}{p\tau})$ and strings $y_A, y_B$ such that conditioning $\mu$ on the events $Y_A = y_A$ and $Y_B = y_B$ gives pseudodistributions $\mu_1$ and $\mu_2$ of degree at least $D'$ such that: 
\begin{enumerate}
\item $\pE_{X,X'\sim \mu_1\times\mu_2}[E(X,X')] \geq \frac{p}{2}$.
\item $\E_{\substack{u_1,v_1 \sim S: u_1 \neq v_1\\ u_2,v_2 \sim S: u_2 \neq v_2}}[I(Y_{u_1,v_1};Y_{u_2,v_2})] \leq \tau.$
\item $\E_{\substack{u_1,v_1 \sim S: u_1 \neq v_1\\ u_2,v_2 \sim S: u_2 \neq v_2}}[I(Y'_{u_1,v_1};Y'_{u_2,v_2})] \leq \tau$,
\end{enumerate}
where for $A \subseteq [N]$, $Y_{A} = (X_{i_1},\ldots, X_{i_{|A|}},p_{i_1},\ldots, p_{i_{|A|}})$ (same for $Y'$) and the mutual information is with respect to the collection of local distributions $\cL(\mu_1 \times \mu_2,p_{\beta,\nu})$.
\end{lemma}

\begin{proof}
Without loss of generality suppose that $S \subseteq [N]$ is $[m]$. Let ${[m] \choose r}$ denote the set of $r$-sized subsets of $[m]$ and $m_r$ denote ${m \choose r}$. Let $Y = (Y_{a_1},\ldots,Y_{a_{m_2}})$ for $a_t = (u,v) \in {[m] \choose 2}$. For brevity we will use $[m_2]$ for ${[m] \choose 2}$. Given the pseudodistribution $\mu \times \mu$ over assignments $(X,X')$ we will consider the 
local collection of distributions $\cL(\mu \times \mu,p_{\beta,\nu})$ over the variables of $X,X',p,p'$ and let this also denote the induced collection of distributions over $(Y,Y')$ (Definition~\ref{defn:local-dist}). Recall the notation for random variable $Y_A$, for any $A \subseteq [m]$: $Y_A = (X_{i_1}, \ldots, X_{i_{|A|}}, p_{i_1}, \ldots, p_{i_{|A|}})$. We will use $y_S$ to denote instantiations of $Y_S$.

Let $\cL$ denote $\cL(\mu \times \mu,p_{\beta,\nu})$.
Applying Lemma~\ref{lem:ragh-tan} to the collection of  distributions $\cL$ over $Y$ we get that for all $\tau' > 0$, conditioning on $t = O(\frac{\log |\Sigma|}{\tau'})$ variables gives:
\begin{equation}\label{eq:rt-1}
\E_{i_1,\ldots,i_t \sim [m_2]}[\E_{a,b \sim [m_2]} I(Y_a ; Y_b \mid Y_{i_1}, \ldots,Y_{i_t})] \leq \tau'.
\end{equation}

Let $\cL|_R$ denote the distribution that $\cL$ induces on the random variable $R$. We know that the distribution $\cL|_{(Y_{i_1}, \ldots, Y_{i_t})}$ can be reduced to the distribution $\cL|_{Y_{i_1 \cup \ldots \cup i_t}}$ by discarding repeating indices in $i_1,\ldots,i_t$. This is because $\cL$ is consistent on the value it assigns to a vertex $u$ when it occurs in $Y_{u,v}$ or $Y_{u,v'}$. We will drop the subscript $R$ from $\cL|_R$ when the random variable is clear from context. Expanding the definition of conditional mutual information we therefore get:
\begin{align*}
I(Y_a;Y_b \mid Y_{i_1}, \ldots, Y_{i_t}) &= \E_{(y_{i_1}, \ldots,y_{i_t}) \sim \cL|_{(Y_{i_1}, \ldots, Y_{i_t})}}[I(Y_a ; Y_b \mid Y_{i_1} = y_{i_1}, \ldots,Y_{i_t} = y_{i_t})] \\
&= \E_{y_{i_1 \cup \ldots \cup i_t} \sim \cL|_{Y_{i_1 \cup \ldots \cup i_t}}}[I(Y_a ; Y_b \mid Y_{i_1 \cup \ldots \cup i_t} = y_{i_1 \cup \ldots \cup i_t}].
\end{align*}
Plugging the above into equation~\eqref{eq:rt-1} and applying Markov's inequality we get that for all $\alpha \in (0,1)$:
\begin{equation}\label{eq:avg1}
\Pr_{\substack{i_1,\ldots,i_t \sim [m_2] \\ y_{i_1 \cup \ldots \cup i_t} \sim \cL}}\left[\E_{a,b \sim [m_2]}[I(Y_a ; Y_b \mid Y_{i_1 \cup \ldots \cup i_t} = y_{i_1 \cup \ldots \cup i_t})] \geq \frac{\tau'}{\alpha} \right] \leq \alpha.
\end{equation}
Repeating the above analysis for $Y'$ we get:
\begin{equation}\label{eq:avg2}
\Pr_{\substack{j_1,\ldots,j_t \sim [m_2] \\ y'_{j_1 \cup \ldots \cup j_t} \sim \cL}}\left[\E_{a,b \sim [m_2]}[I(Y_a ; Y_b \mid Y'_{j_1 \cup \ldots \cup j_t} = y_{j_1 \cup \ldots \cup j_t})] \geq \frac{\tau'}{\alpha} \right] \leq \alpha.
\end{equation}

Now recall that (Definition~\ref{defn:local-dist}) the probability of the event $(X_u,p_u) =  (\sigma,b)$ under $\cL$, for some $\sigma \in \Sigma$ and $b \in \{0,1\}$, corresponds to the pseudoexpectation under $\mu$ of an appropriate polynomial we will denote by $Q_{u,(\sigma,b)}$:
\begin{align}\label{eq:poly-cond}
\pPr_{\cL}[(X_u, p_u) = (\sigma, b)] &= \pE[X_{u,\sigma}(1-b + (-1)^{1-b}p(\val_u(X)))] \nonumber \\
&:= \pE_\mu[Q_{u,(\sigma,b)}].
\end{align}
We can extend the above definition by letting $Q_{S, \cup_{u \in S} (\sigma_u,b_u)}$ denote the polynomial $\prod_{u \in S} Q_{u,(\sigma_u,b_u)}$. For $y_S = (\cup_{u \in S} \sigma_u, \cup_{u \in S} b_u)$ we get that:
\begin{align}\label{eq:poly-cond-prod}
\pPr_{\cL}[(\cup_{u \in S} X_u, \cup_{u \in S} p_u) = Y_S = y_S] = \pE_\mu[Q_{S,\cup_{u \in S}(\sigma_u,b_u)}].
\end{align}
One can define analogous notation for $Y'$ (e.g. the polynomials $Q'_{u,(\sigma,b)}$, etc) and derive the statements above. We know that $\cA_I \vdash_{t~\tO(1/\nu)} Q_{S, \cup_{u \in S} (\sigma_u,b_u)} \in [0,1]$. For ease of notation let us define the following expression:
\begin{align*}
\pE_{\mu \times \mu}[E(X,X') \mid Y_S = y_S, Y'_T = y'_T] &:= \pE_{\mu \times \mu}[E(X,X') \mid  Q_{S, y_S} \cdot  Q'_{T, y'_T}] \\
&= \frac{\pE_{\mu \times \mu}[E(X,X')Q_{S, y_S} Q'_{T, y'_T}]}{\pE_{\mu \times \mu}[Q_{S, y_S} Q'_{T, y'_T}]} \\
&= \frac{\pE_{\mu \times \mu}[E(X,X')Q_{S, y_S} Q'_{T, y'_T}]}{\pPr_{\cL}[Y_{S} = y_{S}, Y'_{T} = y'_{T}]},
\end{align*}
where the last equality follows from equation~\eqref{eq:poly-cond-prod}. Analogous to the definition of conditional expectation we can check that:
\[\pE_{\mu \times \mu}[E(X,X')] = \E_{\substack{i_1,\ldots,i_t \sim [m] \\ j_1,\ldots,j_t \sim [m] \\ y_{i_1 \cup \ldots \cup i_t}, y'_{j_1 \cup \ldots \cup j_t} \sim \cL}}[\pE_{X,X'}[E(X,X') \mid Y_{i_1 \cup \ldots \cup i_t} = y_{i_1 \cup \ldots \cup i_t}, Y'_{j_1 \cup \ldots \cup j_t} = y'_{j_1 \cup \ldots \cup j_t}]].\]

Since $\cA_I \vdash_D E(X,X') \leq 1$ and $\deg(\mu) \geq D + \widetilde{\Omega}(t/\nu)$ we get that $E(X,X') \leq 1$ even after conditioning on a non-negative event $Q$ of degree $\widetilde{O}(t/\nu)$: $\pE[E(X,X') \mid Q] \leq 1$. An averaging argument implies that:
\begin{equation}\label{eq:avg3}
\Pr_{\substack{i_1,\ldots, j_1,\ldots \sim [m_2] \\ y_{i_1 \cup \ldots \cup i_t},y'_{j_1 \cup \ldots \cup j_t} \sim \cL}}\left[\pE_{\mu \times \mu}[E(X,X') \mid Y_{i_1 \cup \ldots \cup i_t} = y_{i_1 \cup \ldots \cup i_t}, Y'_{j_1 \cup \ldots \cup j_t} = y'_{j_1 \cup \ldots \cup j_t}] \geq \frac{p}{2}\right] \geq \frac{p}{2}.
\end{equation}

Choosing $\alpha = p/8$ and $\tau' < p\tau/8$ we can take a union bound over the events in equations~\eqref{eq:avg1},~\eqref{eq:avg2},~\eqref{eq:avg3} to get that there exist sets $i_1,\ldots,i_t, j_1,\ldots,j_t \in [m_2]$ and $y_{i_1 \cup \ldots \cup i_t}, y'_{j_1 \cup \ldots \cup j_t} \sim \cL$ such that,   
\[\E_{a,b \sim [m_2]}[I(Y_a ; Y_b \mid Y_{i_1 \cup \ldots \cup i_t} = y_{i_1 \cup \ldots \cup i_t}] \leq \frac{\tau'}{\alpha} < \tau.\]
\[\E_{a,b \sim [m_2]}[I(Y'_a ; Y'_b \mid Y'_{j_1 \cup \ldots \cup j_t} = y'_{j_1 \cup \ldots \cup j_t})] \leq \frac{\tau'}{\alpha} < \tau.
\]
\[\pE_{\mu \times \mu}[E(X,X') \mid Y_{i_1 \cup \ldots \cup i_t} = y_{i_1 \cup \ldots \cup i_t}, Y'_{j_1 \cup \ldots \cup j_t} = y'_{j_1 \cup \ldots \cup j_t}] \geq \frac{p}{2}.\]

Let $\mu_1$ be the pseudodistribution on $X$ that we get by conditioning $\mu$ on $Q_{i_1 \cup \ldots \cup i_t, y_{i_1 \cup \ldots \cup i_t}}$ and let $\mu_2$ be the pseudodistribution on $X'$ that we get by conditioning $\mu$ on $Q_{j_1 \cup \ldots \cup j_t,y'_{j_1 \cup \ldots \cup j_t}}$. It is easy to check that:
\[\cL(\mu \times \mu, p_{\beta,\nu})| (Y_{i_1 \cup \ldots \cup i_t} = y_{i_1 \cup \ldots \cup i_t}, Y'_{j_1 \cup \ldots \cup j_t} = y'_{j_1 \cup \ldots \cup j_t}) = \cL(\mu_1 \times \mu_2, p_{\beta,\nu}),\]
thus giving us the three properties we need in the lemma. It remains to check the degree bounds on $\mu$ that we require. We have conditioned on $t = O(\log |\Sigma|/p\tau)$ variables from $(Y,Y')$, each having degree $\tO(1/\nu)$ in $X,X'$, therefore it suffices to have degree of $\mu$ to be $\widetilde{\Omega}\left(\frac{\log |\Sigma|}{p\tau\nu}\right)$. After conditioning we get that $\deg(\mu_1 \times \mu_2) \geq D'$ if $\deg(\mu) \geq D'+\widetilde{\Omega}\left(\frac{\log |\Sigma|}{p\tau\nu}\right)$ thus completing the proof.
\end{proof}

\subsection{Proof of Lemma~\ref{lem:correlation}: 
Conditioning Does Not Introduce Correlations}

In this section we prove that given a product distribution over $(X,X')$ that satisfies low global correlation with respect to the variables $(Y,Y')$, where $Y_{u,v} = (X_u,X_v,p_u,p_v)$, and an event $E(X,X')$ that holds with large probability, conditioning on $E$ cannot correlate too many pairs $(Y_i,Y'_i)$. This is a general lemma that holds if the variables $(Y,Y')$ are local functions of the underlying variables $(X,X')$ and satisfy low global correlation, but here we state it for our specific application only.

\begin{lemma}[Restatement of Lemma~\ref{lem:correlation}]
For all $\tau,p,\beta,\nu,\delta \in (0,1)$, integers $D,N$, $m \leq N$ the following holds: Let $\mu_1 \times \mu_2$ be a degree $D+\widetilde{\Omega}(1/\nu)$ pseudodistribution over $(X,X')$ satisfying $\cA_I(X,X')$ and $E(X,X')$ be a polynomial such that, $\cA_I \vdash_D E(X,X') \in [0,1]$ and $\pE_{\mu_1 \times \mu_2}[E(X,X')] \geq p$. Suppose for $S \subseteq [N]$ we have that, 
\begin{enumerate}
\item $\E_{\substack{u_1,v_1 \sim S: u_1 \neq v_1\\ u_2,v_2 \sim S: u_2 \neq v_2}}[I(Y_{u_1,v_1};Y_{u_2,v_2})] \leq \tau.$
\item $\E_{\substack{u_1,v_1 \sim S: u_1 \neq v_1\\ u_2,v_2 \sim S: u_2 \neq v_2}}[I(Y'_{u_1,v_1};Y'_{u_2,v_2})] \leq \tau$,
\end{enumerate}
where $Y_{u,v} = (X_u,X_v,p_u,p_v), Y'_{u,v} = (X'_u,X'_v,p'_u,p'_v)$ and the mutual information is with respect to the collection of local distributions $\cL(\mu_1 \times \mu_2,p_{\beta,\nu})$ (Definition~\ref{defn:local-dist}). Then we have that, 
\[\Pr_{u,v \sim S: u \neq v}[ TV((Y_{u,v},Y'_{u,v})|E , (Y_{u,v},Y'_{u,v})) \geq \delta] \leq O\left(\frac{\sqrt{\tau}+ 1/|S|}{p\delta^2}\right),\]
where the distribution $(Y_{u,v},Y'_{u,v})|E$ refers to the joint distribution on these variables defined by the collection of local distributions $\cL(\mu_1 \times \mu_2|E, p_{\beta,\nu})$ and similarly $(Y_{u,v},Y'_{u,v})$ refers to the distribution defined by $\cL(\mu_1 \times \mu_2,p_{\beta,\nu})$.
\end{lemma}

\begin{proof}
As in the proof of Lemma~\ref{lem:ragh-tan-appln}, without loss of generality suppose $S = [m]$, let $[m_2]$ denote $2$-sized subsets of $m$ and $Y = (Y_{a_1},\ldots,Y_{a_M})$ for $a_t = (u,v), (u,v) \in [m_2]$. Given the pseudodistribution $\cD$ over assignments $(X,X')$ we will consider the local collection of distributions $\cL(\cD, p_{\beta,\nu})$ over the variables $X,X',p,p'$ and therefore the induced collection of distributions over $(Y,Y') = (Y_{a_1},\ldots,Y_{a_M},Y'_{a_1},\ldots,Y'_{a_M})$ (Definition~\ref{defn:local-dist}).

Throughout the proof we will consider random variables $(Y,Y')$ drawn from the collection of local distributions $\cL(\mu_1 \times \mu_2,p_{\beta,\nu})$ and from the conditioned collection of distributions $\cL(\mu_1 \times \mu_2 | E,p_{\beta,\nu})$. For brevity of notation we will use $(Y_a,Y'_a)|E$ to be the joint distribution induced by $\cL(\mu_1 \times \mu_2 | E, p_{\beta,\nu})$ on $(Y_a,Y'_a)$ and use $(Y_a,Y'_a)$ to be the distribution induced by $\cL(\mu_1 \times \mu_2,p_{\beta,\nu})$. We will abbreviate $\cL(\mu_1 \times \mu_2,p_{\beta,\nu})$ to $\cL$. Let $U$ be the set of variables $a \in [m_2]$ for which $TV((Y_a,Y'_a)|E , (Y_a,Y'_a)) \geq \delta$ and let the fractional-size of $U$ be $\gamma$. If $\Omega$ denotes the domain of $Y_u, Y'_u$ then for every $a \in U$ there exists a set $S_a \subseteq \Omega \times \Omega$ such that: 
\begin{equation}\label{eq:e-tv}
\pPr_{\cL}[(Y_a,Y'_a) \in S_a \mid E] - \pPr_{\cL}[(Y_a,Y'_a) \in S_a] \geq \delta.
\end{equation}
Let $e_a$ denote  $\pPr_{\cL}[(Y_a,Y'_a) \in S_a]$. Define the random variables $Z_a = \one((Y_a,Y'_a) \in S_a) - e_a$. Define:
\[Z = \E_{a \sim U}[Z_a] = \E_{a \in U}[\one((Y_a,Y'_a) \in S_a) - e_a].\]

Let $\pE_{\cL}$ denote the natural expectation operator corresponding to the local distributions $\cL$. 
One can check that $\pE_{\cL}[Z] = \pE_{\cL}[Z_a] = 0$, and we now calculate its variance. For two events $A,B$ on the variables $Y,Y'$ let $\pCov(A,B)$ denote $\pE_{\cL}[AB] - \pE_{\cL}[A]\pE_{\cL}[B]$. Firstly for all $a,b \in [m_2]$  using Pinsker's inequality and data processing inequality we have that,
\begin{align*}
\pCov(\one((Y_a,Y'_a) \in S_a),\one((Y_b,Y'_b) \in S_b)) 
&\leq TV((Y_a,Y'_a,Y_b,Y'_b), (Y_a,Y'_a)\times (Y_b,Y'_b)) \\ &\leq O(\sqrt{I(Y_a,Y'_a;Y_b,Y'_b)}) \\
&\leq O(\sqrt{I(Y_a;Y_b) + I(Y'_a;Y'_b)})
\end{align*}

The proof will proceed by proving upper and lower bounds on $\pE_{\cL}[Z^2]$, where the upper bound uses low global correlation properties of $Y$ and $Y'$ and the lower bound uses the large deviation we have by equation~\ref{eq:e-tv}. 
\paragraph{Upper bound for $\pE_{\cL}[Z^2]$:}
We have the following upper bound:
\begin{align*}
\pE_{\cL}[Z^2] &= \E_{a,b \sim U}[\pE_{\cL}[(\one((Y_a,Y'_a) \in S_a) - e_a)(\one((Y_b,Y'_b) \in S_b) - e_b)]] \\
&= \E_{a,b \sim U}[\pE_{\cL}[\one((Y_a,Y'_a) \in S_a) - e_a]\pE_{\cL}[\one((Y_b,Y'_b) \in S_b) - e_b] + \pCov(\one((Y_a,Y'_a) \in S_a),\one((Y_b,Y'_b) \in S_b))] \\
&\leq \E_{a,b \sim U}[O(\sqrt{I(Y_a;Y_b) + I(Y'_a;Y'_b)})] \\
&\leq O(\sqrt{\E_{a,b \sim U}[I(Y_a;Y_b) + I(Y'_a;Y'_b)}])\\
&\leq O\left(\frac{\sqrt{\tau}}{\gamma}\right),
\end{align*}
where the last inequality follows because $\E_{a,b \sim [m_2]}[I(Y_a;Y_b)] \geq \gamma^2 \E_{a,b \sim U}[I(Y_a;Y_b)]$, and by assumption $\E_{a,b \sim [M]}[I(Y_a;Y_b)] \leq \tau$ (similarly for $Y'$).

\paragraph{Lower bound for $\pE_{\cL}[Z^2]$:}
Let $Q_a$ be the polynomial corresponding to the event $\Ind[(Y_a,Y'_a) \in S_a]$, i.e. $\pPr_{\cL}[(Y_a,Y'_a) \in S_a] = \pE_{\mu_1 \times \mu_2}[Q_a]$. Let $P(X,X')$ abbreviated as $P$ denote the polynomial $\E_a[Q_a - e_a]$. Let $Q_{a \cup b}$ be the polynomial such that $\pPr_{\cL}[(Y_a,Y'_a) \in S_a \wedge (Y_b,Y'_b) \in S_b] = \pE_{\mu_1 \times \mu_2}[Q_{a \cup b}]$. If $a \cap b = \phi$ we have that $Q_{a \cup b} = Q_a \cdot Q_b$, whereas this may not hold if they intersect in one or two variables. But we have that $\Pr_{a, b \sim U}[a \cap b \neq \phi] \leq O(1/|U|)$.
Using these facts we first show that, $\pE_{\cL}[Z^2]+O(1/|U|) \geq \pE_{\mu_1 \times \mu_2}[P^2]$ via the following two equations:
\begin{align}
\pE_{\cL}[Z^2] &= \E_{a,b \sim U}[\pE_{\cL}[(\one((Y_a,Y'_a) \in S_a) - e_a)(\one((Y_b,Y'_b) \in S_b) - e_b)]] \notag \\
&=\Pr_{a,b \sim U}[a \cap b = \phi]\E_{a,b \sim U: a\cap b=\phi}[\pE_{\mu_1 \times \mu_2}[(Q_{a} - e_a)(Q_b - e_b)]] \notag\\ &+ \Pr_{a,b \sim U}[a \cap b \neq \phi]\E_{a,b \sim U: a\cap b\neq \phi}[\pE_{\cL}[(\one((Y_a,Y'_a) \in S_a) - e_a)(\one((Y_b,Y'_b) \in S_b) - e_b)]] \notag \\
&\geq \Pr_{a,b \sim U}[a \cap b = \phi]\E_{a,b \sim U: a\cap b=\phi}[\pE_{\mu_1 \times \mu_2}[(Q_{a} - e_a)(Q_b - e_b)]] - O\left(\frac{1}{|U|}\right).\label{eq:z2-1}
\end{align}

On the other hand we have that:
\begin{align}
\pE_{\mu_1 \times \mu_2}[P^2] &= \E_{a,b \sim U}[\pE_{\mu_1 \times \mu_2}[(Q_a - e_a)(Q_b - e_b)]] \notag \\
&= \Pr_{a,b \sim U}[a \cap b = \phi]\E_{a,b \sim U: a\cap b=\phi}[\pE_{\mu_1 \times \mu_2}[(Q_{a} - e_a)(Q_b - e_b)]] \notag \\
&+\Pr_{a,b \sim U}[a \cap b \neq \phi]\E_{a,b \sim U: a\cap b=\phi}[\pE_{\mu_1 \times \mu_2}[(Q_{a} - e_a)(Q_b - e_b)]] \notag \\
&\leq \Pr_{a,b \sim U}[a \cap b = \phi]\E_{a,b \sim U: a\cap b\neq\phi}[\pE_{\mu_1 \times \mu_2}[(Q_{a} - e_a)(Q_b - e_b)]] + O\left(\frac{1}{|U|}\right)\label{eq:z2-2}
\end{align}
Combining \eqref{eq:z2-1} and \eqref{eq:z2-2} we get that:
\[\pE_{\cL}[Z^2]+O\left(\frac{1}{|U|}\right) \geq \pE_{\mu_1 \times \mu_2}[P^2].\]

Recall that $|U| = \gamma |S|$. Since $\cA_I \vdash_{D} E(X,X') \in [0,1]$ we get that:
\[\pE_{\cL}[Z^2]+O\left(\frac{1}{\gamma|S|}\right)  \geq \pE_{\mu_1 \times \mu_2}[P^2] \geq \pE_{\mu_1 \times \mu_2}[E]\pE_{\mu_1 \times \mu_2}[P^2 | E] \geq p \pE_{\mu_1 \times \mu_2}[P | E]^2 = p\pE_{\cL}[Z | E]^2,\]
where the last inequality is by Cauchy-Schwarz on the pseudodistribution $\mu_1 \times \mu_2 | E$ and the last equality follows by the definition of $P$ and $Z$. Using equation~\eqref{eq:e-tv} we know that for all $a \in U$:
\[\pE_{\cL}[Z_a \mid E] = \pPr_{\cL}[(Y_a,Y'_a) \in S_a \mid E] - e_a  \geq \delta,\]
which implies that $\pE_{\cL}[Z \mid E] \geq \delta$.

Combining the upper and lower bounds on $\pE_{\cL}[Z^2]$ we get that $\gamma \leq O\left(\frac{\sqrt{\tau}+ 1/|S|}{p\delta^2}\right)$,
completing the proof of the lemma.





\end{proof}

\subsection{Proof of Lemma~\ref{lem:round-j}: Rounding Subgraphs with Large Shift Potential}
\label{sec:low-ent-j}
In this section, we will show that when the shift-partition potential $\Phi^G_{\beta,\nu}$ (Definition~\ref{def:shift-partition}) is large with respect to a pseudodistribution $\cD$ (with certain nice properties), then the Condition\&Round Algorithm  (Algorithm~\ref{alg:low-ent-restated}) succeeds in returning a good assignment for the unique games instance. The proofs in this section follow along the lines of the analysis of Condition\&Round given in~\cite{BBKSS}, albeit instead of independence between $X,X'$ we only have approximate local independence. Since we will always be working with the graph $G$ we will henceforth drop the superscript $G$ from the shift-partition potential $\Phi$.
We will prove the following theorem in this section:

\begin{lemma}[Restatement of Lemma~\ref{lem:round-j}]\label{thm:round-j}
Let $I = (G, \Pi)$ be an affine instance of Unique Games over the alphabet $\Sigma$ and $\mu_1 \times \mu_2$ be a degree $D+\widetilde{O}(1/\nu)$ pseudodistribution over assignments $(X,X')$ to $I$. Let $E(X,X')$ be a polynomial such that $\cA_I(X,X') \vdash_D E(X,X') \in [0,1]$. Suppose we have that: 
\[\Pr_{u,v \sim V(G): u \neq v}[TV((Y_{u,v},Y'_{u,v})|E , (Y_{u,v},Y'_{u,v})) > \delta] \leq \zeta,\]
where the distribution $(Y_{u,v},Y'_{u,v})|E$ refers to the joint distribution on these variables defined by the collection of local distributions $\cL(\mu_1 \times \mu_2|E, p_{\beta,\nu})$ and similarly $(Y_{u,v},Y'_{u,v})$ refers to the distribution defined by $\cL(\mu_1 \times \mu_2,p_{\beta,\nu})$ (Definition~\ref{defn:local-dist}). 

If $\Phi^G_{\beta,\nu}(\mu_1 \times \mu_2 | E(X,X')) \geq \gamma$, then on at least one of the pseudodistributions $\mu^\sym_1$ or $\mu^\sym_2$ Algorithm~\ref{alg:low-ent-restated} runs in time $\poly(|V(G)|)$ and returns an assignment of expected value at least 
\[
(\beta - \nu)^2\left(\gamma - O(\delta+\zeta) - \frac{1}{|V(G)|}\right) - 3\nu(\beta - \nu)
\]
for $I$.
\end{lemma}


\begin{proof}[Proof of Lemma~\ref{thm:round-j}]
Following the proof strategy of~\cite{BBKSS} we define the alternate shift potential $\Psi(\mu)$:
\begin{definition}[Alternate shift potential~\cite{BBKSS}]
The {\em alternate shift potential} of a degree-$D \ge 4$ pseudodistribution $\mu$ is given by
\[
\Psi(\mu) := \E_{u,v \sim \pi(G)} \left[\sum_{s \in \Sigma} \pPr_\mu[X_v - X_u = s]^2 \cdot \pE[\val_v(X) \mid X_v - X_u = s] \right],
\]
where $\pi(G)$ is the stationary measure on $G$ and $\val_v(X)$ denotes the value of vertex $v$.
\end{definition}

We will show that if $\Phi(\mu_1 \times \mu_2 | E)$ is large and $E$ does not introduce too many correlations, then either $\Psi(\mu_1)$ or $\Psi(\mu_2)$ must be large:
\begin{lemma}\label{lem:relating-ent-j}
Let $\mu_1 \times \mu_2$ be a degree $\widetilde{O}(D+1/\nu)$ pseudodistribution over $(X,X')$ and let $E(X,X')$ be a polynomial such that $\cA_I \vdash_D E(X,X') \in [0,1]$. Suppose we have that: 
\[\Pr_{u,v \sim V(G): u \neq v}[TV((Y_{u,v},Y'_{u,v})|E , (Y_{u,v},Y'_{u,v})) > \delta] \leq \zeta.\]

If the shift-partition potential of $\mu_1 \times \mu_2 |E(X,X')$ is large, then the alternate shift-potential of $\mu_1$ or $\mu_2$ must be large as well:
\[\Phi_{\beta,\nu}(\mu_1 \times \mu_2 | E(X,X')) \leq \frac{\Psi(\mu_1)}{2(\beta-\nu)^2} +\frac{\Psi(\mu_2)}{2(\beta-\nu)^2} + \frac{3\nu}{\beta - \nu} + 2\delta + 2\zeta + \frac{1}{|V(G)|}.\]
\end{lemma}
\begin{proof}
Deferred to Section~\ref{sec:lems}.
\end{proof}
Given a pseudodistribution $\mu$, we 
recall that $\mu^{sym}$ is the shift symmetrized pseudodistribution as in Section~\ref{sec:symm-prelims}.
We note that $\Psi(\mu) = \Psi(\mu^{\sym})$ because $\Psi$ as $\pPr_{\mu^\sym}[X_u - X_v = s] = \pPr_{\mu}[X_u - X_v = s]$ and $\pE_{\mu^\sym}[\val_v(X)|X_v - X_u = s] = \pE_{\mu}[\val_v(X)|X_v - X_u = s]$ for all $s$.  Therefore we can use the following lemma from~\cite{BBKSS} that shows that when the alternate shift potential of a shift-symmetric pseudodistribution is large, a single step of conditioning and rounding returns a solution of high objective value:
\begin{lemma}[Lemma 3.6~\cite{BBKSS}]
Let $I = (G, \Pi)$ be an affine instance of Unique Games over the alphabet $\Sigma$, let $\kappa>0$ and 
let $\mu$ be a degree-$4$ shift-symmetric pseudodistribution for $I$. If $\Psi(\mu) \geq \kappa$, then Algorithm~\ref{alg:low-ent-restated} returns a solution of expected value at least $\kappa$.
\end{lemma}

The statement now follows by applying the last lemma on either $\mu^\sym_1$ or $\mu^\sym_2$ (depending 
which one has a higher $\Psi$).
\end{proof}

 \subsubsection{Relating the potentials: Proof of Lemma~\ref{lem:relating-ent-j}}\label{sec:lems}

\begin{proof}[Proof of Lemma~\ref{lem:relating-ent-j}]
We begin by recalling that in the definition of $\Phi_{\beta,\eta}$, we used an $\nu$-additive polynomial approximation $p(x)$ of degree $\widetilde{O}(1/\nu)$ to the indicator function $\Ind[x \ge \beta]$ on the interval $x \in [0,1]$, guaranteed by Theorem~\ref{thm:step-approx}. We will use $E$ to denote the polynomial $E(X,X')$, $\cD$ to denote the pseudodistribution $\mu_1 \times \mu_2 | E$ and $\pi$ to denote the uniform (in general stationary) measure over $V(G)$. Let $p_u$ denote the polynomial $p_{\beta,\nu}(\val_u(X))$ and $p'_u = p_{\beta,\nu}(\val_u(X'))$. Recall also the collection of local distributions $\cL(\cdot, p_{\beta,\nu})$ (Definition~\ref{defn:local-dist}). We will overload the notation $p_u,p'_u$ to mean the variables from Definition~\ref{defn:local-dist} and the polynomial $p_{\beta,\nu}(\val_u(X))$, and the use should be clear from context. Expanding the definition of $\Phi_{\beta,\nu}(\cD)$ we get:
\begin{align}
    \Phi_{\beta,\nu}(\cD)
    &= \pE_{\cD}\left[\sum_{s \in \Sigma} \left(\E_{u \sim \pi} \Ind[X_u - X_u' = s]p_u p'_u\right)^2\right ] \nonumber\\
    &= \pE_{\cD}\left[\sum_{s \in \Sigma} \E_{u,v \sim \pi} \Ind[X_u - X_u' = X_v - X_v' = s]p_up_v p'_u p'_v \right]\nonumber\\
    &= \E_{u,v \sim \pi}\pE_{\cD}\left[ \Ind[X_u - X_v = X'_u- X'_v]\cdot p_u p_vp'_up'_v\right], \nonumber\\
    &= \E_{u,v\sim \pi}[\pE_{\mu_1 \times \mu_2}\left[ \Ind[X_u - X_v = X'_u- X'_v]\cdot p_u p_vp'_up'_v\right] + {\sf err}(u,v)],\label{eq:relating-potential}
\end{align}
for ${\sf err}(u,v)$ defined accordingly for each pair $(u,v)$ (note that we have switched from the pseudo-distribution $\mathcal{D}$ to the pseudo-distribution $\mu_1\times \mu_2$ in the last transition).

\paragraph{Bounding the error term in~\eqref{eq:relating-potential}.} 
For $u = v$, we get that ${\sf err}(u,v) = \pE_{\cD}[p_u^2 (p'_u)^2] - \pE_{\mu_1 \times \mu_2}[p_u^2 (p'_u)^2] \leq 1$. Let us bound ${\sf err}(u,v)$ for $u \neq v$.
Let $R(X_u,X_v,X'_u,X'_v,p_{u},p_{v}, p'_u,p'_v)$ be the polynomial $\Ind[X_u - X_v = X'_u- X'_v]p_u p_v p'_u p'_v$ ($p_u$ denotes the polynomial $p_{\beta,\nu}(\val_u(X))$). Here we will overload the notation of $p_u$ to also denote the indicator variables from Definition~\ref{defn:local-dist} and let $R(a_1,\ldots,a_4,b_1,\ldots,b_4)$ for $a_1,\ldots,a_4 \in \Sigma$ and $b_1,\ldots,b_4 \in \{0,1\}$ also denote the evaluation of $R(X_u,\ldots,p'_v)$ on the values $a_1,\ldots,b_4$: $\Ind[a_1 - a_2 = a_3- a_4]b_1 b_2b_3 b_4$. Using Definition~\ref{defn:local-dist} we have that when $u \neq v$: $\pE_{\cD}[R] = \pPr_{\cL(\cD,p_{\beta,\nu})}[X_u - X_v = X'_u - X'_v, p_u = p_v = p'_u = p'_v = 1]$ and similarly for the distribution $\mu_1 \times \mu_2$ in place of $\cD$. Therefore we get that for $u \neq v$: 
\begin{align*}
&{\sf err}(u,v)\\
=&\pE_{\cD}[R(X_u,\ldots,p'_v)] - \pE_{\mu_1 \times \mu_2}[R(X_u,\ldots,p'_v)] \\
=&\pPr_{\cL(\cD,p_{\beta,\nu})}[X_u - X_v = X'_u - X'_v, p_u = \ldots =p'_v = 1] - \pPr_{\cL(\mu_1 \times \mu_2,p_{\beta,\nu})}[X_u - X_v = X'_u - X'_v, p_u = \ldots = p'_v = 1]\\
=&\sum_{\substack{a_1,\ldots,a_4 \in \Sigma}}  R(a_1,\ldots,a_4, 1,\ldots,1)(\pPr_{\cL(\cD,p_{\beta,\nu})}[X_u = a_1,\ldots,p_u =1, \ldots] - \pPr_{\cL(\mu_1 \times \mu_2,p_{\beta,\nu})}[X_u = a_1,\ldots,p_u =1 \ldots])\\
&\leq \sum_{\substack{a_1,\ldots,a_4 \in \Sigma}} \left|\pPr_{\cL(\cD,p_{\beta,\nu})}[X_u = a_1,\ldots,p_u = 1,\ldots] - \pPr_{\cL(\mu_1 \times \mu_2,p_{\beta,\nu})}[X_u = a_1,\ldots,p_u = 1,\ldots] \right| \\
&= 2 TV((Y_{u,v},Y'_{u,v})|E , (Y_{u,v},Y'_{u,v})),
\end{align*}
where $Y_{u,v} = (X_u,X_v,p_u,p_v)$, $Y'_{u,v} = (X'_u,X'_v,p'_u,p'_v)$ as in the lemma, the distribution $(Y_{u,v},Y'_{u,v})|E$ refers to the joint distribution on these variables defined by the collection of local distributions $\cL(\mu_1 \times \mu_2|E, p_{\beta,\nu})$ and similarly $(Y_{u,v},Y'_{u,v})$ refers to the distribution defined by $\cL(\mu_1 \times \mu_2,p_{\beta,\nu})$ (Definition~\ref{defn:local-dist}). 

Combining, we get that: 
\begin{align}\label{eq:err-uv}
\E_{u,v \sim \pi}[{\sf err}(u,v)] &\leq \E_{u,v \sim \pi: u \neq v}[{\sf err}(u,v)] + \Pr_{u, v \sim \pi}[u = v]\E_{u}[{\sf err}(u,u)] \notag\\
&\leq \E_{u,v \sim \pi: u \neq v}[2 TV((Y_{u,v},Y'_{u,v})|E , (Y_{u,v},Y'_{u,v}))] + \frac{1}{|V(G)|} \notag\\
&\leq 2\delta + 2\zeta + \frac{1}{|V(G)|},
\end{align}
by the assumption in the lemma statement. 
\paragraph{Bounding the main term in~\eqref{eq:relating-potential}.} We now upper bound the first term in~\eqref{eq:relating-potential}. 
\begin{align*}
&\E_{u,v\sim \pi}[\pE_{\mu_1 \times \mu_{2}}\left[ \Ind[X_u - X_v = X'_u- X'_v]\cdot p_u p_vp'_up'_v\right] \\
&=\E_{u,v\sim \pi}\left[\sum_{s \in \Sigma} \pE_{\mu_1}[\Ind[X_u - X_v = s]\cdot p_u p_v]\pE_{\mu_2}[\Ind[X'_u - X'_v = s]\cdot p'_u p'_v]\right].
\end{align*}

We next use that $p(x),\val_u(X),\val_u(X') \in [0,1]$ and Fact~\ref{fact:bdd-markov} asserting that 
$p(x) \le \frac{x}{\beta-\nu} + \nu$ for all $x \in [0,1]$, and furthermore this is \sos-certifiable.
Thus, pulling out a factor of $p_v,p'_v$ and applying
Fact~\ref{fact:bdd-markov} to bound $p_u,p'_u$ we get 
that the first term on the right hand side of~\eqref{eq:relating-potential} is at most
\begin{align*}
&\E_{u,v\sim \pi}\left[\sum_{s \in \Sigma} \pE_{\mu_1}\left[\Ind[X_u - X_v = s]\cdot \frac{\val_u(X)}{\beta - \nu}\right]\pE_{\mu_{2}}\left[\Ind[X'_u - X'_v = s]\cdot \frac{\val_u(X')}{\beta - \nu}\right]\right] + \frac{3\nu}{\beta-\nu} \\
&\leq \frac{1}{2}\E_{u,v\sim \pi}\left[\sum_{s \in \Sigma} \pE_{\mu_1}\left[\Ind[X_u - X_v = s]\cdot \frac{\val_u(X)}{\beta - \nu}\right]^2\right] + \frac{1}{2}\E_{u,v\sim \pi}\left[\sum_{s \in \Sigma}\pE_{\mu_{2}}\left[\Ind[X'_u - X'_v = s]\cdot \frac{\val_u(X')}{\beta - \nu} \right]^2\right] \\
&+ \frac{3\nu}{\beta-\nu},
\end{align*}
where we have used that the AM-GM inequality. The first term may be bounded as
\begin{align*}
&\E_{u,v\sim \pi}\left[\sum_{s \in \Sigma} \pE_{\mu_1}\left[\Ind[X_u - X_v = s]\cdot \frac{\val_u(X)}{\beta - \nu}\right]^2\right] \nonumber\\
\leq&\E_{u,v\sim \pi}\left[\sum_{s \in \Sigma} \pE_{\mu_1}\left[\Ind[X_u - X_v = s]\cdot \frac{1}{\beta-\nu}\right] \pE_{\mu_1}\left[\Ind[X_u - X_v = s]\frac{\val_u(X)}{\beta - \nu}\right]\right] \nonumber \\
=&\frac{1}{(\beta-\nu)^2}\E_{u,v\sim \pi}\left[ \sum_{s \in \Sigma} \pE_{\mu_1}[\Ind[X_u - X_v = s]]^2 \cdot \pE_{\mu_1}[\val_u(X) \mid X_u - X_v = s] \right]\nonumber\\
=&\frac{\Psi(\mu_1)}{(\beta-\nu)^2},
\end{align*}
where we have used that $\val_u(X) \leq 1$ and applied the definition of conditional pseudoexpectation, 
and similarly the second term is upper bounded by $\frac{\Psi(\mu_{2})}{(\beta-\nu)^2}$. Plugging this and~\eqref{eq:err-uv} into~\eqref{eq:relating-potential} finishes the proof.
\end{proof}

\section{Unique Games with low completeness}\label{sec:johnson_2}
In this section we give an analysis which works for UG instances with arbitrary small completeness (but bounded away from $0$). The only step that changes is concluding that there is a subcube with large shift-potential (after conditioning). 

Recall that in the proof of the analogous lemma for large completeness (Lemma~\ref{lem:sp-j}) 
we used Claims~\ref{claim:subcube-expanse} and \ref{claim:potentials} to conclude that if a subcube has large global shift-partition potential it also has large shift-partition potential. This was possible 
because all the vertices with high global value will also have high value inside the subcube as the expansion of the subcube is small ($O(\eps)$ if the completeness was $1-\eps$). This fact is no longer true in the low completeness regime, as subcubes 
now may have expansion close to $1$.
To circumvent this issue, we no longer go via the intermediate global shift-partition route and 
directly try to conclude that the shift-partition potential is large on a subcube. Towards this end 
we need a stronger conclusion of global hypercontractivity which we refer to 
as an ``edge-covering'' statement: the subcubes on which some part of the shift partition is large cover nearly all of the the internal edges of the shift-partition. 

As we shall see, if the value of our pseudodistribution is originally $c$, the 
fraction of internal edges is at least $c^2$ (in particular $c^2$-fraction of edges are satisfied by both $X$ and $X'$), and as we are able to cover almost all internal edges we get that there exists a subcube with large completeness and in fact large shift-partition potential. Given this version of Lemma~\ref{lem:sp-j}, the rest of the analysis remains the same. 

\paragraph{Notation:}
For the purposes of this section we first consider the simpler shift-partition defined by functions $\{G_s: V(G) \rightarrow \R[X,X']\}_{s \in \Sigma}$:
\[G_s(u) = \one(X_u - X'_u = s).\]
Note that $G_s$'s cover all the vertices, that is $\sum_s \delta(G_s) = 1$ and satisfy Booleanity: $G_s(u)^2 = G_s(u)$ for all $u \in V(G)$.

Given two assignments $X,X'$, let $\val_u(X \wedge X')$ denote the fraction of edges incident on $u$ that are satisfied by \emph{both} $X$ and $X'$. Note that $\val_u(X \wedge X')$ is a degree four polynomial in $(X,X')$: 
\[\val_u(X \wedge X') = \E_{v: (u,v) \in E}[(\sum_{\sigma \in \Sigma} X_{u,\sigma}X_{v,\pi_{u,v}(\sigma)})(\sum_{\sigma \in \Sigma} X'_{u,\sigma}X'_{v,\pi_{u,v}(\sigma)})].\]
Let $\val_I(X \wedge X') = \E_u[\val_u(X \wedge X')]$ denote the fraction of edges in the graph that are satisfied by both $X$ and $X'$. Further for an $i$-restriction $a \in {[n] \choose i}$ and a vertex $u \in J|_a$ let $\val^a_u(X \wedge X')$ denote the fraction of edges incident on $u$ in $J|_a$ that are satisfied by both $X$ and $X'$. This is also a degree $4$ polynomial in $(X,X')$.

\begin{definition}[Approximate Indicator for Dense Subcubes]
Let $\eps_0,\ldots,\eps_r \in (0,1)$ to be determined later. For any $i$-restriction $a$ let $R_{s,a}$ denote the (approximate) indicator that $G_s$ is $\eps_i$-dense in $J|_a$, but not $\eps_{|b|}$-dense inside any subcube $J|_b$ where $b \subsetneq a$. To define this as a polynomial let us first define the approximate indicator $p_{<\beta,\nu}(x)$ as $1 - p_{\beta,\nu}(x)$. One can check that $p_{<\beta,\nu}$  approximates $\Ind[x < \beta]$ when $x \in [0,1]$ with similar properties as $p_{\beta,\nu}$ (Theorem~\ref{thm:step-approx}) that approximates $\Ind[x \geq \beta]$ for $x \in [0,1]$. Formally using the polynomial approximation for an indicator define 
\[R_{s,a} = p_{\eps_i,\nu}(\delta(G_s|_a))\prod_{j < i}\prod_{b \subset a: |b| = j}p_{<\eps_j,\nu}(\delta(G_s|_b)).\] 
\end{definition}

\subsection{The Edge Covering Theorem}
Our argument will need an upper bound on the number of
events $R_{s,a}$ that can occur simultaneously, which 
roughly speaking asks how many $i$-restricted subcubes can a given set $F$ be dense on. 
As stated, there is no good upper bound for this: 
if $F$ is dense on an $i-1$-restricted subcube
then it would be quite dense on many $i$-restricted subcubes containing it. This is the reason 
that in the event $R_{s,a}$, we required that the $i$-restricted subcube is dense but there is no $j$-restricted subcube ($j < i$) containing it on which our set is still somewhat dense.

\begin{claim}\label{claim:dense-subcubes}
If $\eps_0, \ldots, \eps_r$ satisfy that $\eps_{i-1} \leq \eps_i/(2^{i+1} i)$ then we get that:
\[\vdash_{\ell^{O(\ell^i/\nu)}} \E_{a \sim {[n] \choose i}}[R_{s,a}] \leq \frac{4\delta(G_s)}{\eps_i^2\ell^i} + O(\nu).\]
\end{claim}
\begin{proof}
We first present an argument when the polynomials $p_{\beta,\nu}$ in $R_{s,a}$ are replaced by indicators. We then give a sketch of how to convert the proof into an SoS proof. 

Converting the $p_{\beta,\nu}$'s to indicators, we get that $R_{s,a}$ corresponds to $T_{s,a}$: 
\begin{align*}
    T_{s,a} 
    &= \one(\delta(G_s|_a) \geq \eps_i) \prod_{j < i}\prod_{b \subset a: |b| 
    = j}\one(\delta(G_s|_b) < \eps_j)
\end{align*}
that is the event that $G_s$ is $\eps_i$-dense in $a$, but for all subsets $b$ of $a$, it is at most $\eps_j$-dense. We will show that the fraction of such restrictions $a$ must be small.

We consider the Fourier-analytic function $f_{i,F}$ defined in Section~\ref{sec:fourier}, Definition~\ref{def:f_i} and use its alternative formula from Lemma~\ref{lem:restriction}:
\[f_{i,F}(a) = \sum_{b \subseteq a} (-1)^{i -|b|}\delta(F|_b) \geq \delta(F|_a) - \sum_{b \subsetneq a} \delta(F|_b).\]
Technically we have only defined these functions for the Cayley-version of the Johnson graph $C_{n,\ell,\alpha}$, but one can use the above definition for Johnson graphs and derive the same properties that we use here upto $o(1)$ error terms. We will ignore these $o(1)$-error terms in this proof.

If $T_{s,a} = 1$ we get that, $\delta(F|_a) \geq \eps_i$, but each $\delta(F|_b)$ is at most $\eps_{|b|} < \eps_i/(2^{i+1} i)$ so we get:
\[f_{i,G_s}(a) \geq \eps_i - \sum_{j < i} {i \choose j}\eps_j \geq \eps_i - \sum_{j < i} 2^i \eps_j \geq \frac{\eps_i}{2}.\]
This immediately gives: 
\begin{equation}\label{eq:rsa1}
T_{s,a} \leq 1 \leq 4f_{i,G_s}(a)^2/\eps_i^2,
\end{equation}
if $T_{s,a} = 1$. If $T_{s,a} = 0$ then the above continues to hold.

On the other hand by the definition of $f_{i}$ Lemma~\ref{lem:3.2} gives that, 
\[\E[f_{i,G_s}(a)^2] = \frac{W^i[G_s]}{\ell^i} \leq \frac{\E[G_s^2]}{\ell^i}.\]
Combining with equation~\ref{eq:rsa1} we get:
\[\E_{a \sim {[n] \choose i}}[T_{s,a}] \leq \E_{a}\left[\frac{4f_{i,G_s}(a)^2}{\eps_i^2}\right] \leq \frac{4\delta(G_s)}{\eps_i^2\ell^i}.\]

\paragraph{SoS-ing the proof:} We sketch an SoS proof for the statement for $R_{s,a}$ with $a \in {[n] \choose 1}$. The full statement follows analogously. 

Let $\delta_0$ denote $\delta(G_s)$ and $\delta_1$ denote $\delta(G_s|_a).$ Analogous to \eqref{eq:rsa1} we will first show that:
\begin{equation}\label{eq:rsa-non-sos}
\frac{4(\delta_1 - \delta_0)^2}{\eps_1^2} + 2\nu - R_{s,a} > \nu
\end{equation}
if $\delta_1,\delta_0 \in [0,1]$. This will not be an SoS proof, but after showing this we can use a blackbox theorem to convert it into an SoS proof by the properties of our domain $[0,1]^2$. \eqref{eq:rsa-non-sos} follows by case analysis. 
\begin{enumerate}
    \item If $R_{s,a} > \nu$ it implies that $p_{\eps_1,\nu}(\delta_1) > \nu$ and $p_{<\eps_0,\nu}(\delta_0) > \nu$. By the definition of the polynomials $p_{\beta,\nu}$ this implies that $\delta_1 > \eps_1$ and $\delta_0 < \eps_0 < \eps_1/2$, which gives that $(\delta_1 - \delta_0)^2 \geq \eps_1^2/4$. Since $R_{s,a}$ is always $\leq 1$ we get:
    \[R_{s,a} \leq 1 \leq \frac{4(\delta_1 - \delta_0)^2}{\eps_1^2} < \frac{4(\delta_1 - \delta_0)^2}{\eps_1^2}+\nu.\]
    \item If $R_{s,a} \leq \nu$ we get:
    \[R_{s,a} \leq \frac{4(\delta_1 - \delta_0)^2}{\eps_1^2} + \nu.\]
\end{enumerate}
Hence rearranging the above and adding $\nu$ on both sides gives \eqref{eq:rsa-non-sos}. We can now apply Theorem~\ref{thm:blackbox-sos} that shows that positive polynomials bounded away from $0$ on the domain $[0,1]^k$ have a bounded degree SoS proof to get:
\[\delta_0,\delta_1 \in [0,1] \vdash_{\exp(1/\nu)} R_{s,a} \leq \frac{4(\delta_1 - \delta_0)^2}{\eps_1^2} + 2\nu.\]

One can check that the rest of the proof is already SoS therefore we get:
\[\cA_I \vdash_{\tO(1/\nu)} \delta_0,\delta_1 \in [0,1] \vdash_{\exp(1/\nu)} \E_{a \sim {[n] \choose 1}}[R_{s,a}] \leq \frac{4\delta(G_s)}{\eps_1^2\ell} + 2\nu.\]

One can do the exact same steps and use the SoS degree bound in Theorem~\ref{thm:blackbox-sos} to get an SoS proof of the full statement with degree $\ell^{O(\ell^i/\nu)}$. The bound that comes from Theorem~\ref{thm:blackbox-sos} is exponential in the degree of the polynomial $f$ that we want to prove is SoS, and hence our degree bound is also $\exp(\ell^i/\nu)$. We believe there could be a more direct analysis to get a degree bound of $\poly(1/\nu, \ell^i)$ instead.
\end{proof}

We now move on to stating the edge covering theorem. 
Below, we inspect the edges that stay inside a part in the shift partition, and show that they can nearly be encapsulated within subcubes on which some part becomes dense.
\begin{theorem}[Edge-Covering Theorem]\label{thm:edge-covering-restated}
For all constants $r \in \N$, $\eps_0, \ldots, \eps_r \in [0,1]$, where $\eps_r \leq \exp(-r)$ and $\eps_0, \ldots, \eps_r$ satisfy that $\eps_{i-1} \leq \eps_i^5/(2^{6r})$, we get that:
\begin{equation}\label{eq:edge-cov-final}
\cA_I \vdash_{\ell^{\poly(\ell^r/\nu)}}  \val_I(X \wedge X') \leq T_0(X,X') + \ldots + T_r(X,X') + \sf{err},
\end{equation}
where:
\begin{align*}
&T_0(X,X') = \sum_s p_{\eps_0,\nu}(\delta(G_s))\E_{u \sim G}[ G_s(u)\val_u(X \wedge X')],\\
&T_i(X,X') = \ell^{i}\E_{a \in {[n] \choose i}}[\sum_s R_{s,a}
\E_{u \sim J|_a}[G_s(u)\val^a_u(X \wedge X')]] \\
&{\sf err} = 4(1-\alpha)^{r+1} + 2^{6r}\max_{i}\frac{\eps_{i-1}}{\eps_i^4}+ |\Sigma|\ell^{O(r)}\sqrt{\nu}.
\end{align*}
\end{theorem}

\begin{proof}
We will first give a proof of this statement with the definition of $R_{s,a}$ replaced by actual indicator functions in place of the polynomials $p_{\beta,\nu}$. Formally, for an $i \leq r$-restriction $a \in {[n] \choose i}$ let $R_{s,a}$ be an indicator variable that is 1 if $\delta(G_s|_a) \geq \eps_i$ and $\delta(G_s|_b) \leq \eps_{|b|}$ for all $b \subsetneq a$. We will then give a sketch of how to convert this into an \sos proof and work with the true definition of $R_{s,a}$.

For any set $S$ let $\#sat(S)$ denote the number of edges inside $S$ that are satisfied by both $X$ and $X'$. We first have that:
\begin{equation}\label{eq:edge_cover0}
\val_{I}(X \wedge X') \leq \sum_s \frac{\#sat(G_s)}{|E|},
\end{equation}
since we know (from Fact~\ref{fact:z-vars}) that any edge that crosses between parts in the shift-partition must be violated by at least one of $X$ or $X'$. Let us therefore upper bound $\#sat(G_s)$.

Consider the set indicated by $G_s$ for fixed $s \in \Sigma$. Let the set of dense subcubes for $G_s$ be defined as follows:
\begin{align*}
&\cC^0_s = \begin{cases}\{J\} ~~~~~~~~~\text{if $\delta(G_s) \geq \eps_0$} \\
\phi ~~~~~~~~~~~~\text{otherwise}
\end{cases} \\
&\cC^{(1)}_s = \left\{a \in [n]: R_{s,a} = 1\right\} \\
&\vdots \\
&\cC^{(r)}_s = \left\{a \in {[n] \choose r}: R_{s,a} = 1\right\}. 
\end{align*}

Let $D_s$ denote the dense part of $G_s$: $\cup_{i \leq r} \cup_{a \in \cC^{(i)}_s} G_s \cap J|_a$ and let $H_s = G_s \setminus D_s$ be the pseudorandom part of $G_s$. It is easy to check that $H_s$ is $(r, \eps_r)$-pseudorandom. The edges inside $G_s$ can be divided into four types of edges: edges that stay inside the dense subcubes, edges that go between two different dense subcubes, edges inside $H_s$ and edges that go between $H_s$ and $D_s$. Applying this subdivision of edges on $\#sat(G_s)$ we get that:
\begin{align}
\frac{\#sat(G_s)}{|E(J)|} 
&\leq \sum_{i \leq r} \sum_{a \in \cC^{(i)}_s} \frac{\#sat(G_s \cap J|_a)}{|E(J)|}  +
\frac{|E(H_s, H_s)|}{|E(J)|} + \frac{|E(H_s, D_s)|}{|E(J)|}\notag\\\label{eq:edge-cov1}
+ 
&\sum_{i, j \leq r} 
\sum\limits_{h=0}^{\min(i,j)}
\sum_{a \in \cC^{(i)}_s}
\sum_{H\subseteq a, |H| = h}
\sum_{\substack{b \in \cC^{(j)}_s \\ a \neq b, a\cap b = H}} \frac{|E(J|_a, J|_b)|}{|E(J)|}.
\end{align}

We will check that the $i^{th}$-summand in the first term can be rearranged to give the $s^{th}$-summand in $T_i(X,X')$. We have that:
\begin{equation}\label{eq:edge-cov2}
\sum_{a \in \cC^{(i)}_s} \frac{\#sat(G_s \cap J|_a)}{|E(J)|} = \frac{{[n] \choose i}|E(J|_a)|}{|E(J)|}\E_{a \sim {[n] \choose i}}\left[R_{s,a}\frac{\#sat(G_s \cap J|_a)}{|E(J|_a)|}\right].
\end{equation}
By a direct computation we get that $\frac{{[n] \choose i}|E(J|_a)|}{|E(J)|} \leq \ell^i$. As for the second term above we get that, 
\begin{equation*}
\frac{\#sat(G_s \cap J|_a)}{|E(J|_a)|} \leq \frac{\sum_{u \in J|_a}\Ind[u \in G_s]\#sat_a(u)}{|V(J|_a)|d_a} = \E_{u \in J|_a}[\Ind[u \in G_s]\val_u^a(X \wedge X')],
\end{equation*}
where $d_a$ denotes the degree of a vertex inside the subgraph $J|_a$ and $\#sat_a(u)$ denotes the number of edges in $J|_a$ that are incident on $u$ and are satisfied by both $X$ and $X'$. So plugging this into \eqref{eq:edge-cov2} and summing up over $i \leq r$ we get that the first term in \eqref{eq:edge-cov1} gives us:
\begin{equation}\label{eq:first-term-ec}
\sum_{i \leq r} \sum_{a \in \cC^{(i)}_s} \frac{\#sat(G_s \cap J|_a)}{|E(J)|} \leq \sum_{i \leq r} \ell^i \E_{a \sim {[n] \choose i}}[R_{s,a}\E_{u \in J|_a}[\Ind[u \in G_s]\val_u^a(X \wedge X')]]. 
\end{equation}

We will bound the other terms in \eqref{eq:edge-cov1} by some small constant. 

\paragraph{Bounding the last term in~\eqref{eq:edge-cov1}.} 
We focus on $i\leq j$, and in the end multiply the 
bound we get by factor $2$. 

If $R_{s,\phi} = 1$, that is, $\delta(G_s) \geq \eps_0$ then we can check that this is $0$ since none of the other subcubes will have $R_{s,a} = 1$. So let us assume that this is not the case. 
Fix $i$ and $H\subseteq i$ of size $h$. For any $a \in {[n] \choose i}, b \in {[n] \choose j}$
intersecting in a set of size $h$ we have that
$\frac{|E(J|_a, J|_b)|}{|E(J)|} \leq \left(\frac{2\ell}{n}\right)^{i+j-h}$.
Indeed, to see this one can think of sampling an edge 
$(x,y)$ in $J$ uniformly at random, and asking what 
is the probability that $a\subseteq x$ and $b\subseteq y$. For that, we want $a\cap b$ to be contained in 
both $x$ and $y$ (which happens with probability at most $(2\ell/n)^h$), and that $a\setminus(a\cap b)\subseteq x$, $b\setminus(a\cap b)\subseteq y$
which happens with probability at most $(2\ell/n)^{i-h}\cdot (2\ell/n)^{j-h}$. 
Thus, 
we get that the last term in~\eqref{eq:edge-cov1}
is at most
\[
\sum_{i\leq j \leq r} 
\sum\limits_{h=0}^{i-1}
\sum_{a \in \cC^{(i)}_s}
\sum_{H\subseteq a, |H| = h}
|\{b~|b\in \cC^{(j)}_s, a\cap b = H, a\neq b\}|\left(\frac{2\ell}{n}\right)^{i+j-h}.
\]
Here we note that $h$ cannot be $i$, since then we 
would be looking at sub-cubes $J|_a\subseteq J|_b$, 
and so we cannot have both $R_{s,a} = 1$ and $R_{s,b}=1$.

Applying Claim~\ref{claim:dense-subcubes} on the 
restriction $G_s|_H$, we get that 
\[
|\{b~|b\in \cC^{(j)}_s, a\cap b = H, a\neq b\}|
\leq \frac{n^{j-h}}{\ell^{j-h}}\frac{\delta(G_s|_H)}{\eps_{j}^2}
\leq \frac{n^{j-h}}{\ell^{j-h}}\frac{\eps_h}{\eps_{j}^2}.
\]
We used the fact that as $R_{s,b} = 1$,  $G_s$ is at most $\eps_h$ dense in $J|_h$. Plugging that above we get that the last term 
on the right hand side of~\eqref{eq:edge-cov1} is 
at most
\[
\sum_{i\leq j \leq r} 
\sum\limits_{h=0}^{i-1}
\sum_{a \in \cC^{(i)}_s}
\sum_{H\subseteq a, |H| = h}
\frac{n^{j-h}}{\ell^{j-h}}\frac{\eps_h}{\eps_{j}^2}
\leq 
\sum_{i, j \leq r} 
2^i \cdot i\cdot \frac{n^{j-h}}{\ell^{j-h}}\frac{\eps_h}{\eps_{j}^2}
\cdot |\cC^{(i)}_s|\left(\frac{2\ell}{n}\right)^{i+j-h},
\]
and as $|\cC^{(i)}_s|\leq \frac{\delta(G_s)n^i}{\eps_i^2\ell^i}$ due to 
Claim~\ref{claim:dense-subcubes}, we get further upper
bound this by
\[
\sum_{i, j \leq r} 
2^i \cdot i\cdot \frac{n^{i+j-h}}{\ell^{i+j-h}}\frac{\eps_h\delta(G_s)}{\eps_{j}^2\eps_i^2}
\left(\frac{2\ell}{n}\right)^{i+j-h}
\leq
2^{3r+1} r^2\delta(G_s) \max_{i}\frac{\eps_{i-1}}{\eps_i^4}
\leq 2^{6r}\max_{i}\frac{\eps_{i-1}}{\eps_i^4}\delta(G_s).
\]
Summing this over $s$ gives that the contribution of the last term in~\eqref{eq:edge-cov1} to~\eqref{eq:edge_cover0} is at most
$2^{6r}\max_{i}\frac{\eps_{i-1}}{\eps_i^4}$.



\paragraph{Bounding the second and third term in \eqref{eq:edge-cov1}.} 
Here, we use Theorem~\ref{thm:structure-johnson}. 
The point is that the set $H_s$ is pseudo-random, 
and hence we can get good control over the edges from it to other sets. 

More precisely, for the second term note that $H_s$ is  $(r,\eps_r)$-pseudorandom  as per Definition~\ref{def:p.r.}, 
so by Theorem~\ref{thm:structure-johnson}) with $\gamma =\eps_r$ we get: 
\begin{equation}\label{eq:third-term-ec}
\frac{|E(H_s)|}{|E(J)|} = \delta(H_s) - \ip{H_s, LH_s} \leq 2\delta(H_s)(1-\alpha)^{r+1} \leq 2\delta(G_s)(1-\alpha)^{r+1},
\end{equation}
where the first inequality holds as long as $\eps_r \leq (1-\alpha)^{r+1}\exp(-r)$. Note that the second term in the structure theorem $q_a(H_s)(\delta(H_s|_a) - \eps_r)$ is non-positive since $q_a$ is \sos and $H_s$ is pseudorandom. The third term corresponding to Booleanity, $B(H_s)$ is $0$ since $H_s^3(u) = H_s(u)$.
Summing up over $s$ yields that the contribution of this to the right hand side of~\eqref{eq:edge_cover0} 
is at most $2(1-\alpha)^{r+1}$.

The bound on the third term in~\eqref{eq:edge-cov1} is similar. We have that
\begin{equation}\label{eq:fourth-term-ec}
\frac{|E(H_s,D_s)|}{|E(J)|} = \ip{D_s, J H_s} \leq \sqrt{\ip{D_s,D_s} \cdot \ip{H_s, J^2 H_s}} \leq 2\delta(G_s)(1-\alpha)^{r+1}.
\end{equation}
In the last inequality, we used $\ip{D_S,D_s}\leq \delta(G_s)$ and 
$\ip{H_s, J^2 H_s}\leq 2\delta(G_s)(1-\alpha)^{2(r+1)}$ 
which follows by Theorem~\ref{thm:structure-johnson} 
exactly as above. Summing up over $s$ yields total contribution of at most $2(1-\alpha)^{r+1}$.

\paragraph{SoS-ing the proof:} We give a proof sketch of how to convert this into an SoS proof. To avoid Booleanity error we will apply \eqref{eq:edge-cov1-sos} multiple times with $R_{s,a}$ replaced with $p_{\beta,\nu}$'s and take an average. We know that $R_{s,a} = p_1(\delta_1) \cdot \ldots\cdot p_t(\delta_t)$ where $p_i$'s are approximate indicators $p_{\beta_i,\nu}$ (or $p_{<\beta_i,\nu}$) and $\delta_i$'s are of the form $\delta(G_s|_b)$ for $b \subseteq a$. Then define $R^c_{s,a}$ as $p^c_1(\delta_1) \cdot \ldots\cdot p^c_t(\delta_t)$ where $p^c_i(\delta_i) = p_{\beta_i+c\nu^2,\nu^2}$ (resp. $p_{<\beta_i+c\nu^2,\nu^2}$) for all $c \in \{0,\ldots,\lfloor 1/\nu \rfloor\}$.

Note that $G_s$'s and the functions $\#sat(G_s \cap J|_a)$ are actual indicators as written since we can express $G_s(u) = \Ind[X_s - X'_u = s]$ and $\Ind[e \text{ is satisfied }]$ as a polynomial in $(X,X')$. Define $H^c_s(u) = G_s(u)\prod_{a \subset u: |a| \leq r}(1-R^c_{s,a})$ and $D^c_{s} = G_s(u)(1 - H^c_{s,a})$. Defining dense subcubes with $R^c_{s,a}$ we get the following analogue of \eqref{eq:edge-cov1}:
\begin{align}\label{eq:edge-cov1-sos}
\frac{\#sat(G_s)}{|E(J)|} &\leq \sum_{i \leq r} \sum_{a \in {[n] \choose i}}R^c_{s,a} \frac{\# sat(G_s \cap J|_a)}{|E(J)|} +
\ip{H^c_s, AH^c_s} + \ip{D^c_s, JH^c_s} \notag\\+&\sum_{i, j \leq r} 
\sum\limits_{h=0}^{\min(i,j)}
\sum_{a \in {[n] \choose i}} R^c_{s,a}
\sum_{H\subseteq a, |H| = h}
\sum_{\substack{b \in {[n] \choose j} \\ a \neq b, a\cap b = H}} R^c_{s,b} \frac{|E(J|_a, J|_b)|}{|E(J)|}.
\end{align}
We then take an expectation over $c \in \{0,\ldots,\lfloor 1/\nu \rfloor\}$ and then a sum over $s \in \Sigma$.

Now we can use the same analysis to upper bound each of these terms. The analysis for the first term follows as is; we get an $\E_c[R^c_{s,a}]$ (instead of $R_{s,a}$) but we will ignore this slight difference as it does not affect our final analysis of the algorithm.
For the second term we use the \sos version of Claim~\ref{claim:dense-subcubes} (as stated) to say that $\E_{a \in {[n] \choose i}}[R^c_{s,a}]$ is small and the same tricks like those in the SoS-ization of Claim~\ref{claim:dense-subcubes}. For the third and fourth terms we apply the structure theorem on the almost-Boolean function $H_s$ and bound the Booleanity term $\E_c[\sum_s B(H^c_s)]$ by $|\Sigma|\ell^{O(r)}\sqrt{\nu}$ instead. Additionally in the fourth term we apply an \sos version of Cauchy-Schwarz. We omit these details as we've carried out similar arguments before, but formally prove the upper bound on $\E_c[\sum_s B(H^c_s)]$ in Claim~\ref{claim:bool} in the appendix.
\end{proof}

\subsection{The Main Lemma}
In this section, we consider UG instances that have completeness bounded away from $1$ and possibly very close to $0$. The analysis of our algorithm differs only in one step: concluding that the shift potential is large.

The following claim asserts that taking two assignments $X, X'$ from our pseudo-distribution, 
the expected number of edges that are satisfied by 
both $X$ and $X'$ is still bounded away from $0$:
\begin{claim}\label{claim:val-intersection-restated}
$\pE_{\mu \times \mu}[\val_I(X \wedge X')] \geq c^2$.
\end{claim}
\begin{align*}
\pE_{\mu \times \mu}[\val_I(X \wedge X')] &= \pE_{\mu \times \mu}[\E_{e \sim E(G)}[\one(e \text{ satisfied by } X)\one(e \text{ satisfied by } X')]] \\
&= \E_{e \sim E(G)}[\pPr_\mu[e \text{ satisfied by } X]^2] \\
&\geq \E_{e \sim E(G)}[\pPr_\mu[e \text{ satisfied by } X]]^2 ~~~~~~~~~~\text{(by Cauchy-Schwarz)}\\
&= \val_\mu(I)^2.
\end{align*}


\begin{lemma}\label{lem:sp-j-low-comp}
For all positive constants $c, \alpha, \tau \leq 1$, $\nu = \frac{c^4}{|\Sigma|^2\ell^{O(r)}}$, all integers $\ell \geq \Omega(r), n \geq \ell$, where $r = \Theta(\frac{\log c}{\log(1-\alpha)})$ the following holds: Let $I$ be an affine UG instance on $J_{n,\ell,\alpha}$ and $\mu$ be a pseudodistribution over assignments for $I$ with $\val_\mu(I) \geq c$ and degree at least $ \frac{\ell^{O(r)}\poly(|\Sigma|)}{\poly(c)\tau}+ \ell^{\poly(|\Sigma|\ell^r/c)}$. Then there exists a restriction $a \subseteq [n]$ of size $i \leq r$, a degree $\widetilde{O}(1/\nu)$ polynomial $P_a(X,X')$ in a fixed set of $|\Sigma|$ polynomials, subsets $A, B \subseteq V(J|_a)$ of size at most $\frac{\ell^{O(i)}|\Sigma|^2}{c^2\tau}$ and strings $y_A, y_B$ such that conditioning $\mu$ on the events $Y_A = y_A$ and $Y_B = y_B$ gives degree $\widetilde{\Omega}(1/\nu)$ pseudodistributions $\mu_1$ and $\mu_2$ such that:
\begin{enumerate}
\item $\cA_I  ~~\vdash_{\widetilde{O}(1/\nu)} ~~ P_a(X,X') \in [0,1]$.
\item $\pE_{\mu_1 \times \mu_2}[\Phi^a_{\beta,\nu}(X,X')P_a(X,X')] \geq \frac{c^{\Omega(r)}}{\exp(r^2)} \pE[P_a(X,X')]$.
\item $\pE_{\mu_1 \times \mu_2}[P_a(X,X')] \geq \Omega\left(\frac{c^2}{8r\ell^i|\Sigma|}\right)$.
\item $\E_{\substack{u_1,v_1 \sim S: u_1 \neq v_1\\ u_2,v_2 \sim S: u_2 \neq v_2}}[I(Y_{u_1,v_1};Y_{u_2,v_2})] \leq \tau.$
\item $\E_{\substack{u_1,v_1 \sim S: u_1 \neq v_1\\ u_2,v_2 \sim S: u_2 \neq v_2}}[I(Y'_{u_1,v_1};Y'_{u_2,v_2})] \leq \tau,$
\end{enumerate}
where $\beta = \left(\frac{\poly(c)}{\exp(r)}\right)^{r-i+1}$, $S = V(J|_a)$, $Y_{u,v} = (X_u,X_v,p_u,p_v)$, $Y'_{u,v} = (X'_u,X'_v,p'_u,p'_v)$ and the mutual information is taken with respect to the collection of local distributions $\cL(\mu_1 \times \mu_2,p_{\beta,\nu})$. 
\end{lemma}

\begin{proof}
Our strategy will be to use the edge-covering theorem. We will apply $\pE_{\mu \times \mu}$ on both sides of the inequality and use Claim~\ref{claim:val-intersection-restated} to say that $\pE[\val_I(X \wedge X')]$ is large. 

We will set parameters in the end of the proof such that: $r$ is large enough so that $4(1-\alpha)^{r+1} \leq c^2/6$, $\nu$ is small enough such that $|\Sigma|\ell^{O(r)}\sqrt{\nu} \leq c^2/6$ and $\eps_i$'s satisfy $2^{6r}\max_{i}\frac{\eps_{i-1}}{\eps_i^4} \leq c^2/6$ and further $\eps_i$'s satisfy the conditions of Theorem~\ref{thm:edge-covering-restated}. So we get that the error term in Theorem~\ref{thm:edge-covering-restated} is small, 
\begin{equation}\label{eq:error}
{\sf err} = 4(1-\alpha)^{r+1} + 2^{6r}\max_{i}\frac{\eps_{i-1}}{\eps_i^4}+ |\Sigma|\ell^{O(r)}\sqrt{\nu} \leq \frac{c^2}{2}.
\end{equation}

Applying $\pE_{\mu \times \mu}$ on both sides of Theorem~\ref{thm:edge-covering-restated} (assuming $\deg(\mu) \geq \ell^{\poly(\ell^r/\nu)}$) and 
using Claim~\ref{claim:val-intersection-restated} to lower bound pseudoexpectation of the LHS we get:
\[
c^2
\leq
\pE_{\mu \times \mu}[\val_I(X \wedge X')]
\leq \sum_{i \in [r]} \pE_{\mu \times \mu}[T_i(X,X')]+{\sf err}
\leq \sum_{i \in [r]} \pE_{\mu \times \mu}[T_i(X,X')] + \frac{c^2}{2},\]
Therefore we get that there is an $i \in [r]$ for which $\pE_{\mu \times \mu}[T_i(X,X')] \geq c^2/2r$, and so
\begin{equation}\label{eq:low-comp1}
\pE_{\mu \times \mu}\E_{a \sim {[n] \choose i}}[\sum_s R_{s,a} \E_{u \sim J|_a}[G_s(u)\val^a_u(X \wedge X')]] \geq \frac{c^2}{2r\ell^i}.
\end{equation}

We aim to condition $\pE_{\mu \times \mu}$ on an appropriate polynomial (as in the proof of Lemma~\ref{lem:sp-j-restated}) so that the resulting pseudoexpectation operator satisfies: $\pE'[\E_{u \sim J|_a}[G_s(u)\val^a_u(X \wedge X')]]] \geq \Omega(c^2/r)$ for some specific $i$-restriction $a$. To do so we apply $\pE_{\mu \times \mu}$ to Claim~\ref{claim:dense-subcubes} and sum up over $s$, to get
\[\pE_{\mu \times \mu}\left[\E_{a \sim {[n] \choose i}}\sum_s R_{s,a}\right] \leq \frac{4}{\eps_i^2 \ell^i} + O(|\Sigma|\nu) \leq \frac{5}{\eps_i^2\ell_i},\]
where for the last inequality we have assumed $\nu \leq O(1/\ell^i|\Sigma|)$.

We split the RHS in equation~\ref{eq:low-comp1} to $c^2/4r\ell^i + c^2/4r\ell^i$, where in the first term we further use the above equation: $1/\ell^i \geq \frac{\eps_i^2}{5}\pE_{\mu \times \mu}[\E_{a \sim {[n] \choose i}}[\sum_s R_{s,a}]]$ to get that,
\begin{equation*}
\pE_{\mu \times \mu}\E_{a \sim {[n] \choose i}}\left[\sum_s R_{s,a} \E_{u \sim J|_a}[G_s(u)\val^a_u(X \wedge X')]\right] \geq \frac{c^2\eps_i^2}{20r}\pE_{\mu \times \mu}\left[\E_{a \sim {[n] \choose i}}[\sum_s R_{s,a}]\right] + \frac{c^2}{4r\ell^i},
\end{equation*}
and by re-arranging 
\begin{equation*}
\pE_{\mu \times \mu}\E_{a \sim {[n] \choose i}}\left[\sum_s R_{s,a} \left(\E_{u \sim J|_a}[G_s(u)\val^a_u(X \wedge X')] - \frac{c^2\eps_i^2}{20r}\right)\right] \geq \frac{c^2}{4r\ell^i}.
\end{equation*}
By averaging we get an $a \in {[n] \choose i}, s \in \Sigma$ such that:
\begin{equation}\label{eq:low-comp2}
\pE_{\mu \times \mu}\left[R_{s,a} \left(\E_{u \sim J|_a}[G_s(u)\val^a_u(X \wedge X')] - \frac{c^2\eps_i^2}{20r}\right)\right] \geq \frac{c^2}{4r\ell^i|\Sigma|}.
\end{equation}
The last result is the same as equation~\eqref{eq:sp-j-2}
in the proof of Lemma~\ref{lem:sp-j-restated}, 
which is the analogous statement for the completeness s clsoe to $1$ case. The rest of the proof is essentially the same as the rest of the proof of Lemma~\ref{lem:sp-j-restated}, and is outlined below. 

Applying Lemma~\ref{lem:ragh-tan-appln} we get pseudodistributions $\mu_1,\mu_2$ with low global correlation (as required in points 4,5 of the Lemma) with:
\begin{equation}\label{eq:low-comp3}
\pE_{\mu_1 \times \mu_2}\left[R_{s,a} \left(\E_{u \sim J|_a}[G_s(u)\val^a_u(X \wedge X')] - \frac{c^2\eps_i^2}{20r}\right)\right] \geq \frac{c^2}{8r\ell^i|\Sigma|}.
\end{equation}

Letting $P_a(X,X') = R_{s,a}(X,X')$ it is easy to check that:
\[\cA_I \vdash_{\tO(\exp(i)/\nu)} P_a(X,X') \in [0,1], \text{ and }~~~ \pE_{\mu_1 \times \mu_2}[P_a(X,X')] \geq \frac{c^2}{8r\ell^i|\Sigma|},\]
which proves points 1 and 3 of the Lemma statement. The only thing remaining to check is that conditioning $\mu_1 \times \mu_2$ on $R_{s,a}$ results in large shift-partition potential on $J|_a$. Recall that the shift-partition potential on $J|_a$ is: $\phi^a_{\beta,\nu}(X,X') = \sum_s \E_{u \sim J|_a}[G_s(u)p_{\beta,\nu}(\val^a_u(X))p_{\beta,\nu}(\val^a_u(X'))]^2$. From equation~\eqref{eq:low-comp3} we have that conditioning $\mu_1 \times \mu_2$ on $R_{s,a}$ gives:
\[\pE_{\mu_1 \times \mu_2 | R_{s,a}}[\E_{u \sim J|_a}[G_s(u)\val^a_u(X \wedge X')]] \geq  \frac{c^2\eps_i^2}{20r}.\]

We can prove that for $x \in [0,1] \vdash_{\tO(1/\nu)} x \leq  p_{\beta,\nu}(x) + \beta + O(\nu)$ (using the standard machinery in Section~\ref{sec:apx-ind}). Applying this with $\val^a_u(X \wedge X')$ and $\beta = \frac{c^2\eps_i^2}{40r}$, and further noting that
$p_{\beta,\nu}(\val_u^a(X \wedge X')) \leq p_{\beta,\nu}(\val_u^a(X))p_{\beta,\nu}(\val_u^a(X')) + O(\nu)$ we get:
\[\pE_{\mu_1 \times \mu_2 | R_{s,a}}[\E_{u \sim J|_a}[G_s(u)p_{\beta,\nu}(\val_u^a(X))p_{\beta,\nu}(\val_u^a(X'))]] \geq  \frac{c^2\eps_i^2}{40r} - O(\nu).\]

Applying Cauchy-Schwarz we get:
\[\pE_{\mu_1 \times \mu_2 | R_{s,a}}[\E_{u \sim J|_a}[G_s(u)p_{\beta,\nu}(\val_u^a(X))p_{\beta,\nu}(\val_u^a(X'))]^2] \geq  \Omega\left(\frac{c^4\eps_i^4}{r^2}\right).\]

By adding the terms corresponding to other $s \in \Sigma$ we immediately get:
\begin{equation}\label{eq:potential-c0}
\pE_{\mu_1 \times \mu_2 | R_{s,a}}[\phi_{\beta,\nu}^a(X,X')] \geq \Omega\left(\frac{c^4\eps_i^4}{r^2}\right).
\end{equation}

We finish the proof by setting $r,\eps_i,\nu$ according to all the requirements above.

\paragraph{Setting parameters:} To make $\sf{err}$ in \eqref{eq:error} small we set $r = \Theta(\frac{\log c}{\log(1-\alpha)})$, $\nu = \frac{c^4}{|\Sigma|^2\ell^{O(r)}}$, $\eps_{j-1} = \eps_j^5/2^{6r}$ with $\eps_r = \Theta(c^2\exp(-r))$. We need $\beta = c^2\eps_i^2/40r$ which equals $(\poly(c)/\exp(r))^{r-i+1}$. We can set $\deg(\mu) = \ell^{\poly(\ell^r/\nu)} = \ell^{\poly(|\Sigma|\ell^r/c)}$ so that it is valid to apply the edge-covering theorem. Additionally when we apply Lemma~\ref{lem:ragh-tan-appln} we need $\deg(\mu) \geq \Omega(\log |\Sigma|/p\tau\nu)$ where $p = c^2/4r\ell^i|\Sigma|$ so it suffices to take $\deg(\mu) = \frac{\ell^{O(r)}\poly(|\Sigma|)}{\poly(c)\tau}$. One can check that these parameters are good enough to make the rest of the statements in the proof work out. From \eqref{eq:potential-c0} we get that the conditioned shift-partition potential is at least $\Omega(c^2\eps_i^4/r^2) \geq \frac{c^{O(r)}}{\exp(r)}$ and this completes the proof.
\end{proof}

\subsection{Proof of Theorem~\ref{thm:main_johnson}}
The algorithm for arbitrary completeness is the same as that for completeness close to $1$ albeit with different parameter settings for $r, \delta,D,\gamma$ for the main algorithm and $t,\tau,\beta,\nu$ for SubRound. Given Lemma~\ref{lem:sp-j-low-comp} the rest of the analysis of the algorithm remains exactly the same. We set these parameters below but omit the description of the algorithm. The guarantees of this algorithm are potentially worse (running time $\approx n^{\ell^{\poly(\ell)}}$ and approximation factor $\approx c^{O(\log 1/c)}$) than the one for completeness close to $1$, but the important thing to note is that the running time is polynomial in $n$ (as long as $\ell, |\Sigma|$ are $O(1)$) and the approximation factor only depends on $c$ and not on $\ell,n$ or $|\Sigma|$.

\begin{theorem}\label{thm:main-johnson-arbit-comp}
For all positive constants $c \in (0,1]$, $\alpha \in \Q$, all integers $\ell \geq \Omega(r)$ with $\alpha \ell \in \N$ and $r =  \Theta(\frac{\log c}{\log(1-\alpha)})$ and $n$ large enough, Algorithm~\ref{alg:j} (with modified parameter settings given below) has the following guarantee: If $I$ is an instance of affine Unique Games on the $(n,\ell,\alpha)$-Johnson graph $J$ with alphabet $\Sigma$ and $\val(I) = c$, then in time $|V(J)|^{D}$ Algorithm~\ref{alg:j} returns an $\Omega\left(\frac{c^{\Omega(r)}}{\exp(r^2)}\right)$-satisfying assignment for $I$, with $D = \frac{\ell^{\poly(|\Sigma|\ell^r/c)}}{c^{O(r)}}$.
\end{theorem}

\begin{proof}
Let $r = \Theta(\frac{\log c}{\log(1-\alpha)})$. Given the $i \leq r$-restricted subcube $a \subseteq [n], P_a(X,X')$ and $\mu_1 \times \mu_2$ from Lemma~\ref{lem:sp-j-low-comp} we have that $\Phi_{\beta,\nu}^a(\mu_1 \times \mu_2| P_a) \geq \frac{c^{\Omega(r)}}{\exp(r^2)}$, therefore we apply Lemma~\ref{lem:correlation} with $S = J|_a$, the polynomial  $E(X,X') = P_a(X,X')$ and the parameters $\beta = \left(\frac{\poly(c)}{\exp(r)}\right)^{r-i+1}, \nu = \frac{c^4}{|\Sigma|^2\ell^{O(r)}}, \delta = \frac{c^{O(r)}}{\exp(r^2)}, p = \Omega\left(\frac{c^2}{8r\ell^i|\Sigma|}\right)$ and $\tau = \frac{c^{O(r)}}{\exp(r^2)\ell^{O(i)}|\Sigma|}$. The parameter $\tau$ has been chosen so that $O(\frac{\sqrt{\tau}+1/|S|}{p\delta^2}) = \frac{c^{O(r)}}{\exp(r^2)}$ and therefore we can apply the rounding lemma (Lemma~\ref{lem:round-j}) with $\zeta = \frac{c^{O(r)}}{\exp(r^2)}$ and $\gamma = \frac{c^{\Omega(r)}}{\exp(r^2)}$ and the same settings of $\beta, \nu, \delta$. This shows that there exists a subcube $a \in {[n] \choose i}$ with SubRound value that is at least $\frac{c^{\Omega(r)}}{\exp(r^2)}$ if $\deg(\mu) \geq  \frac{\ell^{O(r)}\poly(|\Sigma|)}{\poly(c)\tau}+ \ell^{\poly(|\Sigma|\ell^r/c)}$. Hence it suffices to have degree of $\mu$ equal to $D = \frac{\ell^{\poly(|\Sigma|\ell^r/c)}}{c^{O(r)}}$.

By using SubRound as a subroutine, we finish the proof of this theorem by applying the iteration Lemma~\ref{lem:subroutine} with completeness $c$, $\gamma = c^2$, $\delta = \frac{c^{\Omega(r)}}{\exp(r^2)}$ and $r = \Theta(\frac{\log c}{\log(1-\alpha)})$ so that $(1-\alpha)^r = \poly(c)$. We can check that SubRound (with the parameter settings above) satisfies the hypotheses of the lemma: it finds a $\le r$-restricted subcube and an assignment to it with value $\Omega(\delta)$ in time $|V(J)|^{D}$ (where $D$ is the degree bound above). This gives us that Algorithm~\ref{alg:j} outputs an assignment of value at least $\Omega\left(\frac{c^{\Omega(r)}}{\exp(r^2)}\right)$ in time $|V(G)|^{O(D)}$.
\end{proof}

\section{Affine Unique Games on other Globally Hypercontractive Graphs}\label{sec:other_graphs}
In this section, we explain how to adapt our techniques from Sections~\ref{sec:johnson_1},~\ref{sec:johnson_2} for 
general globally hypercontractive graphs. We first give a semi-formal definition of global hypercontractivity and go over the steps of our algorithm, noting that even though we stated it for Johnson graphs earlier, it works for general graphs as long as we have an SoS certificate of global hypercontractivity. Later we give the specifics for the proofs of Theorems~\ref{thm:main_grass} and~\ref{thm:main_hdx}: UG algorithms for the Grassmann graph and random walks over HDXs.

\subsection{The Components that Go Into the Algorithm}\label{sec:abstraction}
The $4$-components that are used in 
our algorithm proving Theorem~\ref{thm:main_johnson} 
and the interaction between them can be abstracted as follows.
\begin{enumerate}
    \item \textbf{Global hypercontractivity.}\label{def:global-hypcont} 
    The first component going into our algorithm is a global hypercontractive inequality (or a consequence of it).
    This is a result asserting that 
    for our underlying graph $G$, for all $\gamma>0$ there are $r\in\mathbb{N}$, 
    $\eps_0,\ldots,\eps_r,\zeta,\eta>0$ and a collection 
    $\mathcal{C} = \mathcal{C}_0\cup\ldots\cup\mathcal{C}_r$ of sets of vertices (these are the combinations 
    of the basic sets) satisfying the following properties: 
    \begin{enumerate}
        \item 
        \textbf{$\mathcal{C}$ consists of non-expanding sets:} 
        The edge expansion of each $C\in\mathcal{C}$ in $G$ is $1-\zeta$.
        \item 
        \textbf{$\mathcal{C}$ explains all small sets that 
        do not expand well.} A basic form of this property 
        is that sets $S\subseteq V(G)$ that are not dense in 
        any set in $\mathcal{C}$, in the sense that $\delta(S\cap C)\leq \eps_1\delta(C)$
        for all $C\in\mathcal{C}$, have expansion close to $1$, that is, $\Phi(S)\geq 1-\gamma$.
        
        A stronger version of this property, that often follows from the above (as is the case of all graphs of interest in this paper) and that is used in the proof of Theorem~\ref{thm:main_johnson_close_to_1}, is that 
        one can capture a constant fraction of $S$ using $C\in\mathcal{C}$ in which $S$ is dense. 
        Equivalently, this property says that for a set of vertices $S\subseteq V(G)$, taking $\mathcal{C}' = \{C\in\mathcal{C}~|~C\in\mathcal{C}_i,~\delta(S\cap C)\geq \eps_i \delta(C)\}$ and defining $C' = \bigcup_{C\in \mathcal{C}'}C$, if $\delta(S\cap C')\leq \eta \delta(S)$ then
        $\Phi(S)\geq 1-\gamma$.
        
        An even stronger form of global hypercontractivity, that again follows from the most basic form in our cases of interest, is the edge covering theorem. 
        We used this form of global hypercontracitivity in the proof of Theorem~\ref{thm:main_johnson}. 
        Informally, this result says that not only a substantial chunk of $S$ can be covered by non-expanding sets from $\mathcal{C}$ in which $S$ is dense, but in fact one can capture almost all of the edges staying inside $S$ using $C\in\mathcal{C}$. More precisely, 
        this result asserts that 
        we may take $\mathcal{C}''\subseteq \mathcal{C}'$ 
        such that the following properties hold:
        \begin{enumerate}
            \item Edge covering: 
        sampling an edge $(u,v)$ inside $S$, we have that 
        except with probability $\xi$ it holds that $(u,v)$ 
        is an edge inside some $C''\subseteq \mathcal{C}''$.
            \item Maximally dense: in words, this means that 
            every $C''\in\mathcal{C}''$ has no $C\supsetneq C''$ 
            in $\mathcal{C}$ in which $S$ is somewhat dense. 
            Formally, we require that for all $i=1,\ldots,r$, 
            $C''\in\mathcal{C}''\cap \mathcal{C}_i$, 
            $j<i$ and $C_j\in\mathcal{C}_j$ such that $C_{j}\supseteq C''$, it holds that 
            $\delta(S\cap C_j) < \eps_j \delta(C_j)$. 
        \end{enumerate}
    \end{enumerate}
    The edge covering theorem allows us to argue that in our shift partition
    $F_s(X,X') = \{v\in V~|~ X(v) - X'(v) = s\}$ we may capture
    almost all edges that stay within some $F_s$ using the basic 
    sets from $\mathcal{C}$. In particular, since edges that both $X$ and $X'$ satisfy constitute a constant fraction of edges and they all stay within some $F_s$, it means that we may capture almost all of them using sets from $\mathcal{C}$. 
    \item  By elementary arguments, it follows
    that there is such a $C\in\mathcal{C}$ and $s$ such 
    that $F_s$ 
    becomes dense in $C$ with noticeable probability and inside $C$ a constant fraction of the edges are satisfied by both $X$ and $X'$. In our argument we consider this event 
    (this was approximated by the polynomial $P_a(X,X')$ in our algorithm to stay within 
    the realm of \sos). As we explained in the introduction
    this conditioning may create correlation between $X$ and 
    $X'$ that would cause our rounding technique not to work. 
    
    To circumvent this, we showed that: 
    \begin{enumerate}
        \item \textbf{Eliminating global correlations:} using conditioning (as in~\cite{RT12}) we are able to make sure that the two collection 
    of random variables $\{Z_{u,v}\}_{u,v\in C}, \{Z'_{u,v}\}_{u,v\in C}$ of interest 
    (that are concerned with local information about the assignments $X$ and $X'$ on vertices $u,v$) are roughly
    pairwise independent.
        \item \textbf{Conditioning mostly preserves independence:} As the probability of the event that
        $C$ becomes dense in $F_s$ is noticeable, and $\{Z_{u,v}\}_{u,v\in C}$,$\{Z'_{u,v}\}_{u,v\in C}$ were pairwise-independent to begin with, we argue 
        that it cannot create too many dependencies between
        $(Z_{u,v}, Z'_{u,v})$.
    \end{enumerate} 
    
    \item \textbf{Rounding.} We then gave a rounding procedure on a subcube $C$ and a pseudodistribution $\cD$, where on average (over $\cD$) our supposed assignments $X,X'$ satisfy a constant fraction of the edges (formalized via the notion of
    shift-partition potential), some part $F_s$ of the shift partition is large, and on which we have near 
    independence between certain random variables 
    $\{(Z_{u,v},Z'_{u,v})\}_{u,v}$. The output of the rounding procedure is an assignment to the vertices of $C$ that satisfies a constant fraction of the constraints inside it.
    
    \item \textbf{An iteration result.} By that, we mean a result saying that if we have a procedure as above that is able to find a $C\in\mathcal{C}$ and find an assignment to the vertices of $C$ that satisfies a constant fraction of the constraints in it, then we can iterate it to satisfy a constant fraction 
    of the constraints in the whole graph. 
    
    This is achieved by randomizing the edges incident to $C$, which ensures that any subsequent $C'$ we find can only have negligible overlap with the $C$'s we have found thus far.
    The idea is that any $C$ that we found has at least 
    $\delta$ of the edges touching it staying inside it, 
    out of which we manage to satisfy a constant fraction, 
    hence overall we manage to satisfy a constant fraction
    of the edges touching $C$. In return for that, 
    we randomize the edges touching $C$ hence effectively giving up on the objective value coming from those edges. The point is that since the ratio between the number of edges we manage to satisfy, and the number of edges we give up on is $\Omega(1)$, after a few iterations we will manage to satisfy a 
    constant fraction of the edges before we drop the objective value of our instance by too much.
\end{enumerate}

Our of these $4$ components, the second, third and 
fourth components work as is for any of the graphs that we 
care about. Indeed, the elimination of correlations and 
conditioning are completely generic, and the rounding 
procedure only hinges on the constraints being affine. 
As for the fourth component, it only relies on the properties
of the family $\mathcal{C}$ guaranteed from our structure 
theorem and otherwise is also generic.
\skipi
We next discuss the first component, which is a globally 
hypercontractive inequality, and for concreteness we 
consider the most basic version of it. This is a result
saying that a set $S$ which is not dense inside any 
$C\in \mathcal{C}$ has large expansion, and in 
all of our graphs of interest such results are known
~\cite{KMMS,KhotMS18,filmus2020hypercontractivity,keevash2021global,bafna2022hypercontractivity,gur2022hypercontractivity}. 
For such a result to give us an algorithm, we also require 
the proof of this part to be in the \sos proof system of
constant degree (which is the case in all of the above
examples, but in principle may not be the case). 

In the case of the Johnson graph the collection 
$\mathcal{C}$ is simply the collection of basic sets 
\[
\mathcal{C} = \left\{J|_a ~|~|a|\leq r\right\},
\qquad
\text{where }
J|_a = \left\{u\in\binom{[n]}{\ell}~\Big|~u\supseteq a\right\},
\]
the basic form of the global hypercontractive inequality is 
from~\cite{KMMS} (in the appendix we give an \sos version of 
that proof), and we have shown in Section~\ref{sec:johnson_2} how to conclude from it the edge covering theorem in a black-box manner.

\subsection{The Algorithm for the Grassmann Graph}
We now explain the first component above 
for the Grassmann graph, which is established in~\cite{KhotMS18}. Recall that the Grassmann graph 
${\sf Grass}(n,\ell)$ is the graph whose vertices are all
$\ell$-dimensional subspaces $L\subseteq \mathbb{F}_2^n$, 
and $(L,L')$ is an edge if ${\sf dim}(L\cap L') = \ell-1$.

Inside the Grassmann graph ${\sf Grass}(n,\ell)$ we have the two basic sets
\[
Z_x = \left\{L~|~x\in L\right\}\text{ for $x\in \mathbb{F}_2^n$},
\qquad
Z'_W = \left\{L~|~L\subseteq W\right\}\text{ for 
a hyperplane $W\subseteq \mathbb{F}_2^n$},
\]
and using them we can define the collections $\mathcal{C}_0,\ldots,\mathcal{C}_r$ and $\mathcal{C} = \mathcal{C}_0\cup\ldots\cup\mathcal{C}_r$ as above. Indeed, 
for a subspace $X\subseteq \mathbb{F}_2^n$ we define 
$Z_X = \cap_{x\in X}Z_x$ and for a subspace $W\subseteq \mathbb{F}_2^n$ we define $Z'_W=\cap_{W'\supseteq W} Z_{W'}$ 
and then
\[
\mathcal{C}_{m} = 
\bigcup_{m_1+m_2 = m}
\left\{Z_{X}\cap Z'_{W}~\Big|{\sf dim}(X) = m_1, {\sf dim}(W) = n-m_2, X\subseteq W\right\}.
\]
The result of~\cite{KhotMS18} shows that for all $\gamma>0$, 
there are $r\in\mathbb{N}$ and $\eps>0$ 
such that if $S$ is a set of 
vertices in ${\sf Grass}(n,\ell)$ such that $\delta(S\cap C)\leq \eps \delta(C)$ for all $C\in\mathcal{C}$, then $\Phi(S)\geq 1-\gamma$. In fact, the same proof shows that a slightly more general statement holds for bounded functions, 
asserting that if $F\colon V({\sf Grass}(n,\ell))\to [0,1]$ satisfies that 
$\E_{x}[F(x)^2 1_{x\in C}]\leq \eps \E_{x}[1_{x\in C}]$, then 
\[
\langle F, T F\rangle\geq (1-\gamma)\E_x[F^2(x)].
\]
where $T$ is the normalized adjacency matrix of ${\sf Grass}(n,\ell)$. Moreover their proof can be seen 
to be \sos-able since each step uses polynomial inequalities like the Parceval, Cauchy-Schwarz and H\"{o}lder's inequality on the variables $F$. We have used precisely these inequalities in Section~\ref{sec:fourier} to get an SoS certificate of global hypercontractivity for the Johnson graph (Theorem~\ref{thm:structure-johnson}). In fact, \cite{KhotMS18}'s proof is very similar to that of~\cite{KMMS} for the Johnson graph therefore one can similarly get an SoS certificate for the Grassmann.

From this result, one can follow along the lines of Theorem~\ref{thm:edge-covering-restated} to deduce, 
in a black-box manner, an edge covering theorem for 
the Grassmann graph: 

\begin{theorem}
For all $\eta>0$ there 
are $r\in\mathbb{N}$ and
$\eps_0,\ldots,\eps_r>0$ such that 
given a set 
$S\subseteq {\sf Grass}(n,\ell)$ 
one may find $\mathcal{C}_i'\subseteq \mathcal{C}_i$ 
such that 
\begin{enumerate}
    \item For each $i$ and $C\in\mathcal{C}_i'$ it holds 
    that $\delta(S\cap C)\geq \eps_i\delta(C)$.
    \item For $\tilde{S} = \bigcup_{i=0}^{r} \bigcup_{C\in \mathcal{C}_i} C$, one has 
    \[
    \Pr_{(L,L')\in E({\sf Grass}(n,\ell))}\left[L,L'\in S\land 
    (L\not\in\tilde{S}\text{ or }L'\not\in\tilde{S})\right]\leq \eta\delta(S).
    \]
    \item For all $j<i$, $C_i\in\mathcal{C}_i'$ and 
    $C_j\in \mathcal{C}_j$ such that $C_j\supseteq C_i$ 
    it holds that $\delta(S\cap C_j)\leq \eps_j \delta(C_j)$.
\end{enumerate}
\end{theorem}
\begin{proof}[Proof sketch]
We let $\eps_0,\ldots,\eps_r$ to be determined; it will only
matter to us that $\eps_r$ is small enough. 
Take $\mathcal{C}_i'$ to be all $C_i\in \mathcal{C}_i$ 
such that $\delta(S\cap C_i)\geq \eps_i \delta(C_i)$ and 
for which there are no $j$ and $C_j\in\mathcal{C}_j$ 
such that $C_j\supseteq C_i$ and $\delta(S\cap C_j)\geq \eps_j \delta(S)$.

Let $F = 1_S$, $\tilde{F} = 1_{\tilde{S}}$, 
and note that $G = F(1-\tilde{F}) = 1_{S\cap \overline{\tilde{S}}}$ 
is $(r,\eps_r)$ pseudo-random as per the definition of~\cite{KhotMS18}, so $\Phi(S\cap \overline{\tilde{S}})\geq 1-\eta^2/100$ provided that $\eps_r$ is sufficiently small. Note that
\[
\Pr_{(L,L')\in E({\sf Grass}(n,\ell))}\left[L,L'\in S\land 
    (L\not\in\tilde{S}\text{ or }L'\not\in\tilde{S})\right]
    \leq 2\E_{(L,L')}[F[L]F[L'](1-\tilde{F}[L])],
\]
which is equal to $2\langle G, T F\rangle$, where 
$T$ is the adjacency operator of the Grassmann graph with 
self loops. 
We may bound:
\begin{align*}
\langle G, T F\rangle
=
\langle T G, F\rangle
\leq \|T G\|_2\|F\|_2
= \sqrt{\langle T G, TG\rangle} \|F\|_2
&\leq \sqrt{\langle G, TG\rangle} \|F\|_2\\
&=\|G\|_2\sqrt{1-\Phi(S\cap\overline{\tilde{S}})} \|F\|_2.
\end{align*}
We used the fact that $T$ is positive semi-definite. 
Clearly $\|G\|_2,\|F\|_2\leq \sqrt{\delta(S)}$ hence we get
$\langle G, T F\rangle\leq \delta(S)\frac{\eta}{10}$, 
and plugging that above gives the second item.
\end{proof}
Adjusting this result to the setting of the shift 
partition and phrasing it as an \sos statement 
in the standard way, one gets an analog of 
Theorem~\ref{thm:edge-covering-restated} for the 
Grassmann graph. For that, one has to address 
edges going across different $C_i$'s and 
show that they have a small contribution. 
For that one has to choose 
$\eps_{i}$ to be sufficiently smaller than 
$\eps_{i+1}$, and include in $\mathcal{C}_i'$ 
only $C_i\in\mathcal{C}_i$ that are maximally dense, and repeat a calculation analogous to the one in Theorem~\ref{thm:edge-covering-restated} when bounding the last term therein.

The rest of our algorithm then proceeds in exactly the same way.

\subsection{The algorithm for High Dimensional Expanders}\label{sec:ug-on-hdx}
In this section we consider UG instances defined over higher-order walks on two-sided local spectral expanders and discuss why our results for Johnson graphs (that can be seen as partial-swap walks over the complete complex) generalize to HDXs. We refer the reader to~\cite{dikstein2018boolean} for an excellent exposition on HDXs. 


For concreteness, let us consider a consider a two-sided local spectral expanding complex $X$ and the partial-swap walk over $X(\ell)$ with depth $\alpha$. This random-walk naturally corresponds to a graph $G$ over the vertices $X(\ell)$. Given $G$, analogous to the subcubes in the Johnson graph, there are a family of basic sets called the links of the complex. \cite{bafna2022hypercontractivity},~\cite{gur2022hypercontractivity} generalized the result of~\cite{KMMS} to prove global hypercontractivity on HDXs and in particular showed that any non-expanding set in $G$ must have large constant density inside a link. The parameters they get are the same as that in the statement of~\cite{KMMS} for the $\alpha$-noisy Johnson graphs. In fact the proof of~\cite{bafna2022hypercontractivity} proceeded exactly along the lines of~\cite{KMMS} and therefore can be easily seen to be SoS-able to get an analogue of Theorem~\ref{thm:structure-johnson} for HDXs. One can similarly also get the stronger edge-covering theorem for HDXs (analogous to Theorem~\ref{thm:edge-covering-restated}) and we omit the details here.

Fortunately for us,~\cite{DBLP:conf/soda/BafnaHKL22} generalized the algorithm of~\cite{BBKSS} to get a UG algorithm for HDXs albeit suffered the same two drawbacks as~\cite{BBKSS} -- an $\ell$-dependent soundness guarantee and reliance on completeness being close to $1$.~\cite{DBLP:conf/soda/BafnaHKL22} thus show the generality of the~\cite{BBKSS} framework and illustrate that it does not depend on the symmetry properties/regularity of Johnson graphs, since higher-order walks could be over irregular sparse and highly asymmetric graphs. We can similarly generalize the Johnson result to get an algorithm for UG on HDXs via the steps outlined in Section~\ref{sec:abstraction}. To summarize, component 1 therein -- an SoS certificate of global hypercontractivity for HDXs -- can be obtained using~\cite{bafna2022hypercontractivity}, components 2 and 3 work for all graphs, and  component 4 is an iteration result that we can obtain, akin to the iteration result in~\cite{DBLP:conf/soda/BafnaHKL22}.


\section*{Acknowledgements}
We thank Boaz Barak for insightful discussions and encouragement during initial stages of this project. 

\bibliographystyle{amsalpha}
\bibliography{references}

\newcommand{\etalchar}[1]{$^{#1}$}
\providecommand{\bysame}{\leavevmode\hbox to3em{\hrulefill}\thinspace}
\providecommand{\MR}{\relax\ifhmode\unskip\space\fi MR }
\providecommand{\MRhref}[2]{%
  \href{http://www.ams.org/mathscinet-getitem?mr=#1}{#2}
}
\providecommand{\href}[2]{#2}
\begin{thebibliography}{DKK{\etalchar{+}}18b}

\bibitem[AIMS10]{DBLP:journals/eccc/AroraIMS10}
Sanjeev Arora, Russell Impagliazzo, William Matthews, and David Steurer,
  \emph{Improved algorithms for unique games via divide and conquer}, Electron.
  Colloquium Comput. Complex. \textbf{17} (2010), 41.

\bibitem[AKK{\etalchar{+}}08]{AroraKKSTV08}
Sanjeev Arora, Subhash Khot, Alexandra Kolla, David Steurer, Madhur Tulsiani,
  and Nisheeth~K. Vishnoi, \emph{Unique games on expanding constraint graphs
  are easy: extended abstract}, Proceedings of the 40th Annual {ACM} Symposium
  on Theory of Computing, Victoria, British Columbia, Canada, May 17-20, 2008,
  2008, pp.~21--28.

\bibitem[Aus07]{Austrin}
Per Austrin, \emph{Balanced max 2-sat might not be the hardest}, Proceedings of
  the 39th Annual {ACM} Symposium on Theory of Computing, San Diego,
  California, USA, June 11-13, 2007 (David~S. Johnson and Uriel Feige, eds.),
  {ACM}, 2007, pp.~189--197.

\bibitem[BBK{\etalchar{+}}21]{BBKSS}
Mitali Bafna, Boaz Barak, Pravesh~K. Kothari, Tselil Schramm, and David
  Steurer, \emph{Playing unique games on certified small-set expanders}, {STOC}
  '21: 53rd Annual {ACM} {SIGACT} Symposium on Theory of Computing, Virtual
  Event, Italy, June 21-25, 2021, {ACM}, 2021, pp.~1629--1642.

\bibitem[BDH{\etalchar{+}}20]{BakshiDHKKK20}
Ainesh Bakshi, Ilias Diakonikolas, Samuel~B. Hopkins, Daniel Kane, Sushrut
  Karmalkar, and Pravesh~K. Kothari, \emph{Outlier-robust clustering of
  gaussians and other non-spherical mixtures}, 61st {IEEE} Annual Symposium on
  Foundations of Computer Science, {FOCS} 2020, Durham, NC, USA, November
  16-19, 2020, {IEEE}, 2020, pp.~149--159.

\bibitem[BDJ{\etalchar{+}}22]{BakshiDJKKV22}
Ainesh Bakshi, Ilias Diakonikolas, He~Jia, Daniel~M. Kane, Pravesh~K. Kothari,
  and Santosh~S. Vempala, \emph{Robustly learning mixtures of \emph{k}
  arbitrary gaussians}, {STOC} '22: 54th Annual {ACM} {SIGACT} Symposium on
  Theory of Computing, Rome, Italy, June 20 - 24, 2022, {ACM}, 2022,
  pp.~1234--1247.

\bibitem[BHKL22a]{DBLP:conf/soda/BafnaHKL22}
Mitali Bafna, Max Hopkins, Tali Kaufman, and Shachar Lovett, \emph{High
  dimensional expanders: Eigenstripping, pseudorandomness, and unique games},
  Proceedings of the 2022 {ACM-SIAM} Symposium on Discrete Algorithms, {SODA}
  2022, Virtual Conference / Alexandria, VA, USA, January 9 - 12, 2022, 2022,
  pp.~1069--1128.

\bibitem[BHKL22b]{bafna2022hypercontractivity}
Mitali Bafna, Max Hopkins, Tali Kaufman, and Shachar Lovett,
  \emph{Hypercontractivity on high dimensional expanders}, Proceedings of the
  54th Annual ACM SIGACT Symposium on Theory of Computing, 2022, pp.~185--194.

\bibitem[BKS14]{BarakKS14}
Boaz Barak, Jonathan~A. Kelner, and David Steurer, \emph{Rounding
  sum-of-squares relaxations}, Symposium on Theory of Computing, {STOC} 2014,
  New York, NY, USA, May 31 - June 03, 2014, 2014, pp.~31--40.

\bibitem[BKS17]{BKS17}
Boaz Barak, Pravesh~K. Kothari, and David Steurer, \emph{Quantum entanglement,
  sum of squares, and the log rank conjecture}, Proceedings of the 49th Annual
  {ACM} {SIGACT} Symposium on Theory of Computing, {STOC} 2017, Montreal, QC,
  Canada, June 19-23, 2017, 2017, pp.~975--988.

\bibitem[BKS22]{BuhaiKS22}
Rares{-}Darius Buhai, Pravesh~K. Kothari, and David Steurer, \emph{Algorithms
  approaching the threshold for semi-random planted clique}, CoRR
  \textbf{abs/2212.05619} (2022).

\bibitem[BRS11]{BarakRS11}
Boaz Barak, Prasad Raghavendra, and David Steurer, \emph{Rounding semidefinite
  programming hierarchies via global correlation}, {IEEE} 52nd Annual Symposium
  on Foundations of Computer Science, {FOCS} 2011, Palm Springs, CA, USA,
  October 22-25, 2011, 2011, pp.~472--481.

\bibitem[BS14]{BarakS14}
Boaz Barak and David Steurer, \emph{Sum-of-squares proofs and the quest toward
  optimal algorithms}, arXiv preprint arXiv:1404.5236 (2014).

\bibitem[DDFH18]{dikstein2018boolean}
Yotam Dikstein, Irit Dinur, Yuval Filmus, and Prahladh Harsha, \emph{Boolean
  function analysis on high-dimensional expanders}, Approximation,
  Randomization, and Combinatorial Optimization. Algorithms and Techniques
  (APPROX/RANDOM 2018), Schloss Dagstuhl-Leibniz-Zentrum fuer Informatik, 2018.

\bibitem[DGJ{\etalchar{+}}10]{diakonikolas2010bounded}
Ilias Diakonikolas, Parikshit Gopalan, Ragesh Jaiswal, Rocco~A Servedio, and
  Emanuele Viola, \emph{Bounded independence fools halfspaces}, SIAM Journal on
  Computing \textbf{39} (2010), no.~8, 3441--3462.

\bibitem[DKK{\etalchar{+}}18a]{DBLP:conf/stoc/DinurKKMS18a}
Irit Dinur, Subhash Khot, Guy Kindler, Dor Minzer, and Muli Safra, \emph{On
  non-optimally expanding sets in grassmann graphs}, Proceedings of the 50th
  Annual {ACM} {SIGACT} Symposium on Theory of Computing, {STOC} 2018, Los
  Angeles, CA, USA, June 25-29, 2018 (Ilias Diakonikolas, David Kempe, and
  Monika Henzinger, eds.), {ACM}, 2018, pp.~940--951.

\bibitem[DKK{\etalchar{+}}18b]{DinurKKMS18}
Irit Dinur, Subhash Khot, Guy Kindler, Dor Minzer, and Muli Safra,
  \emph{Towards a proof of the 2-to-1 games conjecture?}, Proceedings of the
  50th Annual ACM SIGACT Symposium on Theory of Computing, 2018, pp.~376--389.

\bibitem[FKLM20]{filmus2020hypercontractivity}
Yuval Filmus, Guy Kindler, Noam Lifshitz, and Dor Minzer,
  \emph{Hypercontractivity on the symmetric group}, arXiv preprint
  arXiv:2009.05503 (2020).

\bibitem[FKP19]{TCS-086}
Noah Fleming, Pravesh Kothari, and Toniann Pitassi, \emph{Semialgebraic proofs
  and efficient algorithm design}, Foundations and Trends® in Theoretical
  Computer Science \textbf{14} (2019), no.~1-2, 1--221.

\bibitem[GLL22]{gur2022hypercontractivity}
Tom Gur, Noam Lifshitz, and Siqi Liu, \emph{Hypercontractivity on high
  dimensional expanders}, Proceedings of the 54th Annual ACM SIGACT Symposium
  on Theory of Computing, 2022, pp.~176--184.

\bibitem[Kho02]{Khot02}
Subhash Khot, \emph{On the power of unique 2-prover 1-round games}, Proceedings
  on 34th Annual {ACM} Symposium on Theory of Computing, May 19-21, 2002,
  Montr{\'{e}}al, Qu{\'{e}}bec, Canada, 2002, pp.~767--775.

\bibitem[Kho10]{KhotICM}
\bysame, \emph{On the unique games conjecture (invited survey)}, Proceedings of
  the 25th Annual {IEEE} Conference on Computational Complexity, {CCC} 2010,
  Cambridge, Massachusetts, USA, June 9-12, 2010, {IEEE} Computer Society,
  2010, pp.~99--121.

\bibitem[KKMO07]{KhotKMO07}
Subhash Khot, Guy Kindler, Elchanan Mossel, and Ryan O'Donnell, \emph{Optimal
  inapproximability results for {MAX-CUT} and other 2-variable csps?}, {SIAM}
  J. Comput. \textbf{37} (2007), no.~1, 319--357.

\bibitem[KLLM21]{keevash2021global}
Peter Keevash, Noam Lifshitz, Eoin Long, and Dor Minzer, \emph{Global
  hypercontractivity and its applications}, arXiv preprint arXiv:2103.04604
  (2021).

\bibitem[KMMS18]{KMMS}
Subhash Khot, Dor Minzer, Dana Moshkovitz, and Muli Safra, \emph{Small set
  expansion in the johnson graph}, Electron. Colloquium Comput. Complex.
  \textbf{{TR18-078}} (2018).

\bibitem[KMS17]{KhotMS17}
Subhash Khot, Dor Minzer, and Muli Safra, \emph{On independent sets, 2-to-2
  games, and grassmann graphs}, Proceedings of the 49th Annual {ACM} {SIGACT}
  Symposium on Theory of Computing, {STOC} 2017, Montreal, QC, Canada, June
  19-23, 2017, 2017, pp.~576--589.

\bibitem[KMS18]{KhotMS18}
\bysame, \emph{Pseudorandom sets in grassmann graph have near-perfect
  expansion}, 59th {IEEE} Annual Symposium on Foundations of Computer Science,
  {FOCS} 2018, Paris, France, October 7-9, 2018, 2018, pp.~592--601.

\bibitem[KNS10]{KNSgroth}
Guy Kindler, Assaf Naor, and Gideon Schechtman, \emph{The {UGC} hardness
  threshold of the \emph{L\({}_{\mbox{p}}\)} grothendieck problem}, Math. Oper.
  Res. \textbf{35} (2010), no.~2, 267--283.

\bibitem[KO18]{kaufman2018construction}
Tali Kaufman and Izhar Oppenheim, \emph{Construction of new local spectral high
  dimensional expanders}, Proceedings of the 50th Annual ACM SIGACT Symposium
  on Theory of Computing, 2018, pp.~773--786.

\bibitem[KR08]{KhotR08}
Subhash Khot and Oded Regev, \emph{Vertex cover might be hard to approximate to
  within 2-epsilon}, J. Comput. Syst. Sci. \textbf{74} (2008), no.~3, 335--349.

\bibitem[Las01]{Lasserre00}
Jean~B. Lasserre, \emph{Global optimization with polynomials and the problem of
  moments}, SIAM J. Optim. \textbf{11} (2000/01), no.~3, 796--817. \MR{1814045}

\bibitem[LSV05]{lubotzky2005explicit}
Alexander Lubotzky, Beth Samuels, and Uzi Vishne, \emph{Explicit constructions
  of ramanujan complexes of type ad}, European Journal of Combinatorics
  \textbf{26} (2005), no.~6, 965--993.

\bibitem[MM10]{MMlocalexpander}
Konstantin Makarychev and Yury Makarychev, \emph{How to play unique games on
  expanders}, Approximation and Online Algorithms - 8th International Workshop,
  {WAOA} 2010, Liverpool, UK, September 9-10, 2010. Revised Papers (Klaus
  Jansen and Roberto Solis{-}Oba, eds.), Lecture Notes in Computer Science,
  vol. 6534, Springer, 2010, pp.~190--200.

\bibitem[MSS16]{ma2016polynomial}
Tengyu Ma, Jonathan Shi, and David Steurer, \emph{Polynomial-time tensor
  decompositions with sum-of-squares}, 2016 IEEE 57th Annual Symposium on
  Foundations of Computer Science (FOCS), IEEE, 2016, pp.~438--446.

\bibitem[OZ13]{ODonnellZ13}
Ryan O'Donnell and Yuan Zhou, \emph{Approximability and proof complexity},
  Proceedings of the Twenty-Fourth Annual {ACM-SIAM} Symposium on Discrete
  Algorithms, {SODA} 2013, New Orleans, Louisiana, USA, January 6-8, 2013,
  2013, pp.~1537--1556.

\bibitem[Par00]{Parrilo00}
Pablo~A Parrilo, \emph{Structured semidefinite programs and semialgebraic
  geometry methods in robustness and optimization}, Ph.D. thesis, California
  Institute of Technology, 2000.

\bibitem[Rag08]{Raghavendra08}
Prasad Raghavendra, \emph{Optimal algorithms and inapproximability results for
  every csp?}, Proceedings of the 40th Annual {ACM} Symposium on Theory of
  Computing, Victoria, British Columbia, Canada, May 17-20, 2008, 2008,
  pp.~245--254.

\bibitem[Raz98]{Raz}
Ran Raz, \emph{A parallel repetition theorem}, {SIAM} J. Comput. \textbf{27}
  (1998), no.~3, 763--803.

\bibitem[RS10]{RaghavendraS10}
Prasad Raghavendra and David Steurer, \emph{Graph expansion and the unique
  games conjecture}, Proceedings of the forty-second ACM symposium on Theory of
  computing, 2010, pp.~755--764.

\bibitem[RSS18]{RaghavendraSS18}
Prasad Raghavendra, Tselil Schramm, and David Steurer, \emph{High-dimensional
  estimation via sum-of-squares proofs}, pp.~3389--3423, World Scientific,
  2018.

\bibitem[RST12]{RaghavendraST12}
Prasad Raghavendra, David Steurer, and Madhur Tulsiani, \emph{Reductions
  between expansion problems}, Proceedings of the 27th Conference on
  Computational Complexity, {CCC} 2012, Porto, Portugal, June 26-29, 2012,
  2012, pp.~64--73.

\bibitem[RT12]{RT12}
Prasad Raghavendra and Ning Tan, \emph{Approximating csps with global
  cardinality constraints using sdp hierarchies}, Proceedings of the
  twenty-third annual ACM-SIAM symposium on Discrete Algorithms, SIAM, 2012,
  pp.~373--387.

\bibitem[Sch04]{schweighofer}
Markus Schweighofer, \emph{On the complexity of schm{\"u}dgen's
  positivstellensatz}, Journal of Complexity \textbf{20} (2004), no.~4,
  529--543.

\bibitem[Tre12]{trevisan2012khot}
Luca Trevisan, \emph{On khot's unique games conjecture.}, Bulletin (New Series)
  of the American Mathematical Society \textbf{49} (2012), no.~1.

\end{thebibliography}

\appendix

\section{Proof of the Iteration Lemma}\label{sec:pf_it}
We give a proof of the iteration lemma. Let $\val_f(H)$ denote the fraction of edges inside $H$ that are satisfied by the assignment $f$.

\begin{lemma}\label{lem:subroutine-restated}
Let $c,\gamma,\delta \in (0,1]$ and $r \in \nn$. Let $I$ be an affine UG instance on alphabet $\Sigma$ on $J_{n,\ell,\alpha}$ with $\ell, n$ large enough  and value at least $c$. 
Suppose we have a subroutine $\cA$ which given as input any affine UG instance $I'$ on $J_{n,\ell,\alpha}$ with $\val_\mu(I') \ge c-\gamma$, returns an $\leq r$-restricted subcube $H$ on $J$ and a partial assignment $f$ such that, $\text{val}_{f}(H) \geq \delta$.
Then if $\cA$ runs in time $T(\cA)$, there is a $O(|V(J)|T(\cA)+|V(J)|^3)$-time algorithm which finds a solution for $I$ that satisfies an $\Omega(\delta\gamma(1-\alpha)^r)$-fraction of the  edges of $J$. 
\end{lemma}

\begin{proof}
We will use the algorithm $\cA$ as a subroutine. To get a full assignment, our algorithm below is a generalized version of the Algorithm~\ref{alg:j}, where we've replaced some of the steps in Algorithm~\ref{alg:j} with an arbitrary subroutine $\cA$ that finds a subcube and an assignment to it with high value. We include it here for completeness.

\begin{algorithm-thm}[Partial to Full Assignment]\label{alg:partial} ~\\ 
\begin{compactenum}
\item Set $j = 1, I_0 = I, R_0 = \phi$.
\item While $|R_{j-1}| < \frac{\gamma}{2}|V(J)|$:
\begin{compactenum}
\item Run subroutine $\cA$ on $I_{j-1}$ to find an $\leq r$-restricted subcube $H_{j}$ and partial assignment $f_j$.
\item Let $S_j$ be the induced subgraph of $H_j$ induced by the set $H_j \setminus R_{j-1}$ and assign them using $f_j$. Let $R_j = R_{j-1} \cup H_j$. 
\item Randomize the constraints on edges incident on $R_j$ and let the new instance by $I_j$. 
\item Increment $j$.
\end{compactenum}
\item Output any assignment to $V(G)$ that agrees with all partial assignments $f_j$ (assigned to $S_j$'s) considered above.
\end{compactenum}
\end{algorithm-thm}

Let us analyse the $j^{th}$-iteration of this procedure. We are given an instance $I_{j-1}$.

\begin{claim}\label{claim:it-drop}
$\val(I_{j-1}) \geq \val(I) - \frac{2|V(R_{j-1})|}{|V(G)|}$.
\end{claim}

\begin{proof}
Suppose $X$ is an assignment that achieves the value of $I$. We can check that $X$ satisfies at least $\val(I) - \frac{2|V(R_{j-1})|}{|V(G)|}$-fraction of edges of $I_{j-1}$, hence the claim follows. To see this note that the fraction of edges incident on $R_{j-1}$ and therefore the fraction that is randomized is at most $\frac{2|V(R_{j-1})|}{|V(G)|}$ and every edge not incident on $R_{j-1}$ remains satisfied by $X$.
\end{proof}

We will show that the partial assignment $f_j$ to $H_j \setminus R_{j-1}$ satisfies a large fraction of edges (as counted in the whole graph). Before that let us show that the randomization of edges ensures that the value of the edges randomized is small (even when restricted to some subcube). Let $E(S)$ denote the edges incident on $S$ (including those with only one endpoint in $S$) and let $E(S,S)$ denote the edges with both endpoints in $S$. For a set of edges $E$ let $\val_f(E)$ denote the fraction of edges in $E$ that are satisfied by $f$.

\begin{claim}\label{claim:chernoff}
With probability $1-o(1)$ the instance $I_{j-1}$ satisfies the following for all $r$-restricted subcubes $C$:
\[\val(E(C \cap R_{j-1})) \leq \frac{2}{|\Sigma|},\]
where $\val(E(S))$ denotes the maximum fraction of $E(S)$ that can be satisfied by any UG assignment to the endpoints of $E(S)$.
\end{claim}

\begin{proof}
Let $k$ denote $|\Sigma|$. Fix an $i$-restricted subcube $C$ for $i \leq r$. We know that $R_{j-1}$ is a union of $\leq r$-restricted subcubes, therefore $C\cap R_{j-1}$ also forms a union of $\leq r$-restricted subcubes inside $C$. Let this collection of subcubes be denoted by $\cC$.
Consider the set of edges $E$ inside $C$ on which the constraints were randomized (in the previous iteration), and let $A$ be the set of vertices in $C$ that are endpoints of them. We claim that with probability $1-o(1)$,  for all assignments $F$ to $A$, $F$ satisfies at most $2/k$ of the constraints. First, the number of assignments is $k^{|A|}$, so we fix one assignment $F$, analyze it, and then union bound over all assignments.

Fix $F$; for each $e\in E$, the probability that $F$ satisfies $e$ is $1/k$. Hence in expectation $F$ satisfies
$\mu = 1/k$ fraction of the constraints, and by Chernoff’s bound
\[\Pr[\val_F(E) \geq 2\mu ]\leq e^{-\frac{|E|}{3k}}.\]

Now we lower bound $|E|$. Note that if $v\in A$, then either $v$ was in a subcube in $\cC$ or $v$ was adjacent to a vertex in a subcube from $\cC$. Let the degree of $v$ inside $C$ be denoted by $d$ which we know is ${\ell-i \choose \alpha(\ell - i)}{n-\ell \choose \alpha(\ell-i)}$.

\begin{enumerate}
    \item If $v \in C'$ for some $C' \in \cC$ then at least $(1-\alpha)^{r}$ of the edges incident to $v$ remain in $C'$ and therefore are in $E$. 
    \item If $v$ is adjacent to $u \in R_{j-1}\cap C$ then $u \in C'$ for some $C' \in \cC$. In particular there is a subset $I \subseteq u$ of size at most $r$ such that $C'$ is an $I$-restricted subcube. It follows that in fact $v$ has at least ${\ell-i \choose \alpha(\ell-i)}{n-\ell-r \choose \alpha(\ell-i) – r}$ neighbours in $C'$. By a simple calculation this is at least $n^{-r} d$, and it follows that at least $n^{-r} d$ of the edges adjacent to $v$ are in $E$.
\end{enumerate}

In conclusion, in any case we get that at least $n^{-r}d$ of the edges adjacent to $v \in A$ are inside $E$, hence $|E| \geq n^{-r} d |A|/2 \gg 3n|\Sigma| |A|$, so Chernoff gives a probability of $e^{-n|\Sigma| |A|}$ which is good enough for
a union bound over all ($k^{|A|}$ many) assignments to $A$. Additionally we union bound over all $\leq r$-restricted subcubes $C$ which are only $n^{O(r)}$ many, to get the statement of the lemma.
\end{proof}

At iteration $j$ of the while-loop, since the while condition is met, we know that $I_{j-1}$ has value $\geq c - \gamma$ (Claim~\ref{claim:it-drop}) so inside the while-loop, $\cA$ will always find an $\leq r$-subcube $H_j$ and assignment $f_j$ with $\val_{f_j}(H_j) \geq \delta$. Next, we find an assignment $f_j$ to the set of vertices $V(S_j)$ that by definition don't intersect previously assigned vertices. Since $f_j$ doesn't reassign any vertices, in the final step of the algorithm it is possible to output an assignment that is consistent with all previously considered partial assignments. We will now show that our final partial assignment satisfies a large fraction of the edges, where we say that an edge $(u,v)$ is satisfied by a partial assignment $f_j$, if both vertices $u,v$ have been assigned labels under $f_j$ and the labels satisfy the edge.

\begin{claim}\label{claim:it-val}
The value of the partial assignment found at iteration $j$ satisfies:
\[\val(f_j) \geq \delta(1 - \alpha)^r\frac{|V(H_j)|}{2|V(G)|},\]
where $\val(f_j)$ denotes the fraction of edges (in $E(G)$) satisfied by the partial assignment $f_j$.
\end{claim}

\begin{proof}
Let $H_j = S_j \cup (H_j \cap R_{j-1})$ and let $E_1 = E(S_j,S_j)$ and $E_2 = E(H_j \cap R_{j-1})$ so that $E_1 \cup E_2 = E(H_j,H_j)$ denoted by $E$. By Claim~\ref{claim:chernoff} we know that $\val_{f_j}(E_2) \leq 2/k$ where $k$ denotes $|\Sigma|$. 
By an averaging argument:
\[\delta \leq \frac{\#\text{sat}_{f_j}(E_1)}{|E|}+ \frac{|E_2|}{|E|}\val_{f_j}(E_2) \leq \frac{\#\text{sat}_{f_j}(E_1)}{|E|} + \frac{2}{k},\]
where $\#\text{sat}_{f_j}(E_1)$ denotes the number of edges of $E_1$ satisfied by $f_j$. This implies that $\frac{\#\text{sat}_{f_j}(E_1)}{|E|} \geq \delta/2$. 

Since $H_j$ is a $\leq r$-restricted subcube at least $(1-\alpha)^r$-fraction of edges stay inside $H_j$, hence $|E|/|E(H_j)|$ is at least $(1-\alpha)^r$. Using that $|E(H_j)|/|E(J)| \geq |V(H_j)|/|V(J)|$ we get that: 
\[\val(f_j) =  \frac{\#\text{sat}_{f_j}(E_1)}{|E|} \cdot \frac{|E|}{|E(H_j)|} \cdot \frac{|E(H_j)|}{|E(J)|} \geq \frac{\delta}{2}(1 - \alpha)^r\frac{|V(H_j)|}{|V(G)|},\]
where $\val(f_j)$ denotes the edges satisfied by $f_j$ (assigned to $S_j$) as a fraction in the whole graph. 
\end{proof}

Once we have these facts, the conclusion is immediate. Firstly there cannot be more than $V(G)$ iterations of the while-loop, since at each iteration we assign at least one new vertex. Each iteration takes time $T(\cA) + |V(G)|^2$, hence the algorithm runs in time $|V(G)|(T(\cA) + |V(G)|^2)$. 

Suppose the algorithm exits the while loop at the $t+1$ iteration. Then we know that $R_{t} \geq \frac{\gamma}{2}|V(J)|$.
Using claims~\ref{claim:it-val} above, we get that the value of the final assignment is proportional to the number of vertices assigned: 
\[\sum_{j=1}^t\val(f_j) \geq \delta(1 - \alpha)^r \sum_{j=1}^t \frac{|V(H_j)|}{2|V(J)|} \geq \delta(1 - \alpha)^r\frac{|R_t|}{2|V(J)|} \geq \delta(1 - \alpha)^r\frac{\gamma}{4},\]
thus proving the lemma.
\end{proof}
\section{Missing Proofs}\label{sec:missing_pfs}

The following claim bounds the Booleanity error term 
and shows that there is a \sos proof that they are
negligible (provided that we take the parameter $\nu$ 
to be sufficiently small).
\begin{claim}\label{claim:bool}
\[\cA_I ~~ \vdash_{\poly(\ell^r/\nu)} ~~ \E_{c}[\sum_s B(H^c_s)] \leq |\Sigma|\ell^{O(r)}\sqrt{\nu}.\]
\end{claim}

\begin{proof}
Recall that $B(F) = \frac{4}{3}\E_u[(F(u)^3 - F(u))(\Pi_{\geq \lambda_r}F)(u)] = \frac{4}{3}\ip{F^3 - F, PF}$, where $\Pi$ is the projection operator $\Pi_{\geq \lambda_r}$. Using Cauchy-Schwarz we get that for all $\eta > 0$:
\begin{equation}\label{eq:bool-cs}
B(F) \leq \frac{4}{3}\left(\frac{\eta}{2}||F^3 - F||_2^2 + \frac{1}{2\eta}||\Pi F||_2^2\right) \leq \frac{2\eta}{3}\E_u[F(u) -F(u)^3] + \frac{2}{3\eta}\delta(F^2),
\end{equation}
where the last inequality holds since $F \in [0,1]$ and $\Pi$ is a projection operator. We will use this 
to bound $B(H_s^c)$ as follows. 
Fixing $s, u$, taking $F = H_s^c(u)$, 
the first term above is 
$\E_c[(H_s^c(u)-H_s^c(u)^3]$ and we upper bound it 
by $O(\ell^r\nu)$. Indeed, first as 
$0\leq H_s^c(u)\leq 1$ we have
\begin{equation}\label{eq:bool1}
\E_c[H_s^c(u) - H_s^c(u)^3] = (1+H_s^c(u))\E_c[H_s^c(u) - H_s^c(u)^2] \leq 2 \E_c[H_s^c(u) - H_s^c(u)^2].    
\end{equation}
Next, recall that $H_s^c(u)$ is a product of the terms $(1-R^c_{s,a})$ for $a \subset u, |a| \leq r$. Additionally each variable $R^c_{s,a}$ is a product of approximate indicators: namely there are $\beta_i \in (0,1)$  such that $R_{s,a}^c(u) = p^{c,s,a}_1(\delta^{c,s,a}_1)\cdot \ldots\cdot p^{c,s,a}_k(\delta^{c,s,a}_k)$, where $p^{c,s,a}_i = p_{\beta_i+c\nu^2,\nu^2}$ or $p_{<\beta_i+c\nu^2,\nu^2}$ as $c$ varies in $\{0,1,\ldots,\lfloor 1/\nu \rfloor\}$, and $\delta_i$'s are linear functions of the variables $G_s(u)$ and lie in $[0,1]$ under the axioms $\cA_I$.
By Claim~\ref{claim:silly_sos}(applied twice) it follows that:
\[\E_c[H_s^c(u) - H_s^c(u)^2] \leq \sum_{a \subset u: |a|\leq r}\E_c[R^c_{s,a} - (R^c_{s,a})^2] \leq \sum_{a \subset u: |a|\leq r}\sum_{i \in [k]}\E_c[p_i^{c,s,a}(\delta^{c,s,a}_i) - p_i^{c,s,a}(\delta^{c,s,a}_i)^2],\]
so we can bound each of the terms above.  

We will show that each term is smaller than $O(\nu)$. Let $p$ be an approximate indicator polynomial $p = p_{\beta,\nu}$ and let $p^c$ denote $p_{\beta+c\nu^2,\nu^2}$ for $c \in \{0,\ldots,\lfloor 1/\nu\rfloor\}$. For all $\delta \in [0,1]$, we argue that
\begin{equation}\label{eq:bool-random}
\delta \in [0,1] \vdash_{\tO(1/\nu^2)} ~~ \E_c[p^c(\delta) - p^c(\delta)^2] \leq O(\nu).
\end{equation}
Indeed, this follows by case analysis on $\delta \in [0,1]$ (that we show next), and then using L\'{u}kacs theorem (Corollary~\ref{cor:Lukacs}) we get there is an \sos proof of degree $\deg(p^c)$. 

If $\delta \in [0,\beta] \cup [\beta+\nu, 1]$ then one can check that $p^c(\delta) \in [0,\nu^2]$ for all $c$. Therefore let $\delta \in [\beta,\beta+\nu]$, and specifically in one of the intervals - $[\beta+t\nu^2, \beta+(t+1)\nu^2]$ for $t \in \{0,\ldots,\lfloor1/\nu\rfloor\}$. Then we get that $p^t(\delta)$ could be any number between $[0,1]$ and therefore $p^t(\delta)-p^t(\delta)^2 \leq 1/4$. But for every $t' \neq t$ we get that $p^{t'}(\delta)$ is either in the interval $[0,\nu^2]$ or $[1-\nu^2,1]$ implying that $p^{t'}(\delta) - p^{t'}(\delta)^2 \leq O(\nu^2)$. Averaging over $a$ we therefore get:
\[\E[p^c(\delta) - p^c(\delta)^2] \leq \nu \cdot \frac{1}{4}+ (1-\nu)O(\nu^2) \leq O(\nu).\]
In conclusion, we get that 
$\E_c[H_s^c(u) - H_s^c(u)^2]\leq O(k\ell^r\nu) = \ell^{O(r)}\nu$. Plugging this into~\eqref{eq:bool-cs} 
and taking $\eta = 1/\sqrt{\nu}$ we get that
\[\E_c[B(H^c_s)] \leq  \sqrt{\nu}\ell^{O(r)} + O(\sqrt{\nu}\E_c[\delta(H^c_s)]) \leq \sqrt{\nu}\ell^{O(r)},\]
and summing over $s \in \Sigma$ finishes the proof.
\end{proof}

\begin{claim}\label{claim:silly_sos}
\[
\{0 \leq x_i \leq 1\} \vdash_{2k} 
\prod_{i \in [k]}x_i - (\prod_{i \in [k]}x_i)^2 \leq x_1-x_1^2 + \ldots + x_k - x_k^2. 
\]
\end{claim}
\begin{proof}
We prove this by induction. For $k=1$ this is clear, and for $k=2$ we have
\begin{align*}
xy - (xy)^2 &= xy – xy^2 + xy^2 – x^2 y^2 \\
&= x(y-y^2)+ (x-x^2) y^2 \\
&\leq y-y^2 + x-x^2.    
\end{align*}

For the inductive step, let $k\geq 3$ and denote $y=x_1\cdots x_{k-1}$. Applying the base case and 
then the inductive hypothesis we get
\[\prod_{i \in [k]}x_i - \left(\prod_{i \in [k]}x_i\right)^2 
=y x_k - (y x_k)^2
\leq y-y^2 + x_k-x_k^2
\leq x_1 - x_1^2 + \ldots + x_{k-1}^2 - x_{k-1}^2.
\]
One can check that the SoS degree used is at most $2k$.
\end{proof}

\section{Sum-of-Squares Certificate of Expansion in the Johnson Graph}\label{sec:fourier}
In this section we give an \sos proof of Theorem~\ref{thm:structure-johnson}.
Our proof follows the same lines as~\cite{KMMS}, except that we implement each step carefully by a low-degree SoS proof. For convenience, we restate 
Theorem~\ref{thm:structure-johnson} below.
\begin{theorem}[Expansion Theorem \ref{thm:structure-johnson} for Johnson Graphs  restated]\label{thm:structure-johnson-restated}
For all $\alpha \in (0,1)$, all integers $\ell \geq 1/\alpha$ and $n \geq \ell$, the following holds: Let $J_{n,\ell,\alpha}$ be the $\alpha$-noisy Johnson graph. For every constant $\gamma \in (0,1)$ and positive integer $r \leq O(\ell)$, every function $F: V(J) \rightarrow \R$ that is $(r,\gamma)$-pseudorandom has high expansion:
\begin{enumerate}
\item $\{F(X) \in [0,1]\}_{X \in V(J)}  \vdash_{O(1)}\, \\
\ip{F,LF} \geq \delta(1 - (1-\alpha)^{r+1})(1- \gamma^{1/3}\exp(r)) - \sum_{j = 0}^r \frac{ c_j\ell^{j}}{\gamma}\E_{a \sim {[n] \choose j}}[q_a(F)({\delta(F^{ 2}|_a) - \gamma)}]+ B(F)$
\item $\{F(X) \in [0,1]\} \,\,\, \vdash_{2} \,\,\, 0 \leq q_a(F) \leq 1,$
\end{enumerate}
where for all $j \leq r$, $c_j$'s are positive constants of size at most $\exp(r)$, for all 
$a \subseteq [n]$ of size $j$, $q_a(F)$ is a degree $2$  polynomials, and $B(F)=\frac{4}{3}\E_X[(F^{ 3} - F)\Pi_{\geq \lambda_{r}}F]$ for $\lambda_r = (1-\alpha)^r$. Here, $\Pi_{\geq \lambda}$ is the projection operator on the space of spanned by eigenvectors of$J$ of eigenvalues at least $\lambda$.
\end{theorem}

It will be more convenient for us to move to a closely related Cayley graph, which we denote by $C_{n,\ell,\alpha}$. This graph is essentially the same as the Johnson graph, albeit viewed as a product domain. In Theorem~\ref{thm:structure-johnson-cayley}, we state an analogous structure theorem for ``permutation-invariant'' sets on $C_{n,\ell,\alpha}$ that we show in Theorem~\ref{thm:structure-johnson-cayley}, and it is easy to derive Theorem~\ref{thm:structure-johnson-restated} above from it.

\paragraph{Notation.} 
We use $[n]$ to denote the set $\{0,\ldots,n-1\}$, and also the group $(\Z/n\Z)$, the natural numbers modulo $n$. 
Generally, when we take a set $S$ and raise it to a positive integer power $\ell$, we mean the set of all ordered multisets of elements of $S$ of size $\ell$. We use $\Chi_t$ for $t \in [n]$ to denote the characters of the group $\Z/n\Z$ (or the eigenvectors of the $n$-cycle), where $\Chi_t: [n] \rightarrow \mathbb{C}$ is the function $\Chi_t(x) = e^{\frac{2\pi i tx}{n}}$.
We will use $\lambda_G(v)$ to denote the eigenvalue of $v$ which is an eigenvector of the adjacency matrix of graph $G$. For a string $S \in \Sigma^m$, for some alphabet $\Sigma$, and a set $I \subseteq [m]$, we denote its restriction to the set of coordinates in $I$, by $S|_I$.

\subsection{The Graph $C_{n,\ell,\alpha}$ and its Spectrum}

\begin{definition}[Johnson-approximating Cayley Graph $C_{n,\ell,\alpha}$]
Let $\alpha$ be a number in $(0,1)$ and $\ell$ be a positive integer. Let $n$ be a positive integer such that $n > \ell$. 
\begin{enumerate}
\item The vertex set of $C_{n,\ell,\alpha}$ is $[n]^{\ell}$. We will drop the subscript $(n,\ell,\alpha)$ in $C_{n,\ell,\alpha}$ when these parameters are clear from context.
\item The edges are described by the following random process. For a vertex $X = (x_1,\ldots,x_\ell), x_i \in [n]$, the distribution over the neighbours of $X$ is described by: choose $(y_1,\ldots,y_\ell)$ uniformly at random from $[n]^{\ell}$ and $b = (b_1,\ldots,b_{\ell}) \sim \{0,1\}^{\ell}$ such that the Hamming weight of $b$ equals $\alpha \ell$, and output $Z = (x_1 + b_1 \cdot y_1,\ldots,x_l + b_\ell \cdot y_\ell)$ as the neighbour of $X$.
\end{enumerate}
\end{definition}

It is easy to verify that the graph defined above is a weighted Cayley graph with vertex set being the elements of the group $[n]^\ell = (\Z/n\Z)^{\ell}$. The natural group operation associated with this set is component-wise addition modulo $n$. 

We next discuss the spectral properties of $C_{n,\ell\alpha}$. Overloading notations, we let $C = C_{n,\ell,\alpha}$ also denote the normalized adjacency matrix of the graph $C$. Note firstly that the eigenvectors of $C$ are the characters of the group $[n]^\ell$, namely $\raisebox{2pt}{$\chi$}_{T}$ where $T = (T_1,\ldots,T_\ell) \in [n]^{\ell}$ defined as $\Chi_T(x) = \Chi_{T_1}(x_1) \cdot \ldots \cdot \Chi_{T_{\ell}}(x_\ell)$ for $x\in [n]^{\ell}$ 
We next define the degree of an eigenvector.

\begin{definition}[Degree of $\Chi_T$]
For all $T \in [n]^\ell$, where $T = (T_1,\ldots,T_\ell)$, define the degree of $\chi_{T}$ as:
$${\sf deg}(\chi_T) = |T| := |\{i \mid T_i \neq 0\}|.$$
\end{definition}

The following lemma asserts that the eigenvalue of $\chi_T$ is roughly $(1-\alpha)^{|T|}$:
\begin{lemma}\label{lem:eigenval}
Let $\lambda_C(\Chi_T)$ denote the eigenvalue of $C$ corresponding to the eigenvector $\Chi_T$ for $T \in [n]^\ell$. We have that, 
$$
\lambda_{C}(\raisebox{2pt}{$\chi$}_T) = 
\begin{cases}
\frac{\binom{\ell - |T|}{(1 - \alpha)\ell - |T|}}{\binom{\ell}{(1 - \alpha)\ell}}, ~~~~|T| \leq (1 - \alpha)\ell\\
0, ~~~~~~~~~~~~~~~~~~~~~~~\text{otherwise.} 
\end{cases}
$$
\end{lemma}

\begin{proof}
Let $T = (T_1, \ldots, T_\ell)$. For all $X \in [n]^\ell$, we have that,

\begin{align*}
C \cdot \Chi_T(X) = \E_{y,b}[\Chi_T(x_1+b_1 y, \ldots, x_l + b_l y_l)] 
&= \Chi_T(X) \E_{y,b}[\Chi_{T_1,\ldots,T_l}(b_1y_1,\ldots,b_ly_l)] \\
&= \Chi_T(X) \E_{y,b}[\Chi_{(b_1 T_1,\ldots,b_l T_l)}(y)].
\end{align*}

For $y,S \in [n]$ and $S \neq 0$, we know that the eigenvector $\Chi_S$ is orthogonal to the eigenvector $\Chi_0$, equivalently that $\E_y[\Chi_S(y)] = 0$, whereas if $S = 0$ then $\E_y[\Chi_S(y)] = 1$. So we get that,

\begin{align*}
\lambda_{C}(\Chi_T) = \E_{y,b}[\Chi_{(b_1 T_1,\ldots,b_\ell T_\ell)}(y)] 
&= \Pr_b[(b_1T_1,\ldots,b_\ell T_\ell) = 0^{\ell}] \\
&= \begin{cases}
\frac{\binom{l - |T|}{(1 - \alpha)l - |T|}}{\binom{l}{(1 - \alpha)l}}, ~~~~|T| \leq (1 - \alpha)\ell\\
0, ~~~~~~~~~~~~~~~~~~~~~~~\text{otherwise.} 
\end{cases}.
\end{align*}
\end{proof}

\subsection{Analyzing non-expanding sets of the Johnson graph}
Since we want to deal with sets in the Johnson graph $J_{n,\ell,\alpha}$ we will only consider ``permutation-invariant'' sets on $C_{n,\ell,
\alpha}$. Notice that the vertices of the Johnson graph are subsets of $[n]$ of size $\ell$, whereas the vertices of the Johnson-approximating graph are ordered $\ell$-tuples of $[n]$. Therefore, given a set $S$ in the Johnson graph, it has a natural mapping to the set $S'$ which is a subset of the vertices of the Johnson-approximating graph $C$, 
$S' := \{(x_{\pi(1)},\ldots,x_{\pi(\ell)}) \mid \pi: [\ell] \rightarrow [\ell], \{x_1,\ldots,x_{\ell}\} \in S\}$. This leads to the following definition:

\begin{definition}[Permutation-invariance]\label{def:perm-inv}
We say that a set $S \subseteq C_{n,\ell,\alpha}$ is permutation-invariant if for all permutations $\pi \in \mathcal{S}_\ell$ and all $X = (x_1,\ldots,x_\ell) \in S$ we have that $X_{\pi} = (x_{\pi(1)},\ldots,x_{\pi(\ell)})$ belongs to $S$. Similarly, a function $F: V(C) \rightarrow \R$ is permutation invariant if for all inputs $X = (x_1,\ldots,x_{\ell})$, we have that $F(x_1,\ldots,x_\ell) = F(x_{\pi(1)},\ldots,x_{\pi(\ell)})$, for all $\pi\in S_{\ell}$. Further let $\A_{inv}$ denote the set of axioms that $F$ is permutation-invariant, that is, 
\[\A_{inv} := \{F(x_1,\ldots,x_\ell) = F(x_{\pi(1)},\ldots,x_{\pi(\ell)})\}_{\pi \in \mathcal{S}_\ell, X \in [n]^\ell}.\]
\end{definition}

Since the set of vertices in $C$ that correspond to some set of vertices in $J$ are permutation invariant it will be enough to focus are attention on these special sets and from now on whenever we refer to a set in $V(C)$, the reader can assume that it is permutation-invariant.

To analyze non-expanding sets of $C$, we will consider permutation-invariant functions $F: V(C) \rightarrow [0,1]$. Typically one would consider $0/1$-valued functions $F$, where $F$ is the indicator function of a set $S$, i.e. $F(X) = 1$ when $X \in S$. But since we need to analyze ``approximate-sets'' (the indicator function is approximated by a polynomial that takes values close to $0/1$), $F(X)$ could take any value between $[0,1]$.

Recall that the Fourier decomposition of $F$ gives us that, $F(X) = \sum_T \hat{F}(T)\Chi_T(X)$. We will now define the following for a function $F$:

\begin{definition}\label{def:f_i}
Given a function $F\colon V(C)\to [0,1]$, we have the following level decomposition:
\begin{enumerate}
\item We write 
$F = F_0 + \ldots + F_\ell$,
where $F_i(X) = \sum_{T: |T| = i} \hat{F}(T) \Chi_T(X)$. We will call $F$ a level $i$ function, if its Fourier decomposition has degree $i$ characters only, i.e. $\hat{F}(T) = 0,$ for all $T$ such that $|T| \neq i$.
\item Let $f_{i,F}: [n]^i \rightarrow \R$ be a function defined as, 
$$f_{i,F}(x_1,\ldots,x_i) := \sum_{T_1,\ldots,T_i \in ([n] \setminus 0)^i}\hat{F}(T_1,\ldots,T_i,0,\ldots)\chi_{T_1,\ldots,T_i}(x_1,\ldots,x_i),$$ 
\end{enumerate}
\end{definition}

Let $X = (x_1,\ldots,x_{j}) \in [n]^{j}$ and $I$ be a subset of $\{1,\ldots,j\}$. Let $I = \{k_1,\ldots,k_{|I|}\}$ where $k_1 < k_2 < \ldots < k_{|I|}$. We will use $X|_I$ to denote the ordered tuple of elements $(x_{k_1},\ldots,x_{k_{|I|}})$. We will now state some simple properties of $F$ that are implied by permutation-invariance.

\begin{lemma}\label{lem:props}
For all functions $F:[n]^\ell \rightarrow \R$ that are permutation-invariant, we have that:
\begin{enumerate}
\item $\hat{F}(T_1,\ldots, T_\ell) = \hat{F}(T_{\pi(1)},\ldots,T_{\pi(\ell)})$, for all $(T_1,\ldots,T_\ell) \in [n]^\ell$ and all permutations $\pi:[\ell] \rightarrow [\ell]$.

\item The functions $F_i$ and $f_{i,F}$ are also permutation-invariant.

\item $F_i(X) = \sum\limits_{\substack{I \subseteq [\ell] \\ |I| = i}} f_{i,F}(X|_{I})$.
\end{enumerate}
\end{lemma}
\begin{proof}
A straightforward manipulation of the definitions.
\end{proof}

We note that the Fourier coefficients $\hat{F}(T), \forall T \in \C$, the level-$i$ functions $F_i(X)$ and $f_i(X)$ are all linear functions of the indeterminates $\{F(X)\}$,
hence we shall also think of them as indeterminates 
when arguing about \sos proofs.

\subsection{Restrictions}
In this section, we define restrictions of functions $F$ and state several lemmas related to them.
\begin{definition}[$r$-restricted subcubes of $C$]
Given an ordered tuple, $A = (a_1,\ldots,a_r)$ for $a_i \in [n]$ and $r \leq \ell-1$, we let $C|_{A}$ denote the subset of vertices of $C$ whose first $r$ coordinates are restricted to be $(a_1,\ldots,a_r)$. We call such a subset an $r$-restricted subcube of $C$.
\end{definition}

\begin{definition}[Restrictions]
Given a function $F:[n]^\ell \rightarrow \R$ and an ordered tuple, $A = (a_1,\ldots,a_r)$ for $a_i \in [n]$ and $1 \leq r \leq \ell-1$, we define the restricted function $F|_A: [n]^{\ell - r} \rightarrow \R$ as, 
$$F|_A(x_1,\ldots,x_{\ell-r}) = F(a_1,\ldots,a_r, x_1,\ldots,x_{\ell-r}).$$
Further, let $\delta_A(F)$ denote the mass of the function restricted to $A$, that is,
$$\delta_A(F) := \delta(F|_A) = \E\limits_{X \in [n]^{\ell - r}}[F|_A(X)].$$
For convenience, when $A = \phi$ ($r = 0$), define $F|_A(X) := F(X)$, and $\delta_A(F) := \delta(F) = \E_{X \in [n]^\ell}[F(X)]$.
\end{definition}

The following lemma gives a relation between the level functions of a function $F$ and the level function of its restrictions.
\begin{lemma}\label{lem:restriction}
Let $F$ be a permutation-invariant function on $V(C)$. Then we have the following:
\begin{enumerate}
\item For all $a \in [n]$ and for all $i$ such that $0 \leq i \leq \ell-1$, and all $X \in [n]^i$, we have that,
\[f_{i+1,F}(a, X) = f_{i,F|_{\{a\}}}(X) - f_{i,F}(X).\]

\item For all integers $i$ such that $0 \leq i \leq \ell$ and for all $X \in [n]^i$, we get an inclusion-exclusion formula for $f_i(X)$ in terms of restrictions of $F$: 
\[f_{i,F}(X) = \sum_{B \subseteq \{1,\ldots,i\}} (-1)^{i - |B|} \delta_{X|_B}(F),\]
where $X|_B$ is the ordered tuple of elements of $X$ restricted to the indices in $B$.
\end{enumerate}
\end{lemma}

\begin{proof}
We prove each item separately.
\paragraph{Proof of the first item.}
Using the definition, we can expand out $f_{i+1,F}$ to get that,

\[f_{i,F}(a, X) = \sum_{(T_1,\ldots,T_{i+1}) \in ([n] \setminus 0)^{i+1}}\hat{F}(T_1,\ldots,T_{i+1},0,\ldots, 0)\Chi_{T_1,\ldots,T_{i+1}}(a, X).\]

We can split this sum into two parts, one where $T_{1}$ can take \emph{any} value (even $0$) and the 
second where $T_{1} = 0$. We get that,

\begin{align*}
f_{i,F}(a,X) &= \sum_{ \substack{T_1 \in [n] \\ T \in ([n] \setminus 0)^i}}\hat{F}(T_1,T, 0,\ldots,0,\ldots, 0)\Chi_{T_1,T}(a,X) 
- \sum_{ \substack{T_1 = 0 \\ T \in ([n] \setminus 0)^i}}\hat{F}(0, T ,0,\ldots, 0)\Chi_{T}(X).
\end{align*}
We will show that the first term equals $f_{i, F|_{\{a\}}}(X)$ and the second term equals $f_{i,F}(X)$. This implies the conclusion needed.

For the first term we have that, 
\begin{align}
&\sum_{ \substack{T_1 \in [n] \\ T \in ([n] \setminus 0)^i}}\hat{F}(T_1, T,0,\ldots,0)\Chi_{T_1,T}(a, X) \nonumber \\
= &\sum_{ \substack{T_1 \in [n] \\ T \in ([n] \setminus 0)^i}} \E_{\substack{Y_1 \in [n], \\ Y \in [n]^{\ell - 1}}}\left[F(Y_1,Y)\overline{\Chi_{T_1}}(Y_1)\overline{\Chi_{(T,0,\ldots,0)}(Y)}\right]\Chi_{T_1}(a)\Chi_{T}(X) \nonumber \\
= &\sum_{T \in ([n] \setminus 0)^i} \E_{\substack{Y_1 \in [n], \\ Y \in [n]^{\ell - 1}}}\left[F(Y_1,Y)\overline{\Chi_{(T,0,\ldots,0)}(Y)}\sum_{T_1 \in [n]}\Chi_{T_1}(a-Y_1)\right]\Chi_{T}(X). \label{eq:last}
\end{align}

We now have that $\sum_{T_1 \in [n]}\Chi_{T_1}(a-Y_1) = 0$ if $Y_1 \neq a$ and equals $n$ otherwise. Using this fact we get that~\eqref{eq:last} is equal to
\begin{align*}
\sum_{T \in ([n] \setminus 0)^i} \frac{1}{n} \cdot \E_{Y \in [n]^{\ell - 1}}\left[F(a,Y)\overline{\Chi_{(T,0,\ldots,0)}(Y)}\cdot n\right]\Chi_{T}(X) 
= \sum_{T \in ([n] \setminus 0)^i} \widehat{F|_{\{a\}}}(T,0,\ldots,0)\Chi_{T}(X),
\end{align*}
which is equal to $f_{i, F|_{\{a\}}}(X)$ by definition.

As for the second term, it is equal to
\begin{align*}
&\sum_{ \substack{T_1 = 0 \\ (T_2,\ldots,T_{i+1}) \in ([n] \setminus 0)^i}}\widehat{F}(0, T_2,\ldots,T_{i+1},0,\ldots, 0)\Chi_{T_2,\ldots,T_{i+1}}(X)\\
= &\sum_{(T_2,\ldots,T_{i+1}) \in ([n] \setminus 0)^i}\hat{F}(T_2, \ldots,T_{i+1},0,\ldots, 0)\Chi_{T_2,\ldots,T_{i+1}}(X),
\end{align*}
since by Lemma~\ref{lem:props} (1) we have that $\hat{F}(0, T_2,\ldots,T_{i+1}, 0, \ldots) = \hat{F}(0, T_2,\ldots,T_{i+1}, 0, \ldots)$. Since the last equality is the definition of $f_{i,F}(X)$, the conclusion follows.

\paragraph{Proof of the second item.}
We will prove this claim by induction on $i$. For the base case of $i = 0$, by definition, we have that,

\[f_{0,F}(\phi) = \hat{F}(0,\ldots,0) = \E_{X \in [n]^\ell}[F(X)] = \delta(F) = \delta_{\phi}(F) = \sum_{B \subseteq \phi} \delta_{\phi|_B}(F).\]

Now let us assume that for all permutation-invariant functions $G$ the claim holds for $i-1$, i.e. for all $X \in [n^{i-1}]$, we have that $f_{i-1,G}(X) = \sum_{B \in \{1,\ldots,i-1\}} (-1)^{i-1-|B|} \delta_{X|_{B}}(G)$. Now we will prove the  claim for $f_{i,F}$, thus completing the induction.
Let $X = (x_1, X')$, where $X \in [n]^i, x_1 \in [n]$ and $X' \in [n]^{i-1}$. Then by property (1) of the same lemma, we have that,
\[f_{i,F}(X) = f_{i,F|_{\{x_1\}}}(X') - f_{i,F}(X').\]

Expanding the RHS using the induction hypothesis on the functions $F|_{\{x_1\}}$ and $F$, we get that,

\begin{align*}
f_{i,F}(X) &= \sum_{B' \in \{1,\ldots,i-1\}} (-1)^{i-1-|B'|} \delta_{X'|_{B'}}(F|_{\{x_1\}}) - \sum_{B' \in \{1,\ldots,i-1\}} (-1)^{i-1-|B'|} \delta_{X'|_{B'}}(F) \\
&= \sum_{B' \in \{1,\ldots,i-1\}} (-1)^{i-(1+|B'|)} \delta_{(x_1,X'|_{B'})}(F) + \sum_{B' \in \{1,\ldots,i-1\}} (-1)^{i-|B'|} \delta_{X'|_{B'}}(F) \\
&= \sum_{B \in \{1,\ldots,i\}: 1 \in B} (-1)^{i-|B|} \delta_{X|_B}(F) + \sum_{B \in \{1,\ldots,i\}: 1 \notin B} (-1)^{i-|B|} \delta_{X|_{B}}(F) \\
&= \sum_{B \in \{1,\ldots, i\}} (-1)^{i-|B|} \delta_{X|_B}(F).
\end{align*}

This completes the inductive step and the proof of the lemma.
\end{proof}

\subsection{Calculating $2$-nd Moments}
Let $W_i(F)$ denote the weight of $F$ on the $i^{th}$ level, that is, let $\eta_i = W_i[F] = \E_X[F_i(X)^2]$. Let $\delta(F)$ denote $\E_X[F(X)]$. Now we will derive some lemmas about the second moments. We get the following relation between the level-$i$ Fourier weight $\eta_i$ and the $\ell_2$-norm of the $f_i$'s.

\begin{lemma}\label{lem:3.2}
For all $i \in [\ell]$, given the variables $\{F(X)\}_{X \in V(C)}$:
$$\cA_{inv} ~~\vdash_2 ~~\E_{X \in [n]^{i}}[f_{i,F}(X)^2] =  \frac{\eta_{i}}{\binom{\ell}{i}}.$$
\end{lemma}

\begin{proof}
Since $F_i(A) = \sum_{I \subseteq [\ell]}f_i(A|_I)$, we have that, 
\begin{align*}
    \eta_i=\E_{A \sim [n]^{\ell}}[F_i(A)^2] = 
    \E_A[(\sum_I f_i(A_I))^2] 
    &= \sum_I \E_A[ f_i(A_I)^2] + \sum_{I \neq I'} \E_{A}[f_i(A_I)f_i(A_{I'})] \\
    &= \sum_I \E_A[ f_i(A_I)^2] + 0 \\
    &= \E_{(x_1,\ldots,x_i)}[f_i(X)^2] \cdot \binom{\ell}{i},
\end{align*}
and rearranging implies the lemma. 
\end{proof}

\begin{lemma}\label{lem:p.r.-weight}
For all $i \in [\ell]$, for all $a,b \in [i]$, for all $A \in (\{0,1\}^k)^a$, given the indeterminates $\{F(X)\}_{X \in V(C)}$:
\begin{align*}
\cA_{inv}(F) ~~~\vdash_2 ~~~ &\E_{X \in (\{0,1\}^k)^{i-a}}[f_i(A,X)^2] \leq \frac{\delta(F^{ 2}|_A)}{{l-a \choose i-a}}.
\end{align*}
\end{lemma}

\begin{proof}
By Lemma~\ref{lem:3.2}, we have that,
\begin{align}
\E_{X \in (\{0,1\}^k)^{i-a}}[f_i(A,X)^2] = \E_{X \sim (\{0,1\}^k)^{i-a}}[f_{i-a,F|_A}(B)^2] = \frac{W^{i-a}(F|_A)}{{l-a \choose i-a}}. \label{5.5:eq1}
\end{align}

We also have that,

\begin{align}
W^{i-a}(F|_A) \leq \E_{X \in (\{0,1\}^k)^{i-a}}[F|_A(X)^2] = \delta(F^{ 2}|_A), \label{5.5:eq2}
\end{align}
where the first inequality is a degree 2 SoS-inequality. The proof is concluded by combining~\eqref{5.5:eq1},~\eqref{5.5:eq2}.
\end{proof}

\subsection{Bounding $4$-th moments}
The next part of the proof is to prove a lower 
bound on the fourth moment of $F_i$ (using the fact
it is correlated with a Boolean valued function, 
hence is supposed to have a high $4$th moment) 
as well as an upper bound on the fourth moment of 
$F_i$ (using the fact $F$ is pseudo-random). We 
begin with the following lemma which gives us the lower bound:
\begin{lemma}[Lower Bound]\label{lem:lower-bound}
For all real $\epsilon > 0$ given the indeterminates $\{F(X)\}_{X \in V(C)}$ we have that,
$$\E_X[F_i^4(X)] \geq 4\epsilon^3 \eta_i - 3\epsilon^4\delta(F) + B(F),$$
where $B(F) = 4\eps^3\E_X[(F^{ 3} - F)F_i]+3\eps^4 \E_X[F - F^{ 4}]$.
\end{lemma}

\begin{proof}
We have that, 
\begin{align}\label{eq:eta-i}
\eta_i = \ip{F,F_i} = \E_X[F(X)^3 \cdot F_i(X)] + \E_X[(F(X) - F(X)^3) \cdot F_i(X)],
\end{align}
since $F = \sum F_i$ is an orthogonal decomposition of $F$. Using Fact~\ref{fact:sos-hol} on the first term we get:
\begin{align*}
\E_X[F(X)^3 \cdot F_i(X)] &= 
\E_X[F(X)^3 \cdot F_i(X)]\\
&\leq  
\frac{3\epsilon}{4}\E_X[F(X)^4] + \frac{1}{4\epsilon^3}\E_X[F_i(X)^4],\\
&= \frac{3\epsilon}{4}\E_X[F(X)] +\frac{3\epsilon}{4}\E_X[F(X)^4 - F(X)] +  \frac{1}{4\epsilon^3}\E_X[F_i(X)^4],
\end{align*}
for all real $\epsilon > 0$. Plugging this into~\eqref{eq:eta-i} and rearranging gives the lemma.
\end{proof}

We will now prove an upper bound on $\E_X[F_i^4(X)]$. We have the following lemma which bounds the fourth moment of $F_i$ in terms of the second moment of $F_i$:


\begin{lemma}[Upper Bound]\label{lem:upper-bound}
For all integers $i \in [\ell]$, there exist positive constants $c_1, \ldots, c_i \leq \exp(i)$ such that for all $F$:
\[\cA_{inv} \,\, \vdash_{O(1)}\,\, \E_X[F_i(X)^4] \leq \exp(i)\gamma \eta_i + \sum_{j = 0}^i c_j \ell^{i}\E_{a \sim [n]^{j}}[\E_{b \sim [n]^{\ell-j}}[f_i(a,b)^2]\cdot (\delta(F^{ 2}|_a) - \gamma)].\] 
\end{lemma}

\begin{proof}
When the set $X$ and the function $F$ is clear from context, we use $f_i(I)$ in place of $f_i(X|_I)$. Thus,
\begin{align*}
\E_X[F_i^4(X)] &= \sum_{\substack{ I_1,\ldots,I_4 \subseteq [l] \\ |I_j| = i, \forall j\in [4]}}\E_X[ f_i(I_1) f_i(I_2) f_i(I_3) f_i(I_4)].   
\end{align*}
Denote, $$D(i,d) = \{(I_1,\ldots,I_4) \mid \forall j \in [4], I_j \subseteq [l], |I_j| = i, |I_1 \cup I_2 \cup I_3 \cup I_4| = d\}.$$
Note that for all $d \in [i,4i]$,
\[|D(i,d)| \leq {\ell \choose d}\cdot {d \choose i}^4 \leq \left(\frac{\ell e}{d}\right)^d \left(\frac{de}{i} \right)^{4i} \leq \ell^d \exp(i).\]

We will prove the following lemma about the four-wise products, for $(I_1,\ldots,I_4) \in D(i,d)$. We defer
the proof to Section~\ref{sec:pf_lem-prod}.
\begin{lemma}\label{lem-prod}
Let $i,d$ be integers such that $i \leq d \leq 4i$ and let $I_1,\ldots,I_4 \subseteq [l]$, such that $(I_1,I_2,I_3,I_4) \in D(i,d)$. Then there exist integers $t_1,t_2 \in [i]$ such that for all $\gamma$:
\begin{align*}
&\E_X[f_i(I_1)f_i(I_2)f_i(I_3)f_i(I_4)] \\
&\leq \left(\frac{O(i)}{\ell}\right)^{d}\gamma\eta_i + \left(\frac{O(i)}{\ell}\right)^{d-i}\E_{a_1 \sim [n]^{t_1}}[\E_{b_1 \sim [n]^{i-t_1}}[f_i(a_1,b_1)^2]\cdot (\delta(F^{ 2}|_{a_1}) - \gamma)]  \\ 
&+ \left(\frac{O(i)}{\ell}\right)^{d-i}\E_{a_2 \sim [n]^{t_2}}[\E_{b_2 \sim [n]^{i-t_2}}[f_i(a_2,b_2)^2](\delta(F^{ 2}|_{a_2}) - \gamma)]. 
\end{align*}
\end{lemma}
We now use Lemma~\ref{lem-prod} to finish the proof of Lemma~\ref{lem:upper-bound}. Indeed, we get that
\begin{align*}
&\E_X[F_i^4(X)] \\
&= \sum_{\substack{ I_1,\ldots,I_4 \subseteq [l] \\ |I_j| = i, \forall j\in [4]}}\E_X[ f_i(I_1) f_i(I_2) f_i(I_3) f_i(I_4)] \\
&\leq \sum_{d = i}^{4i} \sum_{\substack{ I_1,\ldots,I_4 \subseteq [l] \\ (I_1,\ldots,I_4) \in D(i,d)}} \E_X[ f_i(I_1) f_i(I_2) f_i(I_3) f_i(I_4)] \\
&\leq \sum_{d = i}^{4i} \sum_{\substack{ I_1,\ldots,I_4 \subseteq [l] \\ (I_1,\ldots,I_4) \in D(i,d)}} \left(\frac{O(i)}{\ell}\right)^{d}\gamma\eta_i +  \left(\frac{O(i)}{\ell}\right)^{d-i}\E_{a_1 \sim [n]^{t_1}}[\E_{b_1 \sim [n]^{i-t_1}}[f_i(a_1,b_1)^2]\cdot (\delta(F^{ 2}|_{a_1}) - \gamma)]  \\ 
& ~~~~~~~~~~~~~~~~~~~~~~~~~~~~~~~~~+\left(\frac{O(i)}{\ell}\right)^{d-i}\E_{a_2 \sim [n]^{t_2}}[\E_{b_2 \sim [n]^{i-t_2}}[f_i(a_2,b_2)^2](\delta(F^{ 2}|_{a_2}) - \gamma)] \\
&\leq \exp(i)\gamma \eta_i + \sum_{j = 0}^i c_j \ell^{i}\E_{a \sim [n]^{j}}[\E_{b \sim [n]^{i - j}}[f_i(a,b)^2]\cdot (\delta(F^{ 2}|_a) - \gamma)],  
\end{align*}
for some positive constants $c_j \leq \exp(i)$, where we have used that $|D(i,d)| \leq \ell^d \exp(i)$.
\end{proof}

\subsection{Proof of Lemma~\ref{lem-prod}}\label{sec:pf_lem-prod}
We begin by analyzing the simple case in which there is an index that appears in only one of the $I_j$'s, 
and show that in this case the expectation is just equal to $0$:
\begin{proposition}
If there is an index $j$ that appears in only one of the $I_k$s, then,
$$\E_X[ f_i(I_1) f_i(I_2) f_i(I_3) f_i(I_4)] = 0.$$
\end{proposition}

\begin{proof}
Without loss of generality we can assume that $1$ appears only in $I_4$. We have that,
\begin{align*}
\E_X[ f_i(I_1) f_i(I_2) f_i(I_3) f_i(I_4)] &= \E_{x_1,\ldots,x_l}[ f_i(I_1) f_i(I_2) f_i(I_3) f_i(I_4)] \\
&= \E_{x_2,\ldots,x_l}[f_i(I_1) f_i(I_2) f_i(I_3) \E_{x_1}[f_i(I_4) \mid x_2 = x_2',\ldots, x_l = x_l']]\\
&= \E_{x_2,\ldots,x_l}[f_i(I_1) f_i(I_2) f_i(I_3) \delta_{(x_2',\ldots,x_l') \cap I_4}(f_i)] \\
&= 0,
\end{align*}
where the last equality follows by Lemma~\ref{lem:props}. 
\end{proof}
For the rest of the proof we assume each index in $I_1\cup I_2\cup I_3\cup I_4$ appears in at least two sets. Without loss of generality we can assume that $\cup_{j = 1}^4 I_j = [d]$. Let $H_2, H_3, H_4 \subseteq [d]$ be the set of elements that appear in $2,3,4$ of the sets $(I_1,\ldots,I_4)$ respectively. Furthermore let $H_2 = X_{12} \cup X_{13} \ldots X_{34},$ where $X_{jk}$ is the set of elements that occur in $I_j$ and in $I_k$. We will abuse notation and use $H_j$ to also denote the set of random variables $(x_{j_1},\ldots, x_{j_m}), m \in H_j$ and $x_m \sim \{0,1\}^k$. Firstly we get that, 

\begin{align}
&\E_{H_2}[ f_i(I_1) f_i(I_2) f_i(I_3) f_i(I_4)] \nonumber\\
= &\E_{X_{12},\ldots,X_{34}}[f_i(I_1) f_i(I_2) f_i(I_3) f_i(I_4)] \nonumber\\
= &\E_{\substack{X_{13},X_{14},\\ X_{23},X_{24}}}[\E_{X_{12}}[f_i(I_1) f_i(I_2)]\E_{X_{34}}[f_i(I_3) f_i(I_4)] ] \nonumber\\
\leq &\frac{\eps}{2}\E_{\substack{X_{13},X_{14},\\ X_{23},X_{24}}}[\E_{X_{12}}[f_i(I_1) f_i(I_2)]^2] +  \frac{1}{2\eps}\E_{\substack{X_{13},X_{14},\\ X_{23},X_{24}}}[\E_{X_{34}}[f_i(I_3) f_i(I_4)]^2] \label{CS1} \\
\leq &\frac{\eps}{2}\E_{\substack{X_{13},X_{14},\\ X_{23},X_{24}}}[\E_{X_{12}}[f_i(I_1)^2]\E_{X_{12}}[f_i(I_2)^2]] +  \frac{1}{2\eps}\E_{\substack{X_{13},X_{14},\\ X_{23},X_{24}}}[\E_{X_{34}}[f_i(I_3)^2]\E_{X_{34}}[f_i(I_4)^2]]\label{CS2} \\
\leq &\frac{\eps}{2}\E_{X_{12},X_{13},X_{14}}[f_i(I_1)^2]\E_{X_{12},X_{23},X_{24}}[f_i(I_2)^2] +  \frac{1}{2\eps}\E_{X_{34},X_{13},X_{23}}[f_i(I_3)^2]\E_{X_{34},X_{14},X_{24}}[f_i(I_4)^2] \nonumber \\
= &\frac{\eps}{2}\E_{H_2}[f_i(I_1)^2]\E_{H_2}[f_i(I_2)^2] +  \frac{1}{2\eps}\E_{H_2}[f_i(I_3)^2]\E_{H_2}[f_i(I_4)^2],\label{last-eq} 
\end{align}
where inequalities~\eqref{CS1},~\eqref{CS2} follow by Cauchy Schwarz (Lemma~\ref{CS1-prelim} and Lemma~\ref{CS2-prelim}). 


\bigskip

We will now bound the expression $\E_X[f_i(I_1)f_i(I_2)f_i(I_3)f_i(I_4)]$. We have that,

\begin{align}
&\E_X[f_i(I_1)f_i(I_2)f_i(I_3)f_i(I_4)] \nonumber \\
= &\E_{H_4,H_3} [\E_{H_2}[f_i(I_1)f_i(I_2)f_i(I_3)f_i(I_4)]]\nonumber \\
\leq &\frac{\eps}{2}\E_{H_4,H_3}[\E_{H_2}[f_i(I_1)^2]\E_{H_2}[f_i(I_2)^2]] +  \frac{1}{2\eps}\E_{H_4,H_3}[\E_{H_2}[f_i(I_3)^2]\E_{H_2}[f_i(I_4)^2]]\nonumber\\
= &\frac{\eps}{2}\E_{A}[\E_{B_1}[f_i(I_1)^2]\E_{B_2}[f_i(I_2)^2]] +  \frac{1}{2\eps}\E_{C}[\E_{D_1}[f_i(I_3)^2]\E_{D_2}[f_i(I_4)^2]],\nonumber\\
= &\frac{\eps}{2}\E_{A}[\E_{B_1}[f_i(a, b_1)^2]\E_{B_2}[f_i(a, b_2)^2]] +  \frac{1}{2\eps}\E_{C}[\E_{D_1}[f_i(c,d_1)^2]\E_{D_2}[f_i(c,d_2)^2]],\label{eq:hypcont}
\end{align}
with $A = H_4 \cup (H_3 \cap I_1 \cap I_2)$, 
$B_1 = (H_2 \cap I_1) \cup ((H_3 \cap I_1) \setminus I_2)$,
$B_2 = (H_2 \cap I_2) \cup ((H_3 \cap I_2) \setminus I_1)$, $C = H_4 \cup (H_3 \cap I_3 \cap I_4)$, $D_1 = (H_2 \cap I_3) \cup ((H_3 \cap I_3) \setminus I_4)$,
$D_2 = (H_2 \cap I_4) \cup ((H_3 \cap I_4) \setminus I_3)$. The sets $a,b_1,b_2,c,d_1,d_2$ are instantiations of the corresponding random variables, $A, \ldots,D_2$. We have that,

\begin{align*}
\E_{A}[\E_{B_1}[f_i(a, b_1)^2]\E_{B_2}[f_i(a, b_2)^2]] &\leq \E_{A}[\frac{\delta(F^{ 2}|_a)}{{\ell-|A| \choose i - |A|}}\E_{B_2}[f_i(a,b_2)^2]] \\
&= \frac{1}{{\ell-|A| \choose i - |A|}}\E_{A}[\gamma \E_{B_2}[f_i(a,b_2)^2]] + \E_A[(\delta(F^{ 2}|_A) - \gamma)\left(\frac{\E_{B_2}[f_i(a,b_2)^2]}{{\ell-|A| \choose i - |A|}}\right)],\\
&\leq \frac{\gamma}{{\ell-|A| \choose i - |A|}}\cdot \frac{\eta_i}{{\ell \choose i}} + \E_A[(\delta(F^{2}|_A) - \gamma)\left(\frac{\E_{B_2}[f_i(a,b_2)^2]}{{\ell-|A| \choose i - |A|}}\right)],
\end{align*}
where in the first inequality we have used Lemma~\ref{lem:p.r.-weight} to bound $\E_{B_1}[f_i(a,b_1)^2]$ for each $a \sim A$ and the last inequality uses Lemma~\ref{lem:3.2}. Similarly we can bound the second term in equation~\eqref{eq:hypcont}:
\[\E_{C}[\E_{D_1}[f_i(c,d_1)^2]\E_{D_2}[f_i(c,d_2)^2]] \leq \frac{\gamma}{{\ell-|C| \choose i - |C|}}\cdot \frac{\eta_i}{{\ell \choose i}} + \E_C[(\delta(F^{2}|_C) - \gamma)\left(\frac{\E_{D_2}[f_i(c,d_2)^2]}{{\ell-|C| \choose i - |C|}}\right)].\]

Plugging in these bounds into equation~\eqref{eq:hypcont} and setting $\eps = \left(\frac{O(i)}{\ell}\right)^{\frac{|A|-|C|}{2}}$ we get:
\begin{align*}
&\E_X[f_i(I_1)f_i(I_2)f_i(I_3)f_i(I_4)] \\
&\leq \left(\frac{O(i)}{\ell}\right)^{\frac{4i - |A|-|C|}{2}}\gamma\eta_i + \E_A[ \left(\frac{O(i)}{\ell}\right)^{\frac{2i-|A|-|C|}{2}}\E_{B_2}[f_i(a,b_2)^2]\cdot (\delta(F^{ 2}|_a) - \gamma)]  \\ 
&+\E_C[\left(\frac{O(i)}{\ell}\right)^{\frac{2i-|A|-|C|}{2}}\E_{D_2}[f_i(c,d_2)^2](\delta(F^{ 2}|_c) - \gamma)]. 
\end{align*}

We know that $4i - |A|-|C| = 2d$ and substituting this gives the lemma statement.

\subsection{Combining the Upper and Lower Bounds on $\|F_i\|_4$: the Level $i$ Inequality}
Combining the upper and lower bounds on $\E[F_i(X)^4]$ we can now prove an upper bound on $\eta_i$.

\begin{lemma}[Level $i$ inequality]\label{lem:leveli-j}
There exist positive constants $c_0,\ldots,c_i \leq \exp(i)$ such that given the indeterminates $\{F(X)\}_{X \in V(C)}$,
\begin{align*}
\cA_{inv} \,\, \vdash_{O(1)}\,\, 
\eta_i \leq \exp(i)\gamma^{1/3}\delta + B(F) + \sum_{j = 0}^i  \frac{c_j\ell^{i}}{\gamma}\E_{a \sim [n]^{j}}[\E_{b \sim [n]^{i-j}}[f_i(a,b)^2]\cdot (\delta(F^{ 2}|_a) - \gamma)],
\end{align*}
where $B(F) =  \frac{4}{3}\E_X[(F - F^{ 3})F_i] + \exp(i)\gamma^{1/3}\E_X[F^{ 4} - F]$.
\end{lemma}

\begin{proof}
Let $\delta = \E_X[F(X)]$ and $\eta_i = \E_X[F_i^2(X)]$. Under the Booleanity axioms on $F$, we have proved an upper and lower bound on ${\E_X[F_i^4(X)]}$,
\begin{itemize}
    \item Upper Bound in Lemma~\ref{lem:upper-bound}:
\begin{equation}\label{eq:ub}
\E_X[F_i(X)^4] \leq \exp(i)\gamma \eta_i + \sum_{j = 0}^i c'_j \ell^{i}\E_{a \sim [n]^{j}}[\E_{b \sim [n]^{i-j}}[f_i(a,b)^2]\cdot (\delta(F^{2}|_a) - \gamma)],
\end{equation}
for some positive constants $c'_j \leq \exp(i)$.

\item Lower Bound in Lemma~\ref{lem:lower-bound}: For all $\eps \in \R$ we get,
\begin{equation}\label{eq:lb}
\E_X[F_i^4(X)] \geq 4\eps^3 \eta_i - 3\eps^4\delta + B_1(F),
\end{equation}
where $B_1(F) = 4\eps^3\E_X[(F^{ 3} - F)F_i]+3\eps^4 \E_X[F - F^{ 4}]$.
\end{itemize}

Combining~\eqref{eq:ub} and~\eqref{eq:lb}, 
setting $\eps = (\exp(i)\gamma)^{1/3}$ and dividing 
by $3\eps^3$ yields
\begin{align*}
\eta_i \leq \exp(i)\gamma^{1/3}\delta + \sum_{j = 0}^i \frac{c'_j}{3\gamma\exp(i)} \ell^{i}\E_{a \sim [n]^{j}}[\E_{b \sim [n]^{i-j}}[f_i(a,b)^2]\cdot (\delta(F^{ 2}|_a) - \gamma)] - \frac{B_1(F)}{3\gamma \exp(i)}.
\end{align*}
The proof is thus completed by setting $c_j = \frac{c'_j}{3\exp(i)}$ and 
$B(F) = \frac{4}{3}\E_X[(F - F^{ 3})F_i] + \exp(i)\gamma^{1/3}\E_X[F^{ 4} - F]$.
\end{proof}

\subsection{Expansion of Pseudorandom sets}
Lemma~\ref{lem:leveli-j} directly implies an expansion theorem for $C_{n,\ell,\alpha}$, as follows:
\begin{theorem}[Expansion Theorem for $C_{n,\ell,\alpha}$]\label{thm:structure-johnson-cayley}
For all $\alpha \in (0,1)$, all integers $\ell \geq 1/\alpha$ and $n \geq \ell$, the following holds: Let $C_{n,\ell,\alpha}$ be the Johnson-approximating Cayley graph. For every constant $\gamma \in (0,1)$ and positive integer $r \leq O(\ell)$, every permutation-invariant function $F: V(C) \rightarrow \R$ that is $(r,\gamma)$-pseudorandom:
\begin{enumerate}
\item $\cA_{inv}(F) \cup \{F(X) \in [0,1]\}_{X \in V(C)} \, \vdash_{O(1)}\, \\
\ip{F, LF} \geq \E[F](1 - (1-\alpha)^{r+1})(1- \gamma^{1/3}\exp(r)) - \sum_{j = 0}^r \frac{ c_j\ell^{j}}{\gamma}\E_{a \sim [n]^{j}}[q_a(F)({\delta(F^{ 2}|_a) - \gamma)}] + B(F).$
\item $\{F(X) \in [0,1]\} \,\,\, \vdash_{2} \,\,\, 0 \leq q_a(F) \leq 1,$
\end{enumerate}
where for all $j \leq r$, $c_j$'s are positive constants $\leq \exp(r)$, for all $a \in [n]^j, j \leq r$, $q_a(F)$ are degree $2$  polynomials and $B(F) = \frac{4}{3}\E_X[(F^{ 3} - F)\sum_{i = 0}^r F_i] + c\E_X[F - F^{ 4}]$, for a positive constant $c \leq \exp(r)\gamma^{1/3}$.
\end{theorem}

\begin{proof}
We first have that: $\ip{F, LF} = \E[F^{ 2}]  - \ip{F,CF}$, where $L$ and $C$ denote the normalized Laplacian and adjacency matrix of the graph $C_{k,\ell,\alpha}$ respectively. Let $1=\lambda_1 \geq \ldots \lambda_\ell \geq -1$ denote the eigenvalues of $C$. We will henceforth bound $\ip{F,CF}$:
\begin{align*}
\ip{F,CF} = \sum_{i \in [l]} \lambda_i \eta_i 
&\leq \sum_{i = 0}^r \lambda_i\eta_i + (\E[F^{ 2}] - \sum_{i=1}^r \eta_i)\lambda_{r+1} \\
&\leq (1 - \lambda_{r+1}) \sum_{i=0}^r \eta_i + \E[F]\lambda_{r+1},\\
\end{align*}
where in the last step we use that $F(X) \in [0,1]$. Let $\delta$ denote $\E[F]$. Using Lemma~\ref{lem:leveli-j} we get: 
\begin{align}
&\ip{F,CF} \notag\\
&\leq (1 - \lambda_{r+1})\sum_{i=0}^r \left(\exp(i)\gamma^{1/3}\delta + B_{i}(F)+ \sum_{j = 0}^i  \frac{c'_{ij}\ell^{i}}{\gamma}\E_{a \sim [n]^{j}}[\E_{b \sim [n]^{i-j}}[f_i(a,b)^2]\cdot (\delta(F^{2}|_a) - \gamma)] \right) + \delta\lambda_{r+1} \notag\\
&\leq (1 - \lambda_{r+1})(\exp(r)\gamma^{1/3}\delta) + (1-\lambda_{r+1})\sum_i B_i(F) \notag\\+ &\sum_{j = 0}^r  \E_{a \sim [n]^{j}}[ (\delta(F^{ 2}|_a) - \gamma)(\sum_{i = j}^r (1 - \lambda_{r+1})\frac{c'_{ij} \ell^{i}}{\gamma} \E_{b \sim [n]^{i-j}}[f_i(a,b)^2])] +  \delta\lambda_{r+1}, \notag
\end{align}
for some positive constants $c'_{ij} \leq \exp(i)$ and $B_i(F) =  \frac{4}{3}\E_X[(F - F^{ 3})F_i] + \exp(i)\gamma^{1/3}\E_X[F^{ 4} - F]$.

Let us now look at the multipliers of $\delta(F^{2}|_a) - \gamma$, for each $a \in [n]^j$ consider the polynomial: $q_a(F) = \frac{1}{c_j\ell^j}\sum_{i = j}^r (1 - \lambda_{r+1})c'_{ij} \ell^{i} \E_{b \sim [n]^{i-j}}[f_i(a,b)^2]$, for fixed constants $c_j \leq \exp(r)$ to be determined later. Using Lemma~\ref{lem:p.r.-weight} to bound $\E_{b \sim [n]^{i-j}}[f_i(a,b)^2]$ we get that:
\begin{align*}
\{F(X) \in [0,1]\} \vdash_2 ~~~~0 \leq q_a(F) \leq \frac{1}{c_j}\sum_{i = j}^r (1 - \lambda_{r+1})c'_{ij}\ell^{i-j} \frac{ \delta(F^{2}|_a)}{{\ell - j \choose i - j}} \leq 1,
\end{align*}
for $c_j = \sum_{i = j}^r (1 - \lambda_{r+1})c'_{ij}\frac{\ell^{i-j}}{{\ell - j \choose i-j}} \leq \exp(r)$. Plugging in the polynomials $q_a(F)$ and rearranging yields
\[\ip{F,LF} \geq \delta(1 - \lambda_{r+1})(1- \gamma^{1/3}\exp(r)) - \sum_{j = 0}^r  \frac{c_j\ell^{j}}{\gamma}\E_{a \sim [n]^{j}}[q_a(F)(\delta(F^{ 2}|_a) - \gamma)] + B(F),\]
for some positive constants $c_j \leq \exp(r)$ and $B(F) = \frac{4}{3}\E_X[(F^{ 3} - F)\sum_{i = 0}^r F_i] + c\E_X[F - F^{ 4}]$, for a positive constant $c \leq \exp(r)\gamma^{1/3}$. 

We will now plug in upper bounds for $\lambda_i$. From Lemma~\ref{lem:eigenval}, we have that $\lambda_i = \frac{\binom{\ell - |T|}{(1 - \alpha)\ell - |T|}}{\binom{\ell}{(1 - \alpha)\ell}}$ for $i \leq (1-\alpha)\ell$ and $0$ otherwise. One can check that ${\lambda_i \leq (1-\alpha)\lambda_{i-1}}$ for all $i$ between $1$ and $\ell$. Since $\lambda_0 = 1$, we get that $\lambda_i \leq (1-\alpha)^i$, thus completing the proof of the lemma.

\end{proof}

\end{document}